\documentclass[12pt]{article}
\usepackage[utf8]{inputenc}
\usepackage[sectionbib]{natbib}

\title{\vspace{-2cm}Quasi-maximum likelihood estimation for scalable ARMA models}
\author{\small Yuchang Lin, Wenyu Li and Qianqian Zhu\thanks{Address for correspondence: Qianqian Zhu, School of Statistics and Data Science, Shanghai University of Finance and Economics, Shanghai, China. Email: zhu.qianqian@mail.shufe.edu.cn}  \\[-0.2cm]
{\footnotesize\textit{The University of Hong Kong, Nankai University,}}\\[-0.1cm]
{\footnotesize\textit{and Shanghai University of Finance and Economics}}}

\usepackage[margin=1in]{geometry}
\usepackage{graphicx}
\usepackage{caption,subcaption}
\usepackage{mathrsfs}
\usepackage{multirow}
\usepackage{amsfonts}
\usepackage{comment}
\usepackage{amsthm}
\usepackage{xr}
\usepackage{xcolor}
\usepackage{mathtools}
\usepackage{algorithm}
\usepackage{algorithmic}
\usepackage{sectsty}
\usepackage[colorlinks]{hyperref}
\usepackage{amsmath}
\usepackage{amssymb}
\usepackage{stmaryrd}
\usepackage[toc]{appendix}
\usepackage[mathscr]{euscript}

\usepackage{chngcntr}
\usepackage{apptools}
\usepackage{booktabs}
\usepackage{lscape}
\usepackage{soul}
\usepackage{epstopdf}
\usepackage{array}
\usepackage{diagbox}
\usepackage{tabularx,longtable,adjustbox}
\usepackage{setspace}
\usepackage{enumerate}
\numberwithin{equation}{section}

\newcolumntype{L}[1]{>{\raggedright\let\newline\\\arraybackslash\hspace{0pt}}m{#1}}
\newcolumntype{C}[1]{>{\centering\let\newline\\\arraybackslash\hspace{0pt}}m{#1}}
\newcolumntype{R}[1]{>{\raggedleft\let\newline\\\arraybackslash\hspace{0pt}}m{#1}}

\newtheorem{assum}{Assumption}

\newtheorem{lemma}{Lemma}
\newtheorem{proposition}{Proposition}

\newtheorem{thm}{Theorem}

\newtheorem{remark}{Remark}

\hypersetup{
    colorlinks=true,
    linkcolor=blue,
    citecolor=blue,
    filecolor=magenta,
    urlcolor=cyan,
}

\DeclareMathOperator*{\argmin}{arg\,min}

\DeclareMathOperator*{\diag}{diag}

\DeclareMathOperator*{\var}{var}
\newcommand{\bm}[1]{\mbox{\boldmath{$#1$}}}
\newcommand{\RN}[1]{\uppercase\expandafter{\romannumeral#1}}

\allowdisplaybreaks[4]

\mathtoolsset{showonlyrefs}

\begin{document}

\setlength{\parindent}{16pt}

\maketitle

\begin{abstract}
The recently proposed scalable ARMA model preserves the parsimony of traditional VARMA models while achieving greater computational tractability. However, existing studies are limited to regularized least squares estimation (LSE) for high-dimensional settings, which is not only statistically less efficient but also requires the sub-Gaussian assumption for its theoretical guarantees. Moreover, it still lacks inference tool for real applications.
To fill this gap, we develop a quasi-maximum likelihood estimation (QMLE) framework for scalable ARMA models. Its asymptotic normality is established under a finite fourth order moment condition, and we formally prove its asymptotic efficiency gain over LSE. We also introduce an efficient block coordinate descent algorithm for computation and a consistent Bayesian information criterion for model selection. Simulation studies validate the finite-sample performance of our methodology, and an empirical application to six macroeconomic indicators demonstrates its practical utility.
\end{abstract}


\newpage
\vspace{-1cm}

\linespread{1.55}
\selectfont{}

\section{Introduction}\label{Introduction}

Multivariate time series analysis constitutes a fundamental methodology in fields such as economics, finance, and environmental studies \citep{luetkpohl2005}.
%{\color{blue}Serving as a barometer for economic health, macroeconomic indicators allow analysts to monitor performance and discern trends.} 
Economists rely on the tracking of key variables, such as industrial production, civilian unemployment rate, consumer price index, real wages, real personal consumption expenditures, and federal funds rate, to evaluate economic conditions and guide policy or investment decisions. 
Critically, these economic variables are usually cross-correlated, and understanding their interdependence will help reveal underlying linkage relationships and co-movements within the economic system. 
In response to this need, several multivariate models have been introduced to capture these dynamic interrelations and enhance forecast accuracy. Notable among these are the vector autoregressive (VAR), vector moving average (VMA), and vector autoregressive moving average (VARMA) \citep{sims1980,TiaoandBox1981, Tiao_Tasy1989, reinsel2003elements}. 
The VAR model is the most widely adopted owing to its theoretical simplicity, computational ease, and clear interpretation \citep{Tsay2014}. 
Despite its popularity, a practical limitation of the VAR model is that it often necessitates a high lag order to achieve a satisfactory fit. This results in a large number of coefficient matrices and thus a lack of parsimony \citep{CHAN2016}. 

Compared to VAR models, the VARMA model has a more parsimonious structure. 
Meanwhile, the invertible VARMA model is equivalent to a VAR($\infty$) form, which enables it to capture temporal dependence in a more flexible and comprehensive manner than finite-order VAR models. 
This flexibility contributes to superior forecasting performance of VARMA models over VAR models, as supported by much empirical evidence \citep{poskitt2003,athanasopoulos2008,athanasopoulos2012,kascha2012comparison,simionescu2013,dufour2022practical}. 
However, the practical application of VARMA models faces challenges due to issues of non-identifiability and computational complexity. 
Many studies have attempted to address these issues. 
For example, various sophisticated identification constraints have been proposed, including the final equations form \citep{zellner1974time, wallis1977multiple}, the Echelon form \citep{hannan1984, lutkepohl1996}, and scalar component models \citep{Tiao_Tasy1989}. 
To mitigate estimation difficulties, an iterative ordinary least squares method was introduced in \cite{dias2018estimation}. 
Moreover, \cite{CHAN2016} developed a Bayesian approach to ensure identification and parsimony in VARMA modeling, and \cite{Wilms2023} established a convex optimization-based framework for identification and estimation. 
Nevertheless, these approaches primarily address the identification problem by imposing complex identification conditions or tackle the computational intractability by optimization algorithms, without fundamentally addressing the inherent limitations rooted in the model's construction.

Recent work by \cite{zheng2025interpretable} and \cite{huang2025} proposed a scalable ARMA model formulated in a VAR($\infty$) form. This model retains the parsimony and rich dynamic patterns (e.g., exponential and damped sinusoidal decays) of traditional VARMA models, while successfully circumventing their inherent issues of non-identifiability and computational intractability. Its VAR($\infty$) representation also allows for direct inference of multivariate Granger causality through its coefficient matrices. Nonetheless, two critical gaps remain. First, existing work is confined to regularized least squares estimation (LSE) for high-dimensional data, supported by non-asymptotic error bounds. These bounds, however, do not facilitate standard statistical inference such as constructing confidence intervals or significance tests. Consequently, valid inferential tools for scalable ARMA models remain unknown. 
Moreover, given the prominent role of factor models in dimension reduction, integrating inference for scalable ARMA with factor structures could offer powerful tools for high-dimensional analysis. 
Second, LSE itself is statistically less efficient, as it ignores the contemporaneous correlations and heterogeneity among innovations. Moreover, non-asymptotic theories require the restrictive sub-Gaussian assumption. This limits its applicability to heavier-tailed economic and financial time series. To address these gaps, this paper develops a more efficient and robust inference framework for scalable ARMA models under a finite fourth order moment condition. The proposed tools enable a deeper understanding of the interdependence and dynamic patterns in multiple series, as illustrated in the application (Section \ref{Sec-Realdata}).

Our main contributions can be summarized as follows: 
(i) we develop the quasi-maximum likelihood estimation (QMLE) for the scalable ARMA model, establishing its asymptotic properties under relaxed moment conditions. The asymptotic efficiency of the QMLE over the least squares estimation (LSE) is also rigorously confirmed; 
(ii) we propose a block coordinate descent (BCD) algorithm for the efficient computation of the QMLEs; 
(iii) we introduce an asymptotically consistent Bayesian information criterion (BIC) for model order selection. 
Moreover, simulation studies validate these theoretical results.
Furthermore, the practical value of our proposed framework is highlighted by an empirical study on exploring the Granger-causal relationship among six macroeconomic variables, as well as the forecasting comparison with VAR and VARMA models. 

The remainder of this paper is organized as follows. Section \ref{Sec-SARMA} introduces the scalable ARMA model. Section \ref{Sec-Methodology} proposes the QMLE method and BCD algorithm for estimation, as well as the BIC method for order selection. Section \ref{Sec-AsymptoticResults} establishes the asymptotic properties for the proposed QMLE and BIC method, and the asymptotic efficiency of the QMLE and LSE is also compared. 
Simulation studies and an empirical example are provided in Sections \ref{Sec-Simulation} and \ref{Sec-Realdata}, respectively.
The conclusion and discussion appear in Section \ref{Sec-Conclusion}. Technical proofs of all theorems and propositions are provided in the Appendix.
Throughout the paper, we denote vectors by small boldface letters and matrices by capital letters.
Denote $\|\bm{x}\|_p$ as the $\ell_p$ norm of a vector $\bm{x}$, that is $\|\bm{x}\|_p=(\sum_{i}|x_i|^p)^{1/p}$. 
For a matrix $A$, $\|A\|_{\operatorname{op}}$, $\|A\|$, $A^{\prime}$ and $\text{vec}(A)$ are the operator norm (i.e., $\|A\|_{\operatorname{op}}=\sup_{\|\bm x\|_2 =1}\|A \bm x\|$), matrix norm satisfying sub-multiplicative property, transpose and vectorization of $A$, respectively. 
For a square matrix $B$, $|B|$, $\text{tr}(B)$ and $\text{vech}({B})$ are the  determinant, trace, and the lower triangular-stacking vector of $B$, respectively. 
And $B \succeq 0$ means that $B$ is semi-positive definite. 
For square matrices $B_j$'s, let $\diag\{B_{1}, \ldots, B_{m}\}$ be the block diagonal matrix whose main diagonal consists of $B_{1}, \ldots, B_{m}$.
Moreover, $\bm 1_N$ is an $N \times 1$ vector of ones, $0_{N\times T}$ is an $N\times T$ matrix of zeroes, and $I_N$ is an $N\times N$ identity matrix. 
Additionally, $O(1)$ denotes a bounded series of nonstochastic variables, $\rightarrow_{p}$ denotes the convergence in probability, and $\rightarrow_{d}$ denotes the convergence in distribution. 
The data that support the findings of this study in Section \ref{Sec-Realdata} are openly available in FRED-MD repository at \url{stlouisfed.org/research/economists/mccracken/fred-databases}, and computer programs for the analysis are available at \url{github.com/LinyuchangSufe/SARMA}.

\section{Scalable ARMA models}\label{Sec-SARMA}

% For an $N$-dimensional process $\{\bm{y}_t\}$, consider the scalable ARMA model \citep{zheng2025interpretable,huang2025} in the following parametric VAR($\infty$) form: 
Following \cite{zheng2025interpretable} and \cite{huang2025}, the scalable ARMA model for an $N$-dimensional process ${\bm{y}_t}$ is given by the parametric VAR($\infty$) form: 
\begin{align}\label{model_VARinf}
	\bm y_t = \sum_{h=1}^{\infty} A_h(\bm \omega, \bm g) \bm y_{t-h} +\bm \varepsilon_t
	=\sum_{k=1}^{d}G_k \sum_{h=1}^{\infty} \ell_{hk}(\bm \omega) \bm y_{t-h}+\bm\varepsilon_t,
\end{align}
where $\{\bm{\varepsilon}_t\}$ is a sequence of independent and identically distributed ($i.i.d.$) random vectors with zero mean and covariance matrix $\Sigma=(\sigma_{ij})_{1\le i,j \le N}$, 
$d=p+r+2s$, 
$\bm{\omega}= (\bm{\lambda}^\prime, \bm{\eta}^\prime)^\prime$, 
$\bm\lambda = (\lambda_1,\ldots,\lambda_r)^\prime$ with $|\lambda_i| \in (0,1)$ for $1\le i\le r$, 
$\bm\eta = (\bm{\eta}_1^\prime, \ldots ,\bm{\eta}_s^\prime)^\prime$ with $\bm{\eta}_j = (\gamma_j, \varphi_j) \in (0,1) \times (0,\pi)$ for $1 \le j \le s$, 
$\bm{g}=\text{vec}(G_1,\ldots,G_d)$ with $G_k \in \mathbb{R}^{N\times N}$ for $1 \le k \le d$, 
% $\ell_{hk}(\bm{\omega})$ is $1$ if $1 \leq k \leq p$, $\lambda_i^{h-p}$ if $k=p+i$ with $1 \leq i \leq r$, $\gamma_j^{h-p} cos[(h-p) \varphi_j]$ if $k=p+r+2j-1$ with $1 \leq j \leq s$, and $\gamma_j^{h-p} sin[(h-p) \varphi_j]$ if $k=p+r+2j$ with $1 \leq j \leq s$
and the entries $\ell_{hk}(\bm{\omega})$ for $h\ge 1$ and $1 \leq k \leq d$ are defined as follows: 
$$
\ell_{hk}(\bm{\omega}) = 
\begin{cases}
1, & \text{if} \;\; 1 \leq h = k \leq p, \\
\lambda_i^{h-p}, & \text{if} \;\; h>p, \;\; k=p+i \;\; \text{with} \;\; 1 \leq i \leq r,  \\
\gamma_j^{h-p} cos[(h-p) \varphi_j], & \text{if} \;\; h>p, \;\; k=p+r+2j-1 \;\; \text{with} \;\; 1 \leq j \leq s,  \\
\gamma_j^{h-p} sin[(h-p) \varphi_j], & \text{if} \;\; h>p, \;\; k=p+r+2j \;\; \text{with} \;\; 1 \leq j \leq s, \\
0, & \text{otherwise}. 
\end{cases}
$$
% $\ell_{hk}(\bm{\omega})$ is the $(h,k)$th entry of the $\infty \times d$ matrix
% \begin{align}\label{L_omega}
% 	L(\bm{\omega})
% 	=\{\ell_{hk}(\bm{\omega})\}_{h\ge 1,1\le k \le d}
% 	=\left( 
% 	\begin{matrix}
% 		I_p & 0_{p\times(d-p)}\\
% 		0_{\infty\times p} & \bm{\ell}^{I}(\lambda_1)\cdots \bm{\ell}^I(\lambda_r) ~~ \bm{\ell}^{II}(\bm{\eta}_1)\cdots \bm{\ell}^{II}(\bm{\eta}_s) 
% 	\end{matrix}
% 	\right)\in\mathbb{R}^{\infty \times d}
% \end{align}
% with $\bm{\ell}^I(\lambda_i)=(\lambda_i,\lambda_i^2,\ldots)^\prime \in \mathbb{R}^{\infty}$ and 
% \begin{align*}
% 	\bm{\ell}^{II}(\bm{\eta}_j)
% 	=\left(
% 	\begin{matrix}
% 		\gamma_j \cos(\varphi_j) &\gamma_j^2 \cos(2\varphi_j) & \cdots \\
% 		\gamma_j \sin(\varphi_j) &\gamma_j^2 \sin(2\varphi_j) & \cdots   
% 	\end{matrix}
% 	\right)^\prime \in \mathbb{R}^{\infty\times 2},
% \end{align*} 
Particularly, $(p,r,s)$ is the order of model \eqref{model_VARinf}. 
To ensure the uniqueness of model \eqref{model_VARinf}, we assume that $\lambda_1> \cdots > \lambda_r$ and $\gamma_1 > \cdots > \gamma_s$. 
Model \eqref{model_VARinf} provides a time series model comparable to the VARMA model, while successfully circumventing the latter's inherent issues of non-identifiability and computational intractability. 
% It is originally derived from an invertible VARMA($1,1$) model $\bm y_t = \Phi \bm y_{t-1} + \bm\varepsilon_t-\Theta\bm\varepsilon_{t-1}$ which has an equivalent VAR($\infty$) form $\bm y_t =\sum_{h=1}^{\infty}\Theta^{h-1}(\Phi-\Theta) \bm y_{t-h} +\bm \varepsilon_t$, by block-diagonalizing $\Theta$ via the Jordan decomposition but relaxing the restrictions on the resulting coefficient matrices $G_k$'s \citep{zheng2025interpretable}. Therefore, model \eqref{model_VARinf} avoids the problems of non-identifiability and computational intractability due to the high-order matrix polynomials encountered by the VARMA model.

For $1 \leq h \leq p$, $\ell_{hk}(\bm{\omega}) = 1$ if $k=h$ and 0 otherwise, which implies that $A_h = G_h$. 
For $h > p$, the coefficient $\ell_{hk}(\bm{\omega})$ follows two distinct decay patterns as $h$ increases: exponential decay and exponentially damped cosine (or sine) waves. In this case, $A_h$ can be decomposed as $A_h = G_h^{\RN{1}} + G_h^{\RN{2}}$, where $G_h^{\RN{1}} = \sum_{i=1}^r \lambda_i^{h-p} G_{p+i}$ and $G_h^{\RN{2}} = \sum_{j=1}^s \gamma_j^{h-p} \{\cos[(h-p)\varphi_j] G_{p+r+2j-1} + \sin[(h-p)\varphi_j] G_{p+r+2j}\}$ capture these two decay patterns, respectively. 
Consequently, model \eqref{model_VARinf} can be rewritten as 
\begin{align}\label{model_VARinf-2}
	\bm{y}_t=\sum_{h=1}^{p}G_h \bm y_{t-h}+\sum_{h=p+1}^{\infty} (G_h^{\RN{1}} + G_h^{\RN{2}}) \bm y_{t-h}+\bm{\varepsilon}_t. 
\end{align} 
Model \eqref{model_VARinf-2} explicitly separates the short-run and long-run autoregressive effects. 
Specifically, $\sum_{h=1}^{p}G_h \bm y_{t-h}$ captures the short-run effects, analogous to a standard VAR($p$) model. 
In contrast, $\sum_{h=p+1}^{\infty} (G_h^{\RN{1}} + G_h^{\RN{2}}) \bm y_{t-h}$ captures the long-run effects. 
The term $G_h^{\RN{1}}$ embodies exponential decay via $\lambda_i^{h-p}$, while the term $G_h^{\RN{2}}$ generates exponentially damped oscillatory patterns via $\gamma_j^{h-p}\cos[(h-p)\varphi_j]$ and $\gamma_j^{h-p}\sin[(h-p)\varphi_j]$. Both ensure that long-run contributions vanish fast as the time lag $h$ increases. 
In addition, we refer to $\lambda_i$'s and $\gamma_j$'s the decay rate coefficients, $\varphi_j$'s the periodic coefficients. 
In the special case where $r = s = 0$, we have $G_h^{\RN{1}} = G_h^{\RN{2}} = 0_{N\times N}$ for all $h > p$, and the model \eqref{model_VARinf-2} will reduce to a VAR($p$) process, retaining only the short-run component.

On the other hand, model \eqref{model_VARinf} can also be rewritten as follows
\begin{align}\label{VARprs}
	\bm{y}_t=\sum_{l=1}^{p}G_l \bm y_{t-l}+ 
	\sum_{i=1}^{r}G_{p+i}\bm{f}^{I}(\bm{x}_t;\lambda_i)+
	\sum_{j=1}^{s}\sum_{m=1}^{2}G_{p+r+2(j-1)+m}\bm{f}^{II,m}(\bm{x}_t;\bm{\eta}_j)+\bm{\varepsilon}_t,
\end{align} 
where $\bm{x}_t=(\bm y_{t-1}^{\prime},\bm y_{t-2}^{\prime},\ldots)^{\prime}$, 
$\bm{f}^{I}(\bm{x}_t;\lambda_i)=\sum_{h=p+1}^{\infty}\lambda_i^{h-p}\bm{y}_{t-h}$ for $1\le i \le r$, $\bm{f}^{II,1}(\bm{x}_t;\bm{\eta}_j)=\sum_{h=p+1}^{\infty}\gamma_j^{h-p}\cos[(h-p)\varphi_j]\bm{y}_{t-h}$ and $\bm{f}^{II,2}(\bm{x}_t;\bm{\eta}_j)=\sum_{h=p+1}^{\infty}\gamma_j^{h-p}\sin[(h-p)\varphi_j]\bm{y}_{t-h}$ for $1\le j \le s$. 
It can be seen that the orders $p$ and $(r,s)$ are analogous to the AR and MA orders of the VARMA model \citep{zheng2025interpretable}. 
Particularly, model \eqref{VARprs} will reduce to a VAR($p$) model when $r=s=0$, and it can be treated as a VMA model when $p=0$. 
Consequently, we refer to $G_1,\ldots,G_p$ the AR coefficient matrices, and $G_{p+1},\ldots,G_d$ the MA coefficient matrices of model \eqref{model_VARinf} with the order $(p,r,s)$.

Additionally, we can use model \eqref{model_VARinf} to infer the Granger causal (GC) relations between any pair of component series in $\bm y_t$ \citep{zheng2025interpretable}. Since $G_k$ captures cross-sectional information associated with a particular sequence $\{\ell_{hk}(\bm \omega)\}_{h=1}^\infty$, it enables us to identify the decay pattern of GC relations. Denote $g_{i,j,k}$ as the $(i,j)$th entry of $G_k$, and $y_{j,t}$ as the $j$th entry of $\bm y_t$. Consider the influence of $\{y_{j,t}\}$ on $\{y_{1,t}\}$ for illustration. The equation for $y_{1,t}$ in model \eqref{VARprs} includes two conditional mean terms: the first term involves the sum of $g_{1,j,k} y_{j,t-k}$ for $1\le k\leq p$, while the second term involves a weighted mixture of $r$ distinct exponential decay rates and $s$ distinct pairs of damped cosine and sine waves to capture the influence beyond lag $p$. Hence, if the GC relation exists, the lagged influence of $\{y_{j,t}\}$ on $\{y_{1,t}\}$ can be classified into three scenarios: (1) \textit{short-term only}, $ g_{1,j,k} \neq 0 $ for some $ 1 \le k \le p $, while $ g_{1,j,p+1} = \cdots = g_{1,j,d} = 0 $; (2) \textit{long-term only}, $ g_{1,j,1} = \cdots = g_{1,j,p} = 0 $, while $ g_{1,j,k} \neq 0 $ for some $ p+1 \le k \le d $; and (3) \textit{both short-term and long-term influences}, $ g_{1,j,k} \neq 0 $ for some $ 1 \le k \le p $ and some $ p+1 \le k \le d $. 

This paper focuses on the case of fixed $N$, thus requiring no sparsity or low-rank conditions on $G_k$ for $1 \le k \le d$. This differs from the high-dimensional frameworks studied by \cite{zheng2025interpretable} and \cite{huang2025}.
%As for the VARMA model, in practice it suffices to use small orders $(p,r,s)$ for model \eqref{model_VARinf}; see Section \ref{Sec-Realdata} for empirical evidence. 
Below, we propose an estimation method for model \eqref{model_VARinf} and an information criterion to select its order $(p,r,s)$. 

%Since this paper focuses on the situation with $N$ fixed, different from \cite{zheng2025interpretable} and \cite{huang2025} where the large dimension $N$ is considered, there is no sparse or low-rank assumption restricted on $G_k$ for $1 \le k \le d$. 
%In the following, we propose an estimation method for the unknown parameters $\bm \omega$ and $\bm g$, and introduce an information criterion to select the order $(p,r,s)$ for specification. 

\section{Methodology}\label{Sec-Methodology}

%This section introduces the Gaussian quasi-maximum likelihood estimators (QMLEs) for model \eqref{model_VARinf}. To efficiently compute the QMLEs, a block coordinate descent (BCD) algorithm is proposed. Moreover, the order selection using Bayesian information criterion (BIC) is investigated based on the QMLEs. 

\subsection{Quasi-maximum likelihood estimation}\label{subSec-MLE}

% Note that the innovation $\bm \varepsilon_{0t}$ of model \eqref{VARprs} has a general covariance matrix $\Sigma_{0}$. Thus the LSE $\widehat{\bm{\alpha}}_{LS}$ typically is not the most efficient estimator because it neglects the concurrent correlations and heterogeneity among $\bm \varepsilon_{0t}$.  
% Next, we will discuss the quasi-maximum likelihood estimation which can be asymptotically more efficient than the LSE.

Suppose that the time series $\{\bm y_t\}_{t=1}^T$ is generated by model \eqref{model_VARinf} with order $(p,r,s)$. 
Recall that $\bm{\omega}= (\bm{\lambda}^\prime, \bm{\eta}^\prime)^\prime$, $\bm{g}=\text{vec}(G_1,\ldots,G_d)$, and $\Sigma$ is the covariance matrix of $\{\bm{\varepsilon}_t\}$. 
Denote $\bm{\alpha}=(\bm{\omega}^\prime, \bm{g}^\prime)^\prime$ and $\bm{\sigma}= \text{vech}(\Sigma)$. Let $\bm{\theta}=(\bm{\alpha}^\prime, \bm{\sigma}^\prime)^\prime$ be the parameter vector of model \eqref{model_VARinf} and $\Theta$ be the parameter space, where $\Theta$ is a compact subspace of $\{(-1,0)\cup (0,1)\}^r \times \{(0,1)\times (0,\pi )\}^s \times \mathbb{R}^{N^2d} \times \mathbb{R}^{N(N+1)/2}$. 
Conditional on $\{\bm y_1,\ldots,\bm y_T\}$, the negative Gaussian quasi-log-likelihood function (ignoring a constant) is 
\begin{align*}
	\mathcal{L}(\bm{\theta})=T^{-1} \sum_{t=1}^{T} l_t(\bm{\theta}) \quad \text{and} \quad l_t(\bm{\theta}) =\frac{1}{2}\ln |\Sigma(\bm{\sigma})|+\frac{1}{2} \bm{\varepsilon}_t^\prime (\bm{\alpha})\Sigma^{-1}(\bm{\sigma})\bm{\varepsilon}_t(\bm{\alpha}),
\end{align*}
where $\bm{\varepsilon}_t(\bm{\alpha}) = \bm y_t - \sum_{h=1}^{\infty}A_{h}(\bm{\omega}, \bm{g})\bm y_{t-h}$ is the error function, and $\Sigma(\bm\sigma)$ denotes the reconstruction of the symmetric matrix $\Sigma$ from its half-vectorized form $\bm\sigma$. 
Note that $\bm{\varepsilon}_t(\bm{\alpha})$, $l_t(\bm{\theta})$ and $\mathcal{L}(\bm{\theta})$ depend on observations in the infinite past. We simply set $\bm{y}_s=0_{N\times 1}$ for $s\le 0$ as the initial values for $\{\bm{y}_s,s\le 0\}$ in this paper, and denote the resulting functions of $\bm{\varepsilon}_t(\bm{\alpha})$, $l_t(\bm{\theta})$ and $\mathcal{L}(\bm{\theta})$ as $\widetilde{\bm{\varepsilon}}_t(\bm{\alpha})$, $\widetilde{l}_t(\bm{\theta})$ and $\widetilde{\mathcal{L}}(\bm{\theta})$, respectively.
Then the Gaussian quasi-maximum likelihood estimator (QMLE) can be defined as 
\begin{align}\label{QMLE}
	\widehat{\bm \theta} = (\widehat{\bm \alpha}^\prime , \widehat{\bm \sigma}^\prime)^\prime= \argmin_{\bm{\theta}\in \Theta} \widetilde{\mathcal{L}}(\bm{\theta}).
\end{align}
Because we do not assume that $\bm\varepsilon_t$ is Gaussian, $\widehat{\bm \theta}$ is called the QMLE.  
Our theoretical analysis will prove that the effect of the initial values on the estimation is asymptotically negligible.
Moreover, the consistency and asymptotic normality of the QMLE $\widehat{\bm \theta}$ will be established under regularity conditions in Section \ref{Sec-AsymptoticResults-QMLE}. 

%Note that calculating the above likelihood function $\widetilde{\mathcal L}(\bm\theta)$ requires $O(N^2T+N^3)$ operations for scalable ARMA models. By contrast, $O(TN^3)$ operations are needed to  
%Owing to the design of scalable ARMA model, calculating the likelihood function $\widetilde{\mathcal L}(\bm\theta)$ requires $O(N^2T+N^3)$ operations, which is computationally more efficient than that of standard VARMA models

%The QMLE of scalable ARMA model can be substantially faster than the Kalman-filter-based QMLE commonly used for standard VARMA models. 
%For fixed scalable ARMA orders, a single evaluation of $\widetilde{\mathcal L}(\bm\theta)$ requires $O(N^2T+N^3)$ operations.
%By contrast, a Kalman-filter likelihood evaluation for a dense state-space representation of a VARMA model typically costs $O(TN^3)$, up to constants depending on the state dimension.  
%Although the total runtime of an estimation procedure also depends on the number of iterations and implementation details, this comparison suggests that the proposed model is computationally more favorable as $N$ increases. 
%This advantage is further supported by the simulation results in Section \ref{sim-TimeCost}.

\subsection{Block coordinate descent algorithm}\label{Sec-BCD_QMLE}

The objective function $\widetilde{\mathcal{L}}(\bm{\theta})$ in \eqref{QMLE} is non-convex and highly nonlinear with respect to the argument $\bm{\theta} =(\bm{\omega}^{\prime},\bm g^\prime,\bm{\sigma}^{\prime})^\prime$, and thus the QMLE $\widehat{\bm \theta}$ has no explicit form and needs to be abtained by optimization algorithm. 
While $\widetilde{\mathcal L}(\bm\theta)$ is not jointly convex, the parameterization in model \eqref{VARprs} allows a convenient conditional block structure: (i) for fixed $(\bm g, \bm\sigma)$, it is smooth in $\bm{\omega}$; (ii) for fixed $(\bm{\omega}, \bm\sigma)$, it is quadratic in $\bm g$; and (iii) for fixed $(\bm{\omega},\bm g)$, it has a closed-form solution for $\bm\sigma$; see details in \eqref{BCD_MLE_lambda}--\eqref{BCD_MLE_sigma}. 
This motivates us to employ a BCD algorithm to solve the QMLE $\widehat{\bm \theta}$. 

Recall that $\bm{\theta} =(\bm{\alpha}^{\prime},\bm{\sigma}^{\prime})^\prime =(\bm{\omega}^{\prime},\bm g^\prime,\bm{\sigma}^{\prime})^\prime$, and rewrite $\widetilde{\mathcal{L}}(\bm\theta)=\widetilde{\mathcal{L}}(\bm{\omega},\bm g, \bm \sigma)$. The BCD algorithm starts from an initial point $\bm{\theta}^{(0)}=(\bm{\alpha}^{(0)\prime},\bm{\sigma}^{(0)\prime})^{\prime}=(\bm{\omega}^{(0)\prime},\bm g^{(0)\prime},\bm{\sigma}^{(0)\prime})^{\prime}$, and generates a sequence of vectors $\{\bm{\theta}^{(k)}\}_{k=0}^\infty$. % with $\bm{\theta}^{(k)}=(\bm{\alpha}^{(k)\prime},\bm{\sigma}^{(k)\prime})^{\prime}=(\bm{\omega}^{(k)\prime},\bm g^{(k)\prime},\bm{\sigma}^{(k)\prime})^{\prime}$. 
Denote $\widetilde{\bm{x}}_t=(\bm{y}_{t-1}^{\prime},\ldots,\bm{y}_{1}^{\prime}, \bm{0}^{\prime})^{\prime}$. 
Define the $N \times N^2d$ matrix $ Z_t^{(k+1)} = \{\bm y_{t-1}^\prime ,\ldots , \bm y_{t-p}^\prime,
 \bm{f}^{I}(\widetilde{\bm{x}}_t;\lambda_1^{(k+1)})^\prime ,\ldots, \bm{f}^{I}(\widetilde{\bm{x}}_t;\lambda_r^{(k+1)})^\prime, 
\bm{f}^{II,1}(\widetilde{\bm{x}}_t;\bm \eta_1^{(k+1)})^\prime, \\
\ldots,
\bm{f}^{II,1}(\widetilde{\bm{x}}_t;\bm \eta_s^{(k+1)})^\prime,
\bm{f}^{II,2}(\widetilde{\bm{x}}_t;\bm \eta_1^{(k+1)})^\prime, \ldots, \bm{f}^{II,2}(\widetilde{\bm{x}}_t;\bm \eta_s^{(k+1)})^\prime
\}\otimes I_N $. 
Then given $\bm \omega^{(k+1)}$ and initial values $\bm{y}_s=0_{N\times 1}$ for $s\le 0$, we can rewrite model \eqref{VARprs} as 
$$\bm y_t =  Z_t^{(k+1)} \bm g+\bm{\varepsilon}_t.$$ 
For each outer iteration from $\bm{\theta}^{(k)}$ to $\bm{\theta}^{(k+1)}$, we update $\bm{\omega}, \bm g$ and $\bm{\sigma}$ sequentially: 
%Specifically, 
\begin{itemize}
	\item[(i)] \textbf{Update $\bm{\omega}$.} 
	Given $\lambda_1^{(k+1)},\ldots,\lambda_{i-1}^{(k+1)},\lambda_{i+1}^{(k)},\ldots,\lambda_{r}^{(k)},\bm \eta^{(k)}$, $\bm{g}^{(k)}$ and $\bm \sigma^{(k)}$, we update
	$\lambda_{i}^{(k+1)}$ by 
	\begin{align}\label{BCD_MLE_lambda}
		\lambda_i^{(k+1)}=\argmin_{\lambda_i} \widetilde{\mathcal{L}}(\lambda_1^{(k+1)},\ldots,\lambda_{i-1}^{(k+1)},\lambda_i,\lambda_{i+1}^{(k)},\ldots,\lambda_{r}^{(k)},\bm \eta^{(k)}, \bm g^{(k)},\bm \sigma^{(k)}),
	\end{align}
	and given $\bm \lambda^{(k+1)},\bm \eta_{1}^{(k+1)},\ldots,\bm \eta_{j-1}^{(k+1)},\bm \eta_{j+1}^{(k)},\ldots,\bm \eta_{s}^{(k)}, \bm \sigma^{(k)}$, we update $\bm \eta_{j}^{(k+1)}$ by 
	\begin{align}\label{BCD_MLE_eta}
		\bm \eta_j^{(k+1)}=\argmin_{\bm \eta_j} \widetilde{\mathcal{L}}(\bm \lambda^{(k+1)},\bm \eta_{1}^{(k+1)},\ldots,\bm \eta_{j-1}^{(k+1)},\bm \eta_j,\bm \eta_{j+1}^{(k)},\ldots,\bm \eta_{s}^{(k)},\bm g^{(k)},\bm \sigma^{(k)}),
	\end{align}
	where \eqref{BCD_MLE_lambda} and \eqref{BCD_MLE_eta} are only one- or two-dimensional, and thus can be solved using the NR algorithm. 
	\item[(ii)] \textbf{Update $\bm g$.} 
	Given $\bm \omega^{(k+1)}$ and $\bm\sigma^{(k)}$, $\bm g$ can be updated by the following closed-form: 
	\begin{align}\label{BCD_MLE_g}
		\bm g^{(k+1)}=
		\left(\sum_{t=1}^T Z_t^{(k+1)\prime}(\Sigma^{(k)})^{-1}Z_t^{(k+1)}\right)^{-1}
		\left(\sum_{t=1}^T Z_t^{(k+1)\prime}(\Sigma^{(k)})^{-1}\bm y_t\right),
	\end{align}
	where $\bm{\sigma}^{(k)}= \text{vech}(\Sigma^{(k)})$. 
	\item[(iii)] \textbf{Update $\bm{\sigma}$.} 
	Given $\bm{\omega}^{(k+1)}$ and $\bm{g}^{(k+1)}$, $\bm \sigma^{(k+1)}$ is updated by  
	\begin{align}\label{BCD_MLE_sigma}
		\bm \sigma^{(k+1)} = \text{vech}\left(\frac{1}{T} \sum_{t=1}^{T} \widetilde{\bm \varepsilon}_t(\bm \alpha^{(k+1)}) \widetilde{\bm \varepsilon}_t^\prime (\bm \alpha^{(k+1)})\right).
	\end{align}
\end{itemize}
Finally, we have $\bm{\theta}^{(k+1)}=(\bm{\omega}^{(k+1)\prime},\bm{g}^{(k+1)\prime},\bm \sigma^{(k+1)\prime})^{\prime}$. The above optimization procedure is repeated until convergence; see Algorithm \ref{alg2}.  

In practice, we usually specify a maximum number of iteration denoted by $k_{\max}$ for $k$, and let $\max_{1\leq j\leq n}|(\theta^{(k+1)}_j-\theta^{(k)}_j)/\theta^{(k)}_j| \leq \delta$ with $\delta>0$ being the stopping criterion, where $\theta^{(k)}_j$ is the $j$th element of $\bm{\theta}^{(k)}$. 
As a result, the QMLE $\widehat{\bm\theta}$ is set to $\bm{\theta}^{(k^{\star}+1)}$, where $k^{\star}=k_{\max}$ or $\bm{\theta}^{(k^{\star}+1)}$ satisfies the stopping criterion. 
Additionally, for the initial point $\bm{\theta}^{(0)}=(\bm{\omega}^{(0)\prime},\bm g^{(0)\prime},\bm{\sigma}^{(0)\prime})^{\prime}$, note that $\bm g^{(0)}$ and $\bm{\sigma}^{(0)}$ can be obtained given $\bm{\omega}^{(0)}$ by \eqref{BCD_MLE_g} and \eqref{BCD_MLE_sigma}, respectively. We perform a dense grid search over the space of $\bm{\omega}$, and select the grid point $\bm{\omega}$ that minimizes the loss function of QMLE as the initial point $\bm{\omega}^{(0)}$. 
% For simulation study and real data analysis in Sections \ref{Sec-Simulation}--\ref{Sec-Realdata}, we consider $k_{\max}=50$ and $\delta=10^{-3}$, and use an evenly spaced grid for initialization with the step size being {\color{red}0.1} for $\lambda_i$'s (or $\gamma_j$'s) and $0.1 \pi$ for $\varphi_j$'s. 

\begin{algorithm}[h!]
	\textsl{}\setstretch{1.8}
	\renewcommand{\algorithmicrequire}{\textbf{Input:}}
	\renewcommand{\algorithmicensure}{\textbf{Output:}}
	\caption{BCD algorithm for QMLE}
	\label{alg2}
	\begin{algorithmic}[1]
		\REQUIRE model order $(p,r,s)$, initialization $\bm{\omega}^{(0)}, \bm{g}^{(0)}$ and $\bm \sigma^{(0)}$.
		\STATE \textbf{repeat} $k=0,1,2,\ldots$
		\STATE \quad \textbf{for} $i=1,\ldots,r$
		\STATE \quad \quad $\lambda_i^{(k+1)}\leftarrow$ \eqref{BCD_MLE_lambda}
		\STATE \quad \textbf{for} $j=1,\ldots,s$
		\STATE \quad \quad $\bm{\eta}_j^{(k+1)}\leftarrow $ \eqref{BCD_MLE_eta}
		\STATE \quad $\bm{g}^{(k+1)} \leftarrow $ \eqref{BCD_MLE_g}
		\STATE \quad $\bm{\sigma}^{(k+1)} \leftarrow $ \eqref{BCD_MLE_sigma}
		\STATE  Until $k=k_{\max}$ or $\bm{\theta}^{(k+1)}$ satisfies the stopping criterion.
		\ENSURE $\bm{\theta}^{(k+1)}=(\bm{\omega}^{(k+1)\prime},\bm{g}^{(k+1)\prime},\bm \sigma^{(k+1)\prime})^{\prime}$ 
	\end{algorithmic}  
\end{algorithm}

\subsection{Model selection} \label{Sec-BIC}

Section \ref{subSec-MLE} introduces the QMLE for model \eqref{model_VARinf} under a correctly specified order $(p,r,s)$. 
However, the true order is unknown in practice.
Now we consider the selection of order $(p,r,s)$ for model \eqref{model_VARinf}. 

In the following, we use $\bm{\theta}_{\iota}$ to emphasize their dependence on the order $\iota = (p,r,s)$, 
and assume the parameter space $\Theta_\iota$ is a compact subset of $\Omega_\iota \times \mathbb{R}^{N^2d(\iota)} \times \mathbb{R}^{N(N+1)/2}$ with $\Omega_\iota=\{(-1,0)\cup(0,1)\}^r \times \{(0,1)\times (0,\pi )\}^s$ and $d(\iota) = p+r+2s$.
Let $p_{\rm max}, r_{\rm max}$ and $s_{\rm max}$ be the predetermined positive integers, and $\mathscr{M} = \{\iota \mid 0\le p \le p_{\rm max}, 0\le r\le r_{\rm max} ,0 \le s\le s_{\rm max}\}$. In practice, the maximum lags $p_{\rm max}, r_{\rm max}$ and $s_{\rm max}$ are typically specified as small integers (e.g., two or three). This is because, like traditional VARMA models, model \eqref{model_VARinf} is inherently parsimonious. 

We introduce the following BIC to select the order of model \eqref{model_VARinf}:
\begin{align}\label{BIC}
	&\text{BIC}(\iota) = 2T\widetilde{\mathcal{L}}(\widehat{\bm{\theta}}_\iota) +n(\iota)\ln T,
\end{align}
where $\widehat{\bm{\theta}}_\iota$ is the QMLE with the order set to $\iota$, $\widetilde{\mathcal{L}}(\widehat{\bm{\theta}}_\iota)$ is the negative Gaussian quasi-log-likelihood evaluated at $\widehat{\bm{\theta}}_\iota$, and $n(\iota) = r+2s+N^2d(\iota) + N(N+1)/2$ is the dimension of $\widehat{\bm{\theta}}_\iota$; see also \cite{schwarz1978estimating}.
% Here we use $n_{\alpha}(\iota)$ for both the QMLE and LSE, because the order $\iota$ is only related to $\bm \alpha_{\iota}$.
Let $\widehat{\iota} =\argmin_{\iota \in \mathscr{M}}\text{BIC}(\iota)$ be the selected order for model \eqref{model_VARinf}. We prove the selection consistency of BIC under regularity conditions in Section \ref{Sec-AsymptoticResults-BIC}.
Simulation results in Section \ref{Sec-Simulation} indicate that BIC performs well in finite samples.
%The selection consistency of BIC at \eqref{BIC} will be proved under regularity conditions in Section \ref{Sec-AsymptoticResults-BIC}.

% In the following, we use $\bm{\theta}_{\iota} =(\bm{\alpha}_{\iota}^{\prime},\bm{\sigma}_{\iota}^{\prime})^\prime$ to emphasize their dependence on the order $\iota = (p,r,s)$, 
% and assume the parameter space $\Theta_\iota$ is a compact subset of $\Omega_\iota \times \mathbb{R}^{N^2d} \times \mathbb{R}^{N(N+1)/2}$ with $\Omega_\iota=\{(-1,0)\cup(0,1)\}^r \times \{(0,1)\times (0,\pi )\}^s$.
% Let $p_{\rm max}, r_{\rm max}$ and $s_{\rm max}$ be the predetermined positive integers and $\mathscr{M} = \{\iota \mid 0\le p \le p_{\rm max}, 0\le r\le r_{\rm max} ,0 \le s\le s_{\rm max}\}$.
% We introduce the following Bayesian information criterion (BIC) to select the order of model \eqref{model_VARinf}:
% \begin{align}\label{BIC}
% 	&\text{BIC}(\iota) = T\ln \left|\widehat{\Sigma}(\widehat{\bm \alpha}_{\iota})\right| +n_{\alpha}(\iota)\ln T,
% \end{align}
% where $\widehat{\bm \alpha}_{\iota}$ is the QMLE with the order set to $\iota$, $\widehat{\Sigma}(\widehat{\bm \alpha}_{\iota})=T^{-1} \sum_{t=1}^{T} \widetilde{\bm{\varepsilon}}_t(\widehat{\bm \alpha}_{\iota}) \widetilde{\bm{\varepsilon}}_t^\prime(\widehat{\bm \alpha}_{\iota})$, and $n_{\alpha}(\iota) = r+2s+N^2d$ is the dimension of $\bm \alpha_{\iota}$; see also \cite{schwarz1978estimating}.
% % Here we use $n_{\alpha}(\iota)$ for both the QMLE and LSE, because the order $\iota$ is only related to $\bm \alpha_{\iota}$.
% Let $\widehat{\iota} =\argmin_{\iota \in \mathscr{M}}\text{BIC}(\iota)$.

\section{Asymptotic results}\label{Sec-AsymptoticResults}

This section establishes the asymptotic properties of the QMLE $\widehat{\bm \theta}$ proposed in Section \ref{subSec-MLE} and the selection consistency of the BIC introduced in Section \ref{Sec-BIC}. 
Additionally, a comparative analysis is conducted to demonstrate the asymptotic efficiency gain of the proposed QMLE over the least squares estimator (LSE).

\subsection{Asymptotic properties for QMLE}\label{Sec-AsymptoticResults-QMLE}

Denote $\bm \theta_0 = (\bm \alpha_0^\prime, \bm \sigma_0^\prime)^\prime =(\bm{\omega}_0^{\prime},\bm g_0^\prime,\bm{\sigma}_0^{\prime})^\prime$ as the true value of $\bm \theta$, where $\bm{\omega}_0= (\bm{\lambda}_0^\prime, \bm{\eta}_0^\prime)^\prime$ with $\bm\lambda_0 = (\lambda_{10},\ldots,\lambda_{r0})^\prime$ and $\bm\eta_0 = (\bm{\eta}_{10}^\prime, \ldots ,\bm{\eta}_{s0}^\prime)^\prime$, $\bm g_0 = \text{vec}(G_{01},\ldots,G_{0d})$, and $\bm{\sigma}_0= \text{vech}(\Sigma_0)$. 
Let $\bm \varepsilon_{0t} = \bm \varepsilon_t(\bm \alpha_0)$.
To establish the asymptotic properties for the QMLE $\widehat{\bm \theta}$, we introduce the following assumptions. 

\begin{assum}\label{assum_identifiability}
	($i$) The elements of $\bm \omega_0$ are nonzero, where $\lambda_{01},\ldots,\lambda_{0r}$ are distinct, and $\gamma_{01},\ldots, \gamma_{0s} $ are distinct; 
	($ii$) there exists a positive constant $C$ such that $\|G_{0k}\|>C$ for $p+1\le k \le d$. 
\end{assum}

\begin{assum}\label{assum_stationarity_y}
	$\{\bm y_t, t\in Z\}$ is strictly stationary and ergodic. 
\end{assum}

\begin{assum}\label{assum_stationarity_epsilon}
	$\{\bm \varepsilon_{0t}, t\in Z\}$ is a sequence of $i.i.d.$ random vectors with $E(\bm \varepsilon_{0t})=0_{N\times 1}$ and $E(\bm \varepsilon_{0t}\bm \varepsilon_{0t}') = \Sigma_0$, where $\Sigma_{0}$ is a positive definite and finite matrix.
\end{assum}

% \begin{assum}\label{assum_space}
% 	($i$) The parameter space $\mathcal{A}$ is compact, and the true parameter vector $\bm \alpha_0$ is an interior point in $\mathcal{A}$;
% 	($ii$) $|\lambda_1|,\ldots , |\lambda_r|, \gamma_1,\ldots, \gamma_s \in \Lambda$, where $\Lambda$ is a compact subset of $(0,{\rho})$, and $0 < {\rho} <1$ is a constant. 
% 	%($iii$) there exists a positive constant $C$ such that $\|G_k\|<C$ for all $\bm{\alpha}\in \mathcal{A}$ and $1\le k \le d$;
% \end{assum}

\begin{assum}\label{assum_space}
	($i$) The parameter space $\Theta$ is compact and the true parameter vector $\bm \theta_0$ is an interior point in $\Theta$;
	($ii$) $|\lambda_1|,\ldots , |\lambda_r|, \gamma_1,\ldots, \gamma_s \in \Lambda$, where $\Lambda$ is a compact subset of $(0,{\rho})$, and $0 < {\rho} <1$ is a constant; 
	%($iii$) there exists a positive constant $C$ such that $\|G_k\|<C$ for all $\bm{\alpha}\in \mathcal{A}$ and $1\le k \le d$;
	($iii$) $\Sigma(\bm \sigma)$ is a finite and positive definite symmetric matrix with its eigenvalue having a positive lower and upper bound denoted by $\underline{\sigma}$ and $\bar{\sigma}$, respectively.
\end{assum}

% \begin{assum}\label{assum_space_ML}
% 	($i$) The parameter space $\Theta$ is compact and the true parameter vector $\bm \theta_0$ is an interior point in $\Theta$;
% 	($ii$) $\Sigma(\bm \sigma)$ is a finite and positive definite symmetric matrix with its eigenvalue having a positive lower and upper bound denoted by $\underline{\sigma}$ and $\bar{\sigma}$, respectively.
% \end{assum}

Assumption \ref{assum_identifiability} is required for model identification, ensuring the uniqueness of the true parameter vector $\bm\alpha_0$ for model \eqref{VARprs} with order ($p,r,s$); see Lemma \ref{lemma_identifiability} in the Appendix for details. Specifically, Assumption \ref{assum_identifiability}(i) requires $\lambda_{01},\ldots,\lambda_{0r}, \gamma_{01},\ldots, \gamma_{0s}$ to be nonzero and different for identification, and Assumption \ref{assum_identifiability}(ii) holds if and only if there exist nonzero elements in $G_{0k}$. Assumption \ref{assum_identifiability} is also necessary to guarantee the true order $(p,r,s)$ to be irreducible; see Remark \ref{order-condition} in Section \ref{Sec-AsymptoticResults-BIC} for details.
For the strict stationarity and ergodicity of $\{\bm y_t\}$ in Assumption \ref{assum_stationarity_y}, Theorem 2.3 in \cite{huang2025} provides a sufficient condition for model \eqref{model_VARinf} by assuming $\{\bm \varepsilon_{0t}\}$ is a strictly stationary and ergodic sequence with $E(\|\bm \varepsilon_{0t}\|^2_2)<\infty$, that is, there exists $0<\rho<1$ such that $\max\{|\lambda_{01}|,\ldots,|\lambda_{0r}|,\gamma_{01},\ldots, \gamma_{0s}\}\le \rho$ and $\sum_{k=1}^{p}\|G_{0k}\|_{\operatorname{op}} +\rho/(1-\rho) \sum_{k=p+1}^{d} \|G_{0k}\|_{\operatorname{op}}<1$. 
Assumption \ref{assum_stationarity_epsilon} is necessary for the consistency and asymptotic normality of QMLEs, which together with Assumption \ref{assum_identifiability} guarantees that $\bm\theta_0$ is the unique minimizer of the expected error loss function $E[\mathcal{L}(\bm{\theta})]$. Note that Assumption \ref{assum_stationarity_epsilon} only imposes moment conditions on the $i.i.d.$ errors $\{\bm \varepsilon_{0t}\}$, which is weaker than the $i.i.d.$ sub-Gaussian condition assumed in \cite{huang2025} and \cite{zheng2025interpretable}.
Finally, Assumption \ref{assum_space} includes basic conditions for the consistency and asymptotic normality of QMLEs. The compactness of the parameter space $\Theta$ is standard for ensuring consistency, while the requirement that $\bm \theta_0$ lies in the interior of $\Theta$ is necessary for asymptotic normality.

Let $n = r+2s+N^2d + N(N+1)/2$ be the dimension of $\bm{\theta}$. 
Define the $n \times n$ matrices 
\begin{align*}
	\mathcal{I} = E\left( 
	\frac{\partial l_{t}(\bm \theta_0)}{\partial \bm \theta} 
	\frac{\partial l_{t}(\bm \theta_0)}{\partial \bm \theta^\prime}
	\right)
	\quad \text{and}\quad
	\mathcal{J} = E\left( 
	\frac{\partial ^2 l_t(\bm{\theta}_0)}{\partial \bm{\theta}\partial \bm{\theta}^\prime}
	\right) .
\end{align*}

\begin{thm}\label{thmQMLE}
	Suppose that Assumptions \ref{assum_identifiability}--\ref{assum_space} hold and $E(\|\bm y_t\|_2^2)<\infty$. As $T \to \infty$, we have ($i$) $\widehat{\bm{\theta}} \rightarrow_{p} \bm{\theta}_0$;
	($ii$) furthermore, if $E(\|\bm \varepsilon_{0t}\|_2^4) <\infty $ and $\mathcal{I} $ is positive definite, then $\sqrt{T}(\widehat{\bm{\theta}} -\bm{\theta}_0 ) \rightarrow_d N(0_{n\times 1},\mathcal{J}^{-1} \mathcal{I} \mathcal{J}^{-1})$.
\end{thm}

Theorem \ref{thmQMLE} establishes the consistency and asymptotic normality for the QMLE $\widehat{\bm{\theta}}$. 
We can see that the consistency requires $E(\|\bm \varepsilon_{0t}\|_2^2)< \infty$, whereas the asymptotic normality further needs $E(\|\bm \varepsilon_{0t}\|_2^4)< \infty$. 
Moreover, the consistency established in Theorem \ref{thmQMLE} also relaxes the assumption of sub-Gaussian errors imposed in \cite{huang2025} and \cite{zheng2025interpretable}.
If $\bm \varepsilon_{0t}$ is multivariate normal, then the QMLE $\widehat{\bm{\theta}}$ reduces to the MLE and its asymptotic normality can be simplified to $\sqrt{T}(\widehat{\bm{\theta}} -\bm{\theta}_0 ) \rightarrow_d N(0_{n\times 1},\mathcal{J}^{-1})$ as $T \to \infty$. 

To calculate the asymptotic variances of $\widehat{\bm{\theta}}$, we can estimate the matrices $\mathcal{I}$ and $\mathcal{J}$ using their sample averages with $\bm{\theta}_0$ replaced by $\widehat{\bm{\theta}}$, leading to a consistent estimator of the covariance matrix denoted by $\widehat{\text{Var}}(\widehat{\bm{\theta}})=\widehat{\mathcal{J}}^{-1} \widehat{\mathcal{I}} \widehat{\mathcal{J}}^{-1}/T$. Then we can construct the significance test for each $\theta_j$ using the $z$ test statistic $z_j=\widehat{\theta}_j/\sqrt{[\widehat{\text{Var}}(\widehat{\bm{\theta}})]_{jj}}$, where $\theta_j$ (or $\widehat{\theta}_j$) is the $j$th element of $\bm{\theta}$ (or $\widehat{\bm{\theta}}$), and $[\widehat{\text{Var}}(\widehat{\bm{\theta}})]_{jj}$ is the $j$th diagonal entry of $\widehat{\text{Var}}(\widehat{\bm{\theta}})$. Since $z_j\to_d N(0,1)$ as $T\to\infty$, we can conclude that $\theta_j\neq 0$ if $|z_j|>z_{\alpha/2}$ at the significance level $\alpha$, where $z_{\alpha}$ is the $\alpha$th quantile of standard normal distribution.

\subsection{Selection consistency for BIC}\label{Sec-AsymptoticResults-BIC}

Let $\iota^* = (p^*,r^*,s^*)$ be the true order of model \eqref{model_VARinf}. Theorem \ref{thmBIC} proves that the selected order $\widehat{\iota}$ via BIC converges to the true order $\iota^*$ consistently under regularity conditions.
%Theorem \ref{thmBIC} below establishes the selection consistency of the BIC in \eqref{BIC}. 

\begin{thm}\label{thmBIC}
	Suppose $\iota^*$ is irreducible. Under the conditions of Theorem \ref{thmQMLE}, if $\iota^* \in \mathscr{M} $, then $P(\widehat{\iota} = \iota^*) \to 1$ as $T \to \infty$.
\end{thm}

%Theorem \ref{thmBIC} establishes the selection consistency of the BIC in \eqref{BIC} under regularity conditions.  
Note that assuming $\iota^*$ to be irreducible is necessary for proving the selection consistency of BIC in Theorem \ref{thmBIC}.
If $\iota^*$ is reducible, then there exists $\iota^{\dagger}\neq\iota^*$ such that model \eqref{model_VARinf} with order $\iota^{\dagger}$ generates the same data  process as model \eqref{model_VARinf} with order $\iota^*$.
To avoid this identification issue, we assume that $\iota^*$ is irreducible, and provide a sufficient condition for its irreducibility in Remark \ref{order-condition}. 
%Note that assuming $\iota^*$ to be irreducible is necessary for proving the selection consistency of BIC, although Assumption \ref{assum_identifiability} guarantees identification for model \eqref{model_VARinf} under a given order ($p,r,s$). Thus we carefully investigate the identification condition for model \eqref{model_VARinf} of any orders; see Remark \ref{order-condition}. 

\begin{remark}[A sufficient condition for irreducible $\iota^*$]\label{order-condition}
	%For purpose of parsimony, we require that model \eqref{model_VARinf} with the true order $\iota^*$ is irreducible. 
	It can be verified that the order $(p,r,s)$ cannot be reduced to smaller ones if Assumption \ref{assum_identifiability} and $G_{0p} \neq \sum_{i=1}^{r}G_{0,p+i} +\sum_{j=1}^{s} G_{0,p+r+2j-1}$ hold. 
	Specifically, Assumption \ref{assum_identifiability} ensures the orders $r$ and $s$ non-degenerate, whereas the restriction on $G_{0p} \neq \sum_{i=1}^{r}G_{0,p+i} +\sum_{j=1}^{s} G_{0,p+r+2j-1}$ guarantees $p$ non-decreasing; see the proof of this remark in the Appendix for details. 	
\end{remark}

\subsection{Efficiency comparison between QMLE and LSE}\label{Sec-AsymptoticResults-comparison}

For scalable ARMA models in the high-dimensional setting with a large $N$, \cite{zheng2025interpretable} and \cite{huang2025} developed the regularized LSEs and derived their non-asymptotic error bounds under the sparsity and low-rank conditions on the parameter matrices $G_k$ for $1 \le k \le d$, respectively. 
For scalable ARMA models with a fixed $N$, we can also consider the LSE without any sparse or low-rank conditions imposed on $G_k$'s. 
%Extending their analysis, we can also consider the LSE in the case where no sparse or low-rank conditions are imposed on the parameter matrices $G_k$'s, and derive its asymptotic properties. 
%However, since the innovation $\bm{\varepsilon}_{0t}$ in model \eqref{VARprs} has a general covariance matrix, the LSE suffers from efficiency loss as it neglects the concurrent correlations and heterogeneity among $\{\bm \varepsilon_{0t}\}$. 
This section compares the proposed QMLE in Section \ref{subSec-MLE} with the LSE theoretically. 

Let $\mathcal{A}$ be the parameter space of $\bm{\alpha}=(\bm{\omega}^\prime, \bm{g}^\prime)^\prime$, where $\mathcal{A}$ is a subset of $\{(-1,0)\cup (0,1)\}^r \times \{(0,1)\times (0,\pi )\}^s \times \mathbb{R}^{N^2d}$. 
Recall that the error function $\bm{\varepsilon}_t(\bm{\alpha}) = \bm y_t - \sum_{h=1}^{\infty}A_{h}(\bm{\omega}, \bm{g})\bm y_{t-h}$. 
Denote $\mathcal{L}_{LS}(\bm{\alpha})= T^{-1} \sum_{t=1}^{T}\| \bm{\varepsilon}_t(\bm{\alpha})\|_2^2$ as the squared error loss function. 
% Note that $\bm{\varepsilon}_t(\bm{\alpha})$ and $\mathcal{L}_{LS}(\bm{\alpha})$ depend on observations in the infinite past, initial values for $\{\bm{y}_s,s\le 0\}$ will be needed in practice. 
Similarly, we simply set $\bm{y}_s=0_{N\times 1}$ for $s\le 0$ and denote the resulting function of $\mathcal{L}_{LS}(\bm{\alpha})$ as $\widetilde{\mathcal{L}}_{LS}(\bm \alpha)$. 
Then the LSE for $\bm\alpha$ is defined as
\begin{align}\label{LSE}
	\widehat{\bm \alpha}_{LS} = \argmin_{\bm{\alpha}\in \mathcal{A}} \widetilde{\mathcal{L}}_{LS}(\bm{\alpha}),
\end{align}
and the LSE for $\Sigma$ is $\widehat{\Sigma}_{LS} = T^{-1}\sum_{t=1}^{T}\widetilde{\bm \varepsilon}_{t}(\widehat{\bm \alpha}_{LS}) \widetilde{\bm \varepsilon}_{t}^\prime (\widehat{\bm \alpha}_{LS})$. 
% In the same way, setting the initial values to zero will be taken into account in our theoretical analysis and be proven to be asymptotically negligible.
Let $n_{\alpha} = r+2s+N^2d$ be the dimension of $\bm{\alpha}$, and define the $n_{\alpha} \times n_{\alpha}$ matrices
\begin{align*}
	\mathcal{I}_{LS} = E\left( \frac{\partial \bm \varepsilon_{t}^\prime(\bm \alpha_0)}{\partial \bm \alpha} \Sigma_0
	\frac{\partial \bm \varepsilon_{t}(\bm \alpha_0)}{\partial \bm \alpha^\prime}\right)
	\quad \text{and}\quad
	\mathcal{J}_{LS} = E\left( \frac{\partial \bm \varepsilon_{t}^\prime(\bm \alpha_0)}{\partial \bm \alpha} 
	\frac{\partial \bm \varepsilon_{t}(\bm \alpha_0)}{\partial \bm \alpha^\prime}\right) .
\end{align*}

\begin{proposition}\label{thmLS}
	Suppose that Assumptions \ref{assum_identifiability}--\ref{assum_space} hold and $E(\|\bm y_t\|_2^2)<\infty$. As $T \to \infty$, we have ($i$) $\widehat{\bm{\alpha}}_{LS} \rightarrow_{p} \bm{\alpha}_0$ and $\widehat{\Sigma}_{LS}\rightarrow_p \Sigma_{0}$;
	($ii$) 
	$\sqrt{T}(\widehat{\bm{\alpha}}_{LS}  -\bm{\alpha}_0 ) \rightarrow_d N(0_{n_{\alpha}\times 1}, \Xi_{LS})$, where $\Xi_{LS} = \mathcal{J}_{LS}^{-1}\mathcal{I}_{LS}\mathcal{J}_{LS}^{-1}$;
	($iii$) 
	furthermore, if $E(\|\bm \varepsilon_{0t}\|_2^4)< \infty$, then $\sqrt{T} ({\rm vec}(\widehat{\Sigma}_{LS})- {\rm vec}(\Sigma_{0}))\rightarrow_{d}N(0_{N^2\times 1},\mathcal{K})$, where $\mathcal{K} = \var[{\rm vec}(\bm \varepsilon_{0t} \bm \varepsilon_{0t}^{\prime})]$.
\end{proposition}

Proposition \ref{thmLS} establishes the consistency and asymptotic normality for the LSEs $\widehat{\bm{\alpha}}_{LS}$ and $\widehat{\Sigma}_{LS}$. 
It is noteworthy that the consistency of $\widehat{\bm{\alpha}}_{LS}$ and $\widehat{\Sigma}_{LS}$ as well as the asymptotic normality of $\widehat{\bm{\alpha}}_{LS}$ only require $E(\|\bm \varepsilon_{0t}\|_2^2)< \infty$, whereas the asymptotic normality of $\widehat{\Sigma}_{LS}$ needs $E(\|\bm \varepsilon_{0t}\|_2^4)< \infty$. 
In comparison with the consistency result in \cite{huang2025} and \cite{zheng2025interpretable}, Proposition \ref{thmLS} relaxes the sub-Gaussian error assumption in their high-dimensional setting.
%In addition, to calculate the asymptotic variances of $\widehat{\bm{\alpha}}_{LS}$ and $\widehat{\Sigma}_{LS}$, we can approximate the expectations in matrices $\Xi_{LS}$ and $\mathcal{K}$ using sample averages with $\bm{\alpha}_0$ replaced by the LSE $\widehat{\bm{\alpha}}_{LS}$.  

To compare the asymptotic efficiency of the proposed QMLE in \eqref{QMLE} and the LSE in \eqref{LSE}, we can compare the asymptotic covariance matrices of the estimators for $\bm\alpha$ and $\Sigma$ separately; see also \cite{poskitt1994asymptotic}.
Based on Theorem \ref{thmQMLE}, we can prove that $\sqrt{T}(\widehat{\bm \alpha}- \bm \alpha_0)\rightarrow_d N(0_{n_{\alpha}\times 1},\Xi_{QML})$ and $\sqrt{T} ({\rm vec}(\widehat{\Sigma})- {\rm vec}(\Sigma_{0}))\rightarrow_{d}N(0_{N^2\times 1},\mathcal{K})$ as $T \to \infty$, where $\widehat{\bm{\sigma}}= \text{vech}(\widehat{\Sigma})$, and $\Xi_{QML} =\left\{E \left[(\partial \bm \varepsilon_t^\prime (\bm \alpha_0)/\partial \bm \alpha) \Sigma_0^{-1} (\partial \bm \varepsilon_t (\bm \alpha_0)/\partial \bm \alpha^\prime) \right]\right\}^{-1}$; see details in the proof of Proposition \ref{thmcompare_efficient} in the Appendix. 
These, together with Proposition \ref{thmLS}, are used to verify the following proposition.

\begin{proposition}\label{thmcompare_efficient}
	Under the conditions of Theorem \ref{thmQMLE} and Proposition \ref{thmLS}, it holds that (i) the asymptotic covariance matrices of the LSE $\widehat{\bm \alpha}_{LS}$ and the QMLE $\widehat{\bm \alpha}$ satisfy $\Xi_{LS}- \Xi_{QML}\ge 0 $, where the equality holds if and only if $\Sigma_{0}=\sigma^2 I_{N}$ for some $\sigma^2>0$; (ii) the LSE ${\rm vec}(\widehat{\Sigma}_{LS})$ and the QMLE ${\rm vec}(\widehat{\Sigma})$ have the same asymptotic covariance matrix $\mathcal{K} = \var[{\rm vec}(\bm \varepsilon_{0t} \bm \varepsilon_{0t}^{\prime})]$. 
\end{proposition}

Proposition \ref{thmcompare_efficient} indicates that the LSE $\widehat{\bm\alpha}_{LS}$ is less efficient than the QMLE $\widehat{\bm \alpha}$ when the innovation $\bm{\varepsilon}_{0t}$ in model \eqref{VARprs} has a general covariance matrix, and they are asymptotically equivalent if and only if the innovation $\bm \varepsilon_{0t}$ is homogeneous and uncorrelated among its components. %When the innovation $\bm{\varepsilon}_{0t}$ in model \eqref{VARprs} has a general covariance matrix, the LSE $\widehat{\bm\alpha}_{LS}$ suffers from efficiency loss as it neglects the concurrent correlations and heterogeneity among $\{\bm \varepsilon_{0t}\}$. 
Moreover, the LSE $\widehat{\Sigma}_{LS}$ and the QMLE $\widehat{\Sigma}$ are asymptotically equivalent. 
Thus, we recommend the QMLE in practice owing to its efficiency gain for general situations.

To measure the relative efficiency of $\widehat{\bm\alpha}$ and $\widehat{\bm\alpha}_{LS}$ numerically, we can compute the asymptotic relative efficiency (ARE) of $\widehat{\bm \alpha}_{LS}$ to $\widehat{\bm \alpha}$, which is defined as $\text{ARE}(\widehat{\bm \alpha}_{LS}, \widehat{\bm \alpha}) 
= \left(|\Xi_{QML}|/|\Xi_{LS}|\right)^{1/n_{\alpha}}$; see \cite{serfling2009approximation}. 
Given a true parameter $\bm \alpha_0$, we can approximate the asymptotic covariance matrices $\Xi_{LS}$ and $\Xi_{QML}$ using their sample averages from a large generated sequence, and then calculate the ARE.
Simulation results in Section \ref{Sec-Simulation} further verify the theoretical findings that the QMLE $\widehat{\bm \alpha}$ is asymptotically more efficient than the LSE $\widehat{\bm \alpha}_{LS}$ while the QMLE $\widehat{\Sigma}$ is asymptotically equivalent to the LSE $\widehat{\Sigma}_{LS}$.

\section{Simulation studies}\label{Sec-Simulation}

This section conducts simulation experiments to examine the finite sample performance of the proposed estimator QMLE and BIC for order selection. 
Moreover, we compare the asymptotic efficiency between the QMLE and the LSE. 
%Finally, we also compare the computational efficiency of scalable ARMA models calculated by BCD algorithm with that of VARMA models computed via L-BFGS algorithm.

The scalable ARMA model \eqref{model_VARinf} is proposed as a tractable alternative to the VARMA model \citep{huang2025,zheng2025interpretable}. 
Thus, we generate data from VARMA models that can be rewritten as scalable ARMA models in \eqref{model_VARinf}. 
Specifically, we consider the following VARMA$(1,1)$ model as the data generating process (DGP):
\begin{equation} \label{DGP-VARMA}
	\bm y_t= \Phi \bm y_{t-1}+ \bm \varepsilon_{0t} - \Theta \bm \varepsilon_{0,t-1},
\end{equation}
where $\{\bm \varepsilon_{0t}\}$ is a sequence of $i.i.d.$ random vectors with zero mean and covariance matrix $\Sigma_0$,
and the AR and MA coefficient matrices are specified via Jordan decompositions $\Phi = B_1J_{\rm AR}B_1^{-1} $ and $\Theta =B_2J_{\rm MA}B_2^{-1}$, respectively.  
Here $B_1$ and $B_2$ are $N\times N$ orthogonal matrices, $J_{\rm AR} $ is a diagonal matrix, and $J_{\rm MA} = \diag\{\lambda_{1},\ldots,\lambda_{r},C_{1},\ldots,C_{s},0_{(N-r-2s)\times 1}^{\prime}\}$ with $C_j$ for $1 \leq j \leq s$ being $2\times 2 $ block matrices defined as follows
\begin{align*}
	C_{j} = \gamma_{j}\left(
	\begin{matrix}
		\cos(\varphi_{j}) & \sin (\varphi_{j})\\
		-\sin (\varphi_{j}) & \cos(\varphi_{j})
	\end{matrix}
	\right).
\end{align*}
Then, this VARMA$(1,1)$ model can be rewritten into a scalable ARMA model in \eqref{model_VARinf} with order $(p,r,s)=(1,r,s)$. 
Particularly, we set $J_{\rm AR}$ with moduli less than one and generate a random orthogonal matrix $B_1$, yielding $\Phi$. For the MA component, we specify $\lambda_i$'s and $C_j$'s to obtain $J_{\rm MA}$, and get $\Theta$ by independently randomly generating orthogonal matrix $B_2$. Consequently, the corresponding $G_k$'s can be obtained as in \cite{zheng2025interpretable}; see details in Appendix Section~\ref{Sec-relatedVARMA}.

\subsection{QMLE} \label{sim-estimation}
The first experiment aims to examine the finite sample performance of the QMLE $\widehat{\bm \theta}=(\widehat{\bm \alpha}^\prime, \widehat{\bm \sigma}^\prime)^\prime$ proposed in Section \ref{subSec-MLE}. %For comparison, we also present that for the LSE $\widehat{\bm \theta}_{LS}=(\widehat{\bm \alpha}_{LS}^\prime, \widehat{\bm \sigma}_{LS}^\prime)^\prime$ in Section \ref{Sec-AsymptoticResults-comparison} with $\widehat{\bm \sigma}_{LS}= \text{vech}(\widehat{\Sigma}_{LS})$. 

The data is generated from model \eqref{DGP-VARMA} under the following setting: 
\begin{align*}
	\text{DGP1}:\; &N=6, \;(p,r,s)= (1,1,1), \; (\lambda_{01}, \gamma_{01} ,\varphi_{01})=(-0.8,0.8,\pi/4),\\
	&J_{\rm AR} = \diag\{0.5,0.5,0.5,-0.5,-0.5,-0.5\}. 
\end{align*}
Moreover, $\{\bm \varepsilon_{0t}\}$ follow the multivariate normal or Student's $t_5$ distribution with zero mean and covariance matrix $\Sigma_0 =   a \bm 1_N \bm 1_N^\prime + (1-a)I_N $, where $a = 0$ or $0.5$. 
The sample size is set to $T=500$, $750$ or $1000$, and $500$ replications are generated for each sample size. 
% Since the number of elements contained in $G_k$'s and $\Sigma_0$ is $O(N^2)$, for the sake of conciseness in table representation, we only present the fitted results in the first row of each matrix. 
We employ BCD algorithm to solve the QMLE, where an evenly spaced grid is used for initialization with the step size being 0.2 for $\lambda_i$'s (or $\gamma_j$'s) and $0.1 \pi$ for $\varphi_j$'s, and $k_{\max}=100$ and $\delta=10^{-6}$ are specified for algorithm convergence.

Table \ref{tab_DGP1} summarizes the biases, empirical standard deviations (ESDs) and asymptotic standard deviations (ASDs) of the QMLE $\widehat{\bm \theta}$ under DGP1 with two settings of $\Sigma_{0}$, respectively. To save space, we only present the estimation results for a subset of the elements from the matrices $G_k$ and $\Sigma_0$.
The findings can be summarized as follows. 
First, as the sample size increases, most of the biases, ESDs and ASDs become smaller, and the ESDs get closer to the corresponding ASDs. This confirms the asymptotic results established in Theorem \ref{thmQMLE}.
Second, heavier-tailed $\{\bm \varepsilon_{0t}\}$ lead to larger ESDs and ASDs, especially for the covariance parameters $\sigma_{ij}$, which is as expected since the heavier tail of $\bm \varepsilon_{0t}$ can make the QMLE less efficient.

%Third, the ESDs of $\widehat{\bm \alpha}_{LS}$ are close to that of $\widehat{\bm \alpha}$ when $\Sigma_{0}$ is an identity matrix, while most of the ESDs of $\widehat{\bm \alpha}$ are smaller than that of $\widehat{\bm \alpha}_{LS}$ when $\Sigma_{0}$ is not proportional to an identity matrix. This is consistent to Proposition \ref{thmcompare_efficient}($i$) that the LSE is asymptotically equivalent to the QMLE if $\Sigma_{0}$ is proportional to an identity matrix and otherwise the QMLE is more efficient than the LSE.
%Fourth, the similar ESDs of $\widehat{\bm \sigma}$ and $\widehat{\bm \sigma}_{LS}$ support the statement in Proposition \ref{thmcompare_efficient}($ii$) that the asymptotic covariance matrices of $\widehat{\bm \sigma}$ and $\widehat{\bm \sigma}_{LS}$ are the same. 

\subsection{Model selection} \label{sim-BIC}
The second experiment verifies the selection consistency of the proposed BIC in Section \ref{Sec-BIC}. The data is generated from model \eqref{DGP-VARMA} under the following setting: 
\begin{align*}
	\text{DGP2}:\;& N=6, \;(p,r,s)= (1,1,0),\;\lambda_{01} \in \{0.7,0.8,0.9\},\\
	&J_{\rm AR} = \diag\{-0.8,0.8,-0.7,0.7,-0.6,0.6\},
\end{align*}
where $\{\bm \varepsilon_{0t}\}$ is specified as in the first experiment.
This design makes DGP2 correspond to the scalable ARMA model \eqref{model_VARinf} with the true order $\iota^* = (p^*,r^*,s^*)=(1,1,0)$. 
The BIC in \eqref{BIC} is employed to select the order $\iota=(p,r,s)$ with $p_{\rm max}=r_{\rm max}=s_{\rm max}=2$. 
The underfitted, correctly selected, and overfitted models by BIC correspond to $\widehat{\iota}\in \mathscr{M}_{\rm under}, \widehat{\iota}\in \mathscr{M}_{\rm true}$, and $\widehat{\iota}\in \mathscr{M}_{\rm over},$ respectively, where 
$\mathscr{M}_{\rm under} = \{\iota \in \mathscr{M} \mid p< p^* \text{ or } r< r^* \text{ or } s< s^*\}$, $\mathscr{M}_{\rm true}
= \{ \iota \in \mathscr{M} \mid p= p^*, r= r^*, s= s^*\}$, and 
$\mathscr{M}_{\rm over}
= \{ \iota \in \mathscr{M} \mid p\ge p^*, r\ge r^*, s\ge  s^*\} \setminus  \{\iota^*\}$.
We consider the sample sizes $T=500$, $750$ or $1000$, and generate $500$ replications for each sample size.

Table \ref{tab_BIC} provides the percentages of underfitting, correct selection, and overfitting cases by the BIC under DGP2 with two settings of $\Sigma_{0}$, respectively. It can be seen that BIC performs better when the sample size $T$ increases, which supports the selection consistency of the BIC in Theorem \ref{thmBIC}. 
Moreover, the selection accuracy of BIC improves with larger values of $\lambda$, which is expected since the larger $\lambda$ implies stronger signal in order selection.

\subsection{Asymptotic efficiency comparison} \label{sim-ARE}
In the third experiment, we investigate the efficiency gain of the proposed QMLE in \eqref{QMLE}, relative to the LSE in \eqref{LSE}. 
Specifically, we calculate the $\text{ARE}(\widehat{\bm \alpha}_{LS}, \widehat{\bm \alpha})$ defined in Section \ref{Sec-AsymptoticResults-comparison}, which measures the relative efficiency of $\widehat{\bm\alpha}_{LS}$ to $\widehat{\bm\alpha}$. 

We generate a sequence of sample size $T=20000$ from model \eqref{DGP-VARMA} under two settings as follows:
\begin{align*}
	&N=6 , J_{\rm AR} = \diag\{-0.8,0.8,-0.7,0.7,-0.6,0.6\}\\
	& \text{DGP3(a)}:\; (p,r,s)= (1,2,0), \;\lambda_{01} = -\lambda_{02} =\lambda \;\text{with}\; \lambda \in \{0.2,0.4,0.6,0.8\}; \\
	&\text{DGP3(b)}:\; (p,r,s)= (1,0,1), \;(\gamma_{01},\varphi_{01})= (\gamma, \pi /4) \;\text{with}\; \gamma \in \{0.2,0.4,0.6,0.8\}.
\end{align*}
Varying $\lambda$ and $\gamma$ allows us to examine how relative efficiency of $\widehat{\bm\alpha}_{LS}$ to $\widehat{\bm\alpha}$ changes with these parameters. 
Specifically, larger values of $\lambda$ and $\gamma$ imply the roots of the MA characteristic equation $|I_N-\Theta B|=0$ are nearer to the unit circle. 
To evaluate the impact of $\Sigma_0$ on ARE, we let $\{\bm \varepsilon_{0t}\}$ follow either a multivariate normal or Student's $t_5$ distribution with zero mean and different covariance matrices, i.e., $\Sigma_{0} = a\bm 1_N \bm 1_N^\prime+(1-a)I_N$, where $a = 0, 0.6$, or $0.9$.
This design allows $a$ to control the degree of cross-sectional correlation, with a larger value of $a$ indicating that $\Sigma_{0}$ deviates further from a diagonal matrix.

Table \ref{tab_ARE} presents $\text{ARE}(\widehat{\bm \alpha}_{LS}, \widehat{\bm \alpha})$ under DGP3(a) and DGP3(b) with different settings for $\Theta$ and $\Sigma_0$.
It can be observed that the ARE is almost one when $\Sigma_0 = I_N$, which confirms the asymptotic equivalence between the QMLE and LSE for this case established in Proposition \ref{thmcompare_efficient}($i$). 
By contrast, the ARE decreases as the off-diagonal elements of $\Sigma_0$ increase.
This implies that the concurrent correlation between components of $\bm \varepsilon_{0t}$ will lead to efficiency gain of the QMLE relative to the LSE. 
Moreover, when $\Sigma_0$ is not proportional to an identity matrix, the efficiency gain of the QMLE over the LSE becomes slightly more pronounced as $\lambda$ or $\gamma$ increases. This observation aligns with the finding for univariate ARMA models in \cite{brockwell2009time}, which states that the efficiency gain of the MLE over the LSE becomes more significant as the roots of the MA characteristic equation move closer to the unit circle.

In summary, the above simulation results confirm the asymptotic results in Section \ref{Sec-AsymptoticResults}. 
Moreover, the results indicate that the QMLE outperforms the LSE in efficiency for the scalable ARMA model, and thus we suggest to use the QMLE method in practice. 
% Additionally, the scalable ARMA model with the QMLE calculated by BCD algorithm, is computationally more efficient than the QMLE of VARMA models especially when the dimension increases. 

\section{An empirical example}\label{Sec-Realdata}

This section analyzes six monthly U.S. macroeconomic indicators from the FRED-MD database \citep{mccracken2016fred}, covering the period from January 1959 to August 2025.
The indicators include industrial production (IP), the civilian unemployment rate (UR), the consumer price index (CPI), real wages (RW), real personal consumption expenditures (CONS), and the federal funds rate (FFR).
This set of indicators is chosen to provide a parsimonious while comprehensive profile of real economic activity, labor market conditions, inflation dynamics, household demand, and monetary policy.

Following \citet{mccracken2016fred}, each series is transformed to be stationary as summarized in Table~\ref{table_preproceed}, and then standardized to have zero mean and unit variance.
As a result, we obtain a time series $\{\bm y_t\}$ of dimension $N=6$ and length $T= 798$.
Figure \ref{MACM_realdate} reports the multivariate autocorrelation matrix (MACM) of the standardized series up to lag 15, which reveals persistent cross-lag dependence with slow decay and occasional oscillatory patterns.
Such behavior motivates us to analyze ${\bm y_t}$ using the scalable ARMA model \eqref{model_VARinf} and the proposed inference tools in Section \ref{Sec-Methodology}.

We first fit model \eqref{model_VARinf} to the full sample. 
Using the BIC in Section \ref{Sec-BIC} with $p_{\rm max}=r_{\rm max}=s_{\rm max}=2$, the selected order is $(p,r,s)=(1,1,0)$. 
This together with the representation in \eqref{model_VARinf-2} implies that the dynamics can be decomposed into a short-run autoregressive component and a geometrically decaying infinite-lag component,
\begin{align} \label{model-realdata}
	\bm y_t = G_1 \bm y_{t-1} + \sum_{h=1}^{\infty}\lambda_{1}^h G_2 \bm y_{t-1-h} + \bm \varepsilon_{t}.
\end{align}
The estimated decay parameter is $\widehat{\lambda}_1 = 0.84$ with an asymptotic standard error of $0.02$, which is highly statistically significant. 
This strongly supports the presence of a persistent long-run adjustment channel. %, which is not captured by a finite-order VAR. 
Additionally, the positive value of $\lambda_1$ indicates a gradual and stable long-run response to past shocks, consistent with gradual and persistent adjustment in macroeconomic dynamics over time.
Figure~\ref{coef_G} shows heatmaps of the estimates for coefficient matrices $G_1$ and $G_2$, with statistical significance indicated by asymptotic standard errors. 
From model \eqref{model-realdata}, $G_1$ captures short-run dynamics, and its heatmap reveals significant immediate cross-variable transmission. The matrix $G_2$ governs the long-run dynamics, with its heatmap depicting the pattern of persistent influences.
Moreover, Figure~\ref{coef_G} shows that Real Wages (RW) and Real Personal Consumption Expenditures (CONS) have immediate effects on most macroeconomic variables. In contrast, Industrial Production (IP), Unemployment Rate (UR), and Federal Funds Rate (FFR) primarily exert long-run influences on the system. 
Overall, these findings highlight the complex interplay of short- and long-term factors within the macroeconomic system.

To evaluate the forecasting performance of the fitted models, we next conduct one-step-ahead predictions through a rolling forecasting procedure with a fixed moving window. 
Specifically, we begin with the forecast origin $t_0=n_0+1= 739$, and fit the scalable ARMA model \eqref{model_VARinf} with order $(1,1,0)$ by the QMLE method using the data from the beginning to $n_0=738$. We then compute the one-month-ahead forecast for $\bm y_{t_0+1}$, denoted by $\widehat{\bm y}_{t_0+1}=\widehat{G}_1 \bm y_{t_0} + \sum_{h=1}^{\infty}\widehat{\lambda}_{1}^h \widehat{G}_2 \bm y_{t_0-h}$, where $\widehat{\lambda}_{1}$, $\widehat{G}_{1}$ and $\widehat{G}_{2}$ are the QMLEs. After each prediction, we move the window forward by one, repeating the procedure until all data are used.
In total, we obtain $T_{\rm{test}} = 60$ one-month-ahead forecasts, using data from the last 5 years as the test set.
To assess the forecasting accuracy, we calculate the root mean squared forecast error (RMSFE) and mean absolute forecast error (MAFE) defined as $\text{RMSFE}=[(NT_{\rm{test}})^{-1}\sum_{t=t_0+1}^T\|\widehat{\bm y}_{t}-\bm y_t\|_2^2]^{1/2}$ and $\text{MAFE}=(NT_{\rm{test}})^{-1}\sum_{t=t_0+1}^T\|\widehat{\bm y}_{t}-\bm y_t\|_1$, respectively. 
To facilitate comparison with the VAR benchmark, we further report the relative improvement, defined for method $\rm M$ as
\(\operatorname{RI}_{\rm RMSFE}(\rm M)=\{1-\operatorname{RMSFE}(\rm M)/\operatorname{RMSFE}({\rm VAR})\}\times100\%\) 
and
\(\operatorname{RI}_{\rm MAFE}(\rm M)=\{1-\operatorname{MAFE}(\rm M)/\operatorname{MAFE}({\rm VAR})\}\times100\%\). 
A positive (or negative) value indicates that method $\rm M$ yields a lower (or higher) forecast error than the VAR benchmark. %whereas a negative value indicates a higher forecast error.
%Table \ref{table_forecast} summarizes the RMSFEs and MAFEs for one-month-ahead forecasts. 

To compare the forecasting performance of the scalable ARMA model fitted by the proposed QMLE method (SARMA$_1$) with other counterparts, we also perform the rolling forecasting procedure for the following methods: 
\begin{itemize}
	\item SARMA$_2$: Fit a scalable ARMA model with order $(p,r,s)=(1,1,0)$ to $\{\bm y_t\}$ using the LSE method;
	\item VAR: Fit a VAR($2$) model to $\{\bm y_t\}$ using the LSE method, where the order $2$ is selected by BIC;
	\item VARMA$_1$: Fit a VARMA($1,1$) model to $\{\bm y_t\}$ using the QMLE method;  
	\item VARMA$_2$: Fit a VARMA($1,1$) model to $\{\bm y_t\}$ using the iterative ordinary least squares (IOLS) method \citep{dias2018estimation}.
\end{itemize}
Since the LSE and QMLE of coefficient matrices are equivalent for the VAR model, we only compare our methods with the VAR model fitted by the least squares method. 
For comparison of the forecasting performance, Table \ref{table_forecast} reports the relative improvements for each method over the VAR benchmark.
%reports the RMSFE and MAFE of one-month-ahead forecasts for each method.
We conclude that the scalable ARMA model has the smallest RMSFEs and MAFEs, and thus outperforms the VAR and VARMA models in the out-of-sample forecast. 
Particularly, the scalable ARMA model fitted by the QMLE approach performs the best.   
These results collectively affirm the advantages of the scalable ARMA framework in modeling and forecasting multivariate time series. 
As a result, we suggest the scalable ARMA model fitted by the QMLE method as a competitive and reliable approach for multivariate time series analysis.

\section{Conclusion and discussion}\label{Sec-Conclusion}

This paper develops theoretically grounded inferential tools for the scalable ARMA model, including the quasi-maximum likelihood method for estimation and the BIC method for order selection. 
By establishing the asymptotic properties of the proposed QMLE, we enable formal significance testing for the scalable ARMA model to uncover causal relationships among time series components. 
The QMLE is proven to be asymptotically more efficient than the LSE in the presence of cross-sectionally correlated or heteroskedastic innovations.
For efficient computation, we also introduce a BCD algorithm. 
%\textcolor{red}{The QMLE implemented via the BCD algorithm is computationally much more efficient than the QMLE for standard VARMA models as the dimension increases.} 
Moreover, the superior efficiency and forecasting accuracy of the QMLE over the LSE are demonstrated via simulations and an application to macroeconomic indicators. 

The future research may extend this work in two ways. 
First, more robust inference methods can be investigated for the scalable ARMA model to relax the moment conditions on the process and innovation. 
Second, the proposed inference framework can be extended to more general multivariate time series models such as the scalable ARMA model with factors. 
We leave these topics for future investigation.

\begin{table}[htbp]
	\centering
	\renewcommand{\arraystretch}{1.3}
	\caption{\label{tab_DGP1}Biases ($\times 10$), ESDs ($\times 10$) and ASDs ($\times 10$) of the QMLE for DGP1, where $\{\bm \varepsilon_{0t}\}$ follow the multivariate normal or Student's $t_5$ distribution with zero mean and covariance matrix $\Sigma_0 =   a \bm 1_N \bm 1_N^\prime + (1-a)I_N $.}
	\begin{adjustbox}{width=\textwidth}
		\begin{tabular}{crrrrrrrrrrrrr}
			\toprule
			&       & \multicolumn{3}{c}{Normal ($a=0$)} & \multicolumn{3}{c}{$t_5$ ($a=0$)} & \multicolumn{3}{c}{Normal ($a=0.5$)} & \multicolumn{3}{c}{$t_5$ ($a=0.5$)} \\
			\cmidrule(lr){3-5} \cmidrule(lr){6-8} \cmidrule(lr){9-11} \cmidrule(lr){12-14} 
			& \multicolumn{1}{c}{$T$} & \multicolumn{1}{c}{Bias} & \multicolumn{1}{c}{ESD} & \multicolumn{1}{c}{ASD} & \multicolumn{1}{c}{Bias} & \multicolumn{1}{c}{ESD} & \multicolumn{1}{c}{ASD} & \multicolumn{1}{c}{Bias} & \multicolumn{1}{c}{ESD} & \multicolumn{1}{c}{ASD} & \multicolumn{1}{c}{Bias} & \multicolumn{1}{c}{ESD} & \multicolumn{1}{c}{ASD} \\
			\midrule
 $\lambda$ & 500 & -0.119 & 0.728 & 0.362 & -0.111 & 0.722 & 0.366 & 0.071 & 1.279 & 0.528 & 0.109 & 1.237 & 0.538 \\
 & 750 & -0.063 & 0.506 & 0.323 & -0.055 & 0.541 & 0.326 & 0.057 & 0.855 & 0.444 & 0.027 & 0.900 & 0.442 \\
 & 1000 & -0.028 & 0.410 & 0.293 & -0.034 & 0.401 & 0.294 & 0.021 & 0.623 & 0.383 & 0.025 & 0.651 & 0.383 \\
$\gamma$ & 500 & 0.094 & 0.498 & 0.260 & 0.100 & 0.499 & 0.261 & 0.089 & 0.770 & 0.381 & 0.090 & 0.807 & 0.383 \\
 & 750 & 0.049 & 0.340 & 0.227 & 0.049 & 0.341 & 0.228 & 0.041 & 0.536 & 0.333 & 0.066 & 0.527 & 0.331 \\
 & 1000 & 0.027 & 0.283 & 0.203 & 0.029 & 0.264 & 0.203 & 0.021 & 0.405 & 0.298 & 0.036 & 0.424 & 0.298 \\
$\phi$ & 500 & -0.030 & 0.610 & 0.313 & -0.037 & 0.599 & 0.314 & -0.074 & 0.917 & 0.431 & -0.096 & 0.970 & 0.432 \\
 & 750 & -0.008 & 0.401 & 0.275 & -0.018 & 0.420 & 0.277 & -0.040 & 0.610 & 0.374 & -0.038 & 0.610 & 0.372 \\
 & 1000 & -0.002 & 0.354 & 0.248 & -0.012 & 0.352 & 0.249 & -0.037 & 0.498 & 0.335 & -0.038 & 0.513 & 0.335 \\
$g_{1,11}$ & 500 & -0.023 & 0.454 & 0.438 & -0.036 & 0.464 & 0.441 & -0.092 & 0.605 & 0.568 & -0.062 & 0.603 & 0.571 \\
 & 750 & -0.015 & 0.372 & 0.358 & -0.012 & 0.372 & 0.360 & -0.038 & 0.495 & 0.464 & -0.016 & 0.495 & 0.466 \\
 & 1000 & -0.020 & 0.329 & 0.310 & -0.024 & 0.329 & 0.311 & -0.043 & 0.423 & 0.402 & -0.038 & 0.424 & 0.403 \\
$g_{2,11}$ & 500 & 0.189 & 0.764 & 0.640 & 0.211 & 0.773 & 0.643 & 0.191 & 1.220 & 0.946 & 0.270 & 1.207 & 0.956 \\
 & 750 & 0.138 & 0.587 & 0.528 & 0.130 & 0.591 & 0.531 & 0.163 & 0.858 & 0.742 & 0.153 & 0.925 & 0.744 \\
 & 1000 & 0.094 & 0.484 & 0.461 & 0.093 & 0.484 & 0.461 & 0.094 & 0.664 & 0.630 & 0.099 & 0.679 & 0.632 \\
$g_{3,11}$ & 500 & -0.057 & 0.538 & 0.499 & -0.061 & 0.525 & 0.500 & -0.038 & 0.887 & 0.697 & -0.083 & 0.845 & 0.702 \\
 & 750 & -0.039 & 0.413 & 0.409 & -0.027 & 0.410 & 0.409 & 0.001 & 0.685 & 0.557 & -0.030 & 0.592 & 0.545 \\
 & 1000 & -0.034 & 0.371 & 0.356 & -0.031 & 0.358 & 0.356 & -0.018 & 0.519 & 0.475 & -0.038 & 0.493 & 0.472 \\
$g_{4,11}$ & 500 & 0.058 & 0.433 & 0.397 & 0.077 & 0.440 & 0.399 & 0.099 & 0.631 & 0.535 & 0.079 & 0.642 & 0.536 \\
 & 750 & 0.046 & 0.351 & 0.326 & 0.034 & 0.346 & 0.327 & 0.082 & 0.470 & 0.426 & 0.022 & 0.475 & 0.422 \\
 & 1000 & 0.039 & 0.294 & 0.284 & 0.033 & 0.299 & 0.284 & 0.071 & 0.387 & 0.364 & 0.035 & 0.394 & 0.364 \\
$\sigma_{11}$ & 500 & -0.495 & 0.580 & 0.597 & -0.487 & 1.090 & 0.973 & -0.523 & 0.595 & 0.596 & -0.522 & 1.088 & 0.966 \\
 & 750 & -0.327 & 0.486 & 0.497 & -0.294 & 0.956 & 0.834 & -0.331 & 0.521 & 0.497 & -0.309 & 0.952 & 0.831 \\
 & 1000 & -0.254 & 0.426 & 0.433 & -0.295 & 0.749 & 0.732 & -0.261 & 0.450 & 0.433 & -0.305 & 0.745 & 0.730 \\
$\sigma_{21}$ & 500 & 0.010 & 0.436 & 0.424 & 0.015 & 0.691 & 0.616 & -0.310 & 0.466 & 0.471 & -0.282 & 0.807 & 0.706 \\
 & 750 & -0.003 & 0.352 & 0.353 & 0.012 & 0.635 & 0.546 & -0.197 & 0.390 & 0.392 & -0.163 & 0.768 & 0.620 \\
 & 1000 & 0.003 & 0.315 & 0.308 & 0.033 & 0.491 & 0.478 & -0.155 & 0.333 & 0.343 & -0.135 & 0.540 & 0.550 \\
$\sigma_{31}$ & 500 & -0.048 & 0.450 & 0.600 & -0.087 & 0.700 & 0.955 & -0.317 & 0.466 & 0.596 & -0.371 & 0.798 & 0.952 \\
 & 750 & -0.027 & 0.358 & 0.500 & -0.000 & 0.610 & 0.849 & -0.195 & 0.389 & 0.496 & -0.175 & 0.773 & 0.856 \\
 & 1000 & -0.015 & 0.315 & 0.437 & -0.017 & 0.501 & 0.771 & -0.146 & 0.335 & 0.436 & -0.170 & 0.581 & 0.767 \\
			\bottomrule
		\end{tabular}%
	\end{adjustbox}
\end{table}%

\begin{table}[htbp]
	\centering
	\caption{\label{tab_BIC} Percentages of underfitted, correctly selected and overfitted cases under DGP2, where the BIC is based on the QMLE, and $\bm \varepsilon_{0t}$ follow the multivariate normal or Student's $t_5$ distribution with zero mean and covariance matrix $\Sigma_0 =   a \bm 1_N \bm 1_N^\prime + (1-a)I_N $. }
		\renewcommand{\arraystretch}{1.1}
	\begin{adjustbox}{width=\textwidth}
		\begin{tabular}{rrrrrrrrrrrrrr}
			\toprule
			&       & \multicolumn{3}{c}{Normal ($a=0$)} & \multicolumn{3}{c}{$t_5$ ($a=0$)} & \multicolumn{3}{c}{Normal ($a=0.5$)} & \multicolumn{3}{c}{$t_5$ ($a=0.5$)} \\
			\cmidrule(lr){3-5} \cmidrule(lr){6-8} \cmidrule(lr){9-11} \cmidrule(lr){12-14}
		\multicolumn{1}{c}{$\lambda$}	&  \multicolumn{1}{c}{$T$}  & \multicolumn{1}{l}{Under} & \multicolumn{1}{l}{Exact} & \multicolumn{1}{l}{Over} & \multicolumn{1}{l}{Under} & \multicolumn{1}{l}{Exact} & \multicolumn{1}{l}{Over} & \multicolumn{1}{l}{Under} & \multicolumn{1}{l}{Exact} & \multicolumn{1}{l}{Over} & \multicolumn{1}{l}{Under} & \multicolumn{1}{l}{Exact} & \multicolumn{1}{l}{Over} \\
			\midrule
0.7 & 500 & 95.0 & 5.0 & 0.0 & 93.0 & 7.0 & 0.0 & 92.0 & 8.0 & 0.0 & 89.6 & 10.4 & 0.0 \\
0.7 & 750 & 18.8 & 81.2 & 0.0 & 18.6 & 81.4 & 0.0 & 11.8 & 88.2 & 0.0 & 8.2 & 91.8 & 0.0 \\
0.7 & 1000 & 0.2 & 99.8 & 0.0 & 0.2 & 99.8 & 0.0 & 0.0 & 100.0 & 0.0 & 0.2 & 99.8 & 0.0 \\
\addlinespace[3pt]
0.8 & 500 & 45.2 & 54.8 & 0.0 & 43.2 & 56.8 & 0.0 & 33.6 & 66.4 & 0.0 & 31.8 & 68.2 & 0.0 \\
0.8 & 750 & 0.4 & 99.6 & 0.0 & 0.4 & 99.6 & 0.0 & 0.2 & 99.8 & 0.0 & 0.2 & 99.8 & 0.0 \\
0.8 & 1000 & 0.2 & 99.8 & 0.0 & 0.0 & 100.0 & 0.0 & 0.0 & 100.0 & 0.0 & 0.2 & 99.8 & 0.0 \\
\addlinespace[3pt]
0.9 & 500 & 7.2 & 92.8 & 0.0 & 8.0 & 92.0 & 0.0 & 5.8 & 94.2 & 0.0 & 4.4 & 95.6 & 0.0 \\
0.9 & 750 & 0.4 & 99.6 & 0.0 & 0.4 & 99.6 & 0.0 & 0.2 & 99.8 & 0.0 & 0.2 & 99.8 & 0.0 \\
0.9 & 1000 & 0.2 & 99.8 & 0.0 & 0.0 & 100.0 & 0.0 & 0.0 & 100.0 & 0.0 & 0.0 & 100.0 & 0.0 \\
\bottomrule
\end{tabular}
\end{adjustbox}
\end{table}

\begin{table}[htbp]
	\centering
	\renewcommand{\arraystretch}{0.85}
	\caption{\label{tab_ARE} $\text{ARE}(\widehat{\bm \alpha}_{LS}, \widehat{\bm \alpha})$ under DGP3(a) and DGP3(b), where $\{\bm \varepsilon_{0t}\}$ follow the multivariate normal or Student's $t_5$ distribution with zero mean and covariance matrix $\Sigma_{0} = a\bm 1_N \bm 1_N^\prime+(1-a)I_N$. }
	\begin{tabular}{crrrrrrrrr}
		\toprule
		&& \multicolumn{4}{c}{$\lambda$} & \multicolumn{4}{c}{$\gamma$}  \\
		\cmidrule(lr){3-6} \cmidrule(lr){7-10}
		&\multicolumn{1}{c}{$a$} & 0.2 & 0.4 & 0.6 & 0.8 & 0.2 & 0.4 & 0.6 & 0.8 \\
\midrule
Normal & 0.0 & 1.000 & 1.000 & 1.000 & 1.000 & 1.000 & 1.000 & 1.000 & 1.000 \\
 & 0.6 & 0.983 & 0.984 & 0.983 & 0.983 & 0.984 & 0.984 & 0.983 & 0.983 \\
 & 0.9 & 0.961 & 0.960 & 0.958 & 0.957 & 0.961 & 0.961 & 0.958 & 0.956 \\
\midrule
$t_5$ & 0.0 & 1.000 & 1.000 & 1.000 & 1.000 & 1.000 & 1.000 & 1.000 & 1.000 \\
 & 0.6 & 0.983 & 0.984 & 0.983 & 0.983 & 0.984 & 0.984 & 0.983 & 0.983 \\
 & 0.9 & 0.961 & 0.960 & 0.958 & 0.957 & 0.961 & 0.961 & 0.958 & 0.956 \\
		\bottomrule
	\end{tabular}%
\end{table}%

%\begin{figure}[htbp]
%	\centering
%	\centerline{\includegraphics[width=\textwidth]{plots/comparetime.pdf}}
%	\caption{\label{fig_comparetime} 
%			Average computation time per iteration against $N$ for the proposed scalable ARMA QMLE (SARMA) and the Kalman-filter-based VARMA MLE (VARMA). The panel titles indicate different settings of $\bm{\varepsilon}_{0t}$.}
%\end{figure}

\begin{table}[htbp]
	\centering
	\caption{\label{table_preproceed} Description and transformation of the six macroeconomic indicators, where transformations are coded as : $1=$ first difference, $2=$ first difference of the log series, and $3=$ second difference of the log series.}
	\begin{tabular}{llr}
		\toprule
		\multicolumn{1}{l}{Variable}  & \multicolumn{1}{l}{Description} & \multicolumn{1}{c}{Transformation}\\
		\midrule
		IP & Industrial Production Index & 2\\
		UR & Civilian Unemployment Rate & 1\\
		CPI & Consumer Price Index & 3\\
		RW & Real wages  & 2\\
		CONS & Real Personal Consumption Expenditures & 2\\
		FFR & Federal Funds Rate & 1\\
		\bottomrule
	\end{tabular}%
\end{table}

\begin{figure}[htbp]
    \begin{center}
        \includegraphics[width=\textwidth]{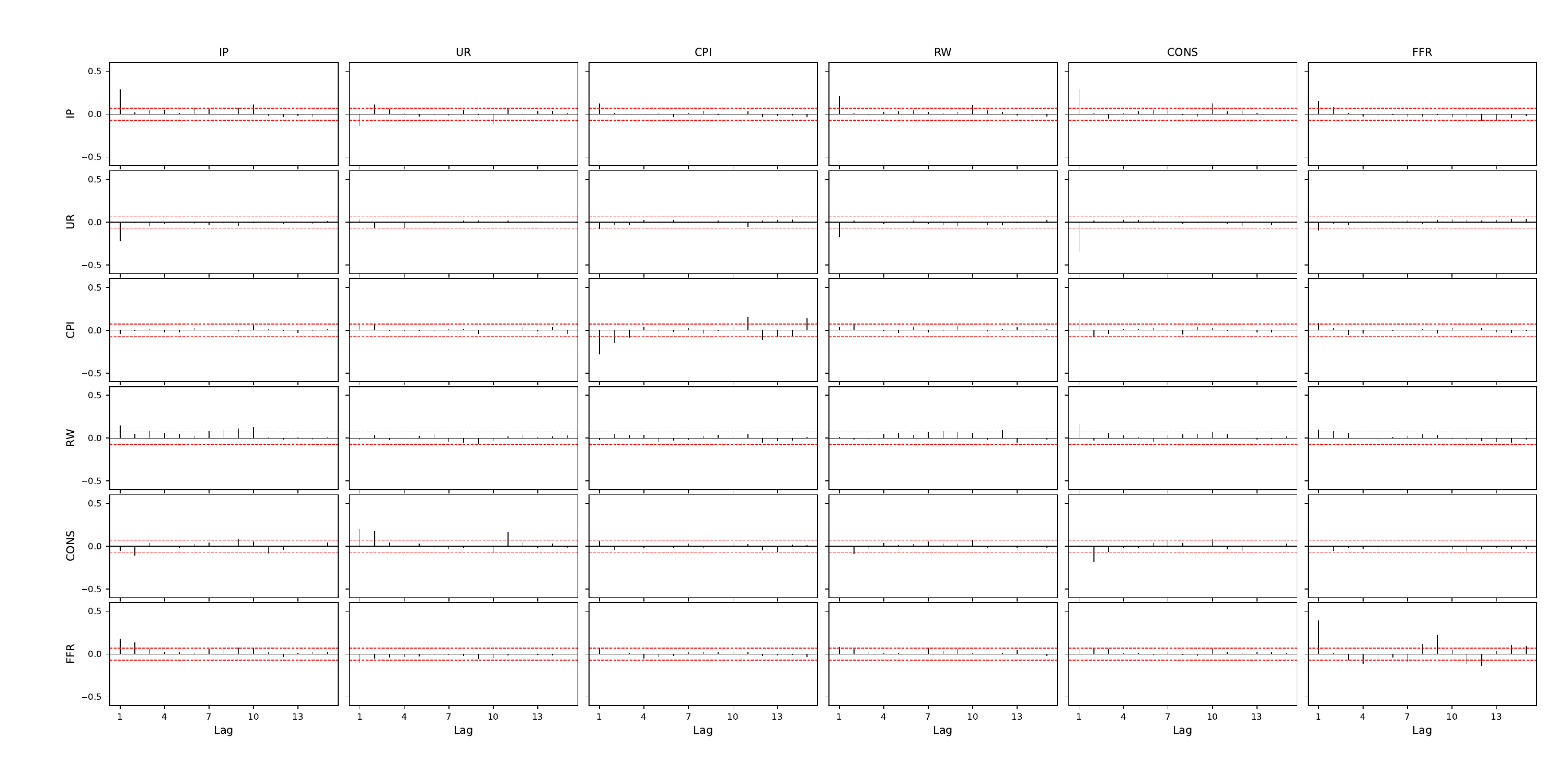}
    \end{center}
    \caption{\label{MACM_realdate} The MACM plots of $\{\bm y_t\}$ with the 95\% confidence intervals in dashed lines.}
\end{figure}
\begin{figure}[htbp]
	\centering
	\centerline{\includegraphics[width=\textwidth]{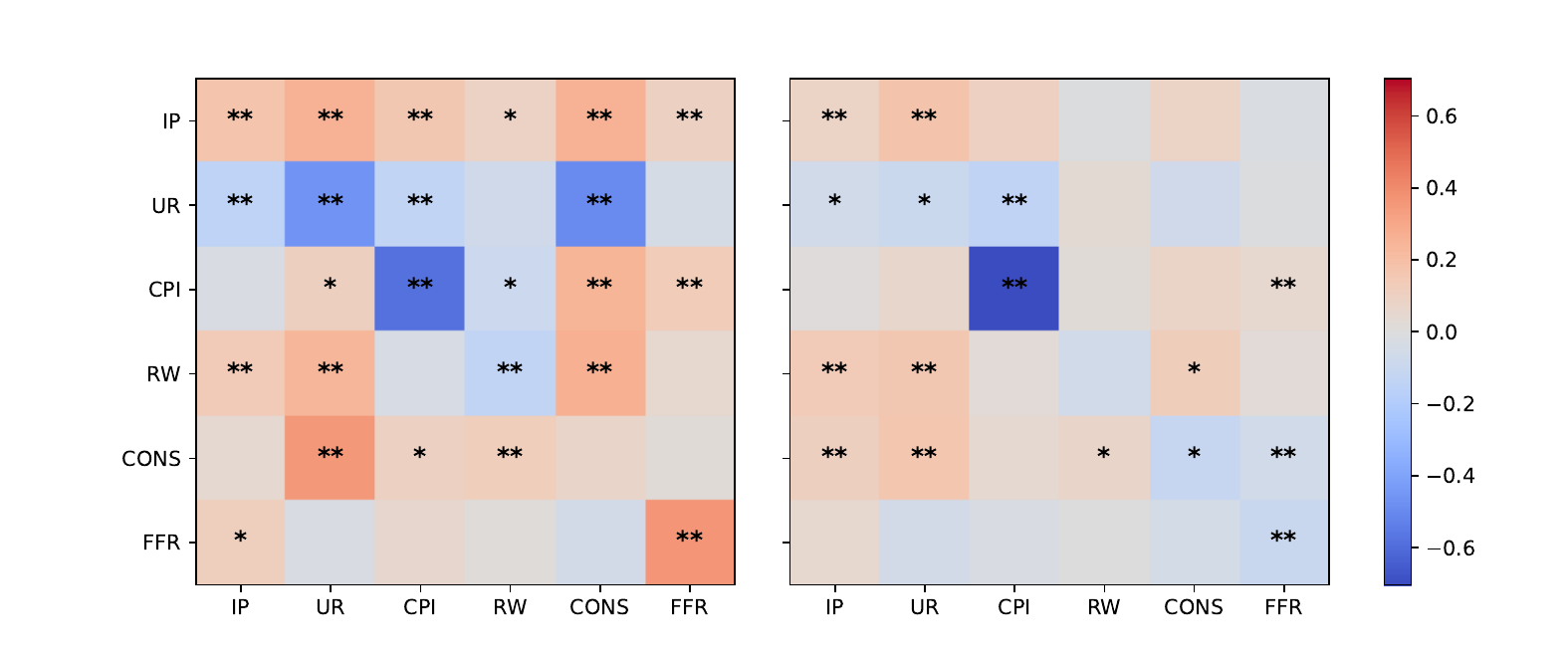}}
	\caption{\label{coef_G} Heatmap of the estimates for coefficient matrices $G_1$ (left) and $G_2$ (right) in model \eqref{model-realdata}, respectively. Statistically significant coefficients are marked with asterisks: $**$ for significance at the $1\%$ level and $*$ for significance at the $5\%$ level. The color scale indicates the magnitude of the coefficients.}
\end{figure}

\begin{table}[htbp]
	\centering
	\caption{\label{table_forecast} Relative improvement in out-of-sample prediction accuracy over the VAR benchmark for six U.S. macroeconomic indicators. SARMA$_1$ (or SARMA$_2$) represents the scalable ARMA model fitted by the QMLE (or LSE) method and VARMA$_1$ (or VARMA$_2$) represents the VARMA model fitted by the QMLE (or IOLS) method. The largest improvement in each row is marked in bold.}
	\begin{tabular}{lrrrr}
		\toprule
		& \multicolumn{1}{c}{SARMA$_1$} & \multicolumn{1}{c}{SARMA$_2$} & \multicolumn{1}{c}{VARMA$_1$} & \multicolumn{1}{c}{VARMA$_2$} \\
		\midrule
		RI-RMSFE (\%) & \textbf{7.24} & 6.72 & 5.94 & -6.98 \\
	RI-MAFE (\%) & \textbf{3.20} & 2.85 & 1.25 & -2.85 \\
		\bottomrule
	\end{tabular}%
	\label{tab:addlabel}%
\end{table}%

\newpage

\section*{Appendix}\label{appendix}
This Appendix provides technical details for Theorems \ref{thmQMLE}--\ref{thmBIC}, Propositions \ref{thmLS}--\ref{thmcompare_efficient} and Remark \ref{order-condition}, and introduces Lemmas \ref{lemma_boundness_partial}--\ref{lemma_score_asy_Norm_LS} which give some preliminary results for proving Theorem \ref{thmQMLE} and Proposition \ref{thmLS}. 
		It also includes the details on the connection between  scalable ARMA models with vector ARMA models. 
		Throughout the Appendix, 
		we denote vectors by small boldface letters $\bm{x},\bm{y},\ldots$ and matrices by capital letters $A,B,\ldots$.
		For a vector $\bm{x}$, 
		$\|\bm{x}\|_p$ denotes its $\ell_p$ norm (i.e., $\|\bm{x}\|_p=(\sum_{i}|x_i|^p)^{1/p}$). 
		For a matrix $A \in \mathbb{R}^{d_1 \times d_2}$, 
		denote by $\|A\|_{\operatorname{op}}$ its operator norm (i.e., $\|A\|_{\operatorname{op}}=\sup_{\|\bm x\|_2 =1}\|A \bm x\|$), 
		by $\|A\|$ its sub-multiplicative matrix norm (i.e., $\|AB\|\le \|A\| \|B\|$ for any matrices $A$ and $B$), 
		by $A^{\prime}$ its transpose, 
		and by $\text{vec}(A)$ its vectorization. 
		For a square matrix $B$, $|B|$, $\text{tr}(B)$ and $\text{vech}({B})$ are the determinant, trace, and the lower triangular-stacking vector of $B$, respectively. 
		For square matrices $B_j$'s, $\diag\{B_{1}, \ldots, B_{m}\}$ denotes the block diagonal matrix whose main diagonal consists of $B_{1}, \ldots, B_{m}$. 
		Denote $\otimes$ as the Kronecker product. 
		Moreover, $\bm 1_N$ is the $N$-dimensional vector with all elements being one, $0_{N\times T}$ is the $N\times T$ zero matrix and $I_N$ is the $N\times N$ identity matrix. 
		In addition, $O(1)$ denotes a bounded series of nonstochastic variables, 
		$O_p(1)$ denotes a sequence of random variables bounded in probability, 
		$o_p(1)$ denotes a sequence of random variables converging to zero in probability, 
		$\rightarrow_{p}$ denotes the convergence in probability 
		and $\rightarrow_{d}$ denotes the convergence in distribution.
\par

\setcounter{section}{0}
\setcounter{equation}{0}
\setcounter{figure}{0}
\setcounter{table}{0}
\setcounter{theorem}{0}
\setcounter{thm}{0}
\setcounter{lemma}{0}
\setcounter{corollary}{0}
\setcounter{cor}{0}
\setcounter{proposition}{0}
\setcounter{assum}{0}
\setcounter{definition}{0}
\setcounter{example}{0}
\setcounter{remark}{0}
\def\theequation{A\arabic{section}.\arabic{equation}}
\def\thesection{A\arabic{section}}
\renewcommand{\theHsection}{appendix.\arabic{section}}
\renewcommand{\theHsubsection}{appendix.\arabic{section}.\arabic{subsection}}
\renewcommand{\theHequation}{appendix.\arabic{section}.\arabic{equation}}
\renewcommand{\theHfigure}{appendix.\arabic{figure}}
\renewcommand{\theHtable}{appendix.\arabic{table}}
\renewcommand{\thetheorem}{A\arabic{theorem}}
\renewcommand{\thethm}{A\arabic{thm}}
\renewcommand{\thelemma}{A\arabic{lemma}}
\renewcommand{\thecorollary}{A\arabic{corollary}}
\renewcommand{\thecor}{A\arabic{cor}}
\renewcommand{\theproposition}{A\arabic{proposition}}
\renewcommand{\theassum}{A\arabic{assum}}
\renewcommand{\thedefinition}{A\arabic{definition}}
\renewcommand{\theexample}{A\arabic{example}}
\renewcommand{\theremark}{A\arabic{remark}}
\renewcommand{\thefigure}{A\arabic{figure}}
\renewcommand{\thetable}{A\arabic{table}}

\section{Lemmas}

\begin{lemma}\label{lemma_boundness_partial}
	Under Assumption \ref{assum_space}, there exists a constant $C>0$ such that
	\begin{align*}
		(i)& \sup_{\bm \alpha\in \mathcal{A}} \|\bm \varepsilon_t(\bm \alpha)\| \le C\xi_{{\rho},t-1} +\|\bm \varepsilon_{0t}\|;
		\quad \quad (ii) \sup_{\bm \alpha \in \mathcal{A}} \left\| \frac{\partial \bm \varepsilon_{t}^\prime(\bm \alpha)}{\partial \bm \alpha } \right\| \le C\xi_{{\rho},t-1};\\
		(iii)&\sup_{\bm \alpha\in \mathcal{A}}\left\|\frac{\partial(\operatorname{vec}(\partial \bm{\varepsilon}_t^\prime (\bm{\alpha}) /\partial \bm{\alpha}))}{\partial\bm{\alpha}^\prime}\right\| \le C \xi_{{\rho},t-p-1}, 
	\end{align*}
	where $\xi_{\rho t}= 1+\sum_{h=0}^{\infty}\rho^{h}\|\bm y_{t-h}\|$ with a constant $\rho\in(0,1)$. 
\end{lemma}
\begin{proof}
	Recall that $A_h(\bm \alpha) = G_h$ for $1\le h\le p$ and $A_h(\bm \alpha) = \sum_{k=p+1}^{d} \ell_{hk}(\bm \omega) G_k$ for $j \ge p+1$ by \eqref{model_VARinf}. 
	Then by Assumption \ref{assum_space}, there exists a positive
	constant $C$ such that $\|G_k\|<C$ for $1\le k\le d$, and then we can obtain that $\|A_h(\bm \alpha)\| = \|G_h\| \le C$ for $1\le h\le p$, and $\|A_h(\bm \alpha)\| = \sum_{k=p+1}^{d}|\ell_{hk}(\bm\omega)|\|G_k\| \le (r+2s)C {\rho}^{h-p}$ for $h \ge p+1$. 
	Let $C^\star =C  (r+2s){\rho}^{-p} $, it follows that 
	\begin{align}\label{||A_h||}
		\|A_h(\bm \alpha)\| \le C^\star {\rho}^h \;\; \text{for} \;\; h \ge 1.
	\end{align}
	
	We first show ($i$). 
	Recall that $\bm{\varepsilon}_t(\bm{\alpha}) = \bm y_t - \sum_{h=1}^{\infty}A_{h}(\bm{\alpha})\bm y_{t-h}$ and $\bm \varepsilon_{0t} = \bm y_t - \sum_{h=1}^{\infty}A_{0h}\bm y_{t-h}$. 
	Hence it holds that $\bm \varepsilon_{t}(\bm \alpha)- \bm \varepsilon_{0t} =\sum_{h=1}^{\infty} [A_{0h} - A_{h}(\bm{\alpha})] \bm y_{t-h}$. 
	Moreover, by Assumption \ref{assum_space}($i$) and \eqref{||A_h||}, it can be shown that 
	\begin{align*}
		\sup_{\bm \alpha \in \mathcal{A}}\|A_{0h} - A_{h}(\bm{\alpha})\|\le 2 \sup_{\bm \alpha\in \mathcal{A}}\|A_{h}(\bm{\alpha})\| = O({\rho}^h). 
	\end{align*}
	It then follows that 
	\begin{align*}
		\sup_{\bm \alpha\in \mathcal{A}} \|\bm \varepsilon_{t}(\bm \alpha)\|\le  
		\sup_{\bm \alpha\in \mathcal{A}}\|\bm \varepsilon_{t}(\bm \alpha)- \bm \varepsilon_{0t}\|+\|\bm \varepsilon_{0t}\| \le 
		\sum_{h=1}^{\infty } O({\rho}^h) \|\bm y_{t-h}\|  +\|\bm \varepsilon_{0t}\| \le C \xi_{{\rho},t-1}+ \|\bm \varepsilon_{0t}\|
	\end{align*}
	holds for some constant $C > 0$. Thus ($i$) holds.
	
	Next we show ($ii$). 
	Suppose that $\rho_1 < \rho <1$, 
	it can be easily shown that $h (\rho_1/\rho)^{h-1} <C$ holds for $h\ge1$ and some constant $C$. 
	It follows that 
	\begin{align*}
		1+\sum_{h=1}^{\infty} h \rho_1^{h-1}\| \bm y_{t-h}\| 
		\le 1+\sum_{h=1}^{\infty} h \left(\frac{\rho_1}{\rho}\right)^{h-1}\rho^{h-1} \| \bm y_{t-h}\|
		<C \xi_{{\rho},t-1}.
	\end{align*}
	This together with \eqref{de_epsilon_alpha} and Assumption \ref{assum_space}, implies that 
	\begin{align*}
		\sup_{\bm \alpha \in \mathcal{A}}\left\| \frac{\partial \bm \varepsilon_{t}(\bm \alpha)}{\partial \lambda_i } \right\|
		\le O(1)\xi_{{\rho},t-p-1},
		\sup_{\bm \alpha \in \mathcal{A}}\left\| \frac{\partial \bm \varepsilon_{t}(\bm \alpha)}{\partial \bm \eta_j^\prime } \right\| 
		\le O(1)\xi_{{\rho},t-p-1} \text{ and }
		\sup_{\bm \alpha \in \mathcal{A}}\left\| \frac{\partial \bm \varepsilon_{t}(\bm \alpha)}{\partial \bm g^\prime } \right\|
		\le O(1)\xi_{{\rho},t-1},
	\end{align*}
	where $1\le i \le r$ and $1\le j \le s$. 
	Then it can be shown that 
	\begin{align} \label{sup||de_epsilon_alpha||}
		\sup_{\bm \alpha \in \mathcal{A}} \left\| \frac{\partial \bm \varepsilon_{t}(\bm \alpha)}{\partial \bm \alpha^\prime } \right\|
		\le \sup_{\bm \alpha \in \mathcal{A}}\left(\sum_{i=1}^{r} \left\| \frac{\partial \bm \varepsilon_{t}(\bm \alpha)}{\partial \lambda_i } \right\|
		+\sum_{j=1}^{s} \left\| \frac{\partial \bm \varepsilon_{t}(\bm \alpha)}{\partial \bm \eta_j^\prime } \right\|
		+\left\| \frac{\partial \bm \varepsilon_{t}(\bm \alpha)}{\partial \bm g^\prime } \right\|\right)
		\le O(1)\xi_{{\rho},t-1}.
	\end{align}
	Thus ($ii$) holds.
	
	Similarly, for ($iii$), 
	suppose that $\rho_1 < \rho <1$, 
	it can be shown that $h(h-1)(\rho_1/\rho)^{h-2} <C$ holds for $h\ge1$ and some constant $C$, and it follows that 
	\begin{align*}
		1+\sum_{h=2}^{\infty} h(h-1) \rho_1^{h-2}\| \bm y_{t-h}\| 
		\le 1+\sum_{h=2}^{\infty} h (h-1)\left(\frac{\rho_1}{\rho}\right)^{h-2}\rho^{h-2} \| \bm y_{t-h}\|
		<C \xi_{{\rho},t-2}.
	\end{align*}
	This together with \eqref{de_epsilon_alpha^2} and Assumption \ref{assum_space}, implies that 
	\begin{align*}
		&\sup_{\bm \alpha \in \mathcal{A}}\left\| \frac{\partial^2 \bm \varepsilon_{t}(\bm \alpha)}{\partial \lambda_i^2 } \right\|
		\le O(1) \xi_{{\rho},t-p-2},
		\sup_{\bm \alpha \in \mathcal{A}}\left\| \frac{\partial^2 \bm \varepsilon_{t}(\bm \alpha)}{\partial \gamma_j^2 } \right\| 
		\le O(1) \xi_{{\rho},t-p-2},
		\sup_{\bm \alpha \in \mathcal{A}}\left\| \frac{\partial^2 \bm \varepsilon_{t}(\bm \alpha)}{\partial \varphi_j^2 } \right\|
		\le O(1) \xi_{{\rho},t-p-1}	,\\
		&\sup_{\bm \alpha \in \mathcal{A}}\left\| \frac{\partial^2 \bm \varepsilon_{t}(\bm \alpha)}{\partial \gamma_j \partial \varphi_j } \right\| 
		\le O(1) \xi_{{\rho},t-p-1},
		\sup_{\bm \alpha \in \mathcal{A}}\left\| \frac{\partial^2 \bm \varepsilon_{t}(\bm \alpha)}{\partial \lambda_i \partial \bm g_{p+i}^\prime } \right\|
		\le O(1) \xi_{{\rho},t-p-1},\\
		&\sup_{\bm \alpha \in \mathcal{A}}\left\| \frac{\partial^2 \bm \varepsilon_{t}(\bm \alpha)}{\partial \gamma_j \partial \bm g_{p+r+2j-1}^\prime } \right\|
		\le O(1) \xi_{{\rho},t-p-1},
		\sup_{\bm \alpha \in \mathcal{A}}\left\| \frac{\partial^2 \bm \varepsilon_{t}(\bm \alpha)}{\partial \gamma_j \partial \bm g_{p+r+2j}^\prime } \right\|
		\le O(1) \xi_{{\rho},t-p-1},\\
		&\sup_{\bm \alpha \in \mathcal{A}}\left\| \frac{\partial^2 \bm \varepsilon_{t}(\bm \alpha)}{\partial \varphi_j \partial \bm g_{p+r+2j-1}^\prime } \right\|
		\le O(1) \xi_{{\rho},t-p-1} \text{ and }
		\sup_{\bm \alpha \in \mathcal{A}}\left\| \frac{\partial^2 \bm \varepsilon_{t}(\bm \alpha)}{\partial \varphi_j \partial \bm g_{p+r+2j}^\prime } \right\|
		\le O(1) \xi_{{\rho},t-p-1},
	\end{align*}
	where $1\le i \le r$ and $1\le j \le s$. 
	Then with similar argument of \eqref{sup||de_epsilon_alpha||}, we can show that 
	\begin{align*}
		\sup_{\bm \alpha\in \mathcal{A}}\left\|\frac{\partial(\text{vec}(\partial \bm{\varepsilon}_t^\prime (\bm{\alpha}) /\partial \bm{\alpha}))}{\partial\bm{\alpha}^\prime}\right\| 
		\le O(1) \xi_{{\rho},t-p-1}.
	\end{align*}
	As a result, ($iii$) holds.
\end{proof}

\begin{lemma} \label{lemma_boundness_MLE}
	Under Assumptions \ref{assum_identifiability}--\ref{assum_space}, if $E(\|\bm{y}_t\|^2) < \infty$, it holds that 
	$$\text{ ($i$) } E\left( \sup _{\bm \theta \in \Theta}\left|l_{t}(\bm \theta)\right|\right)<\infty; 
	~~ \text{ ($ii$) } E\left( \sup_{\bm \theta \in \Theta} \left\| \dfrac{\partial l_{t}(\bm \theta)}{\partial \bm \theta } \right\|\right)<\infty;~~
	\text{ ($iii$) } E\left( \sup _{\bm \theta \in \Theta} \left\|\dfrac{\partial^{2} l_{t}(\bm \theta)}{\partial \bm \theta \partial \bm \theta^{\prime}}\right\|\right)<\infty.$$
\end{lemma}

\begin{proof}
	We first establish ($i$). 
	Recall that $l_t(\bm{\theta}) =1/2 \ln |\Sigma(\bm{\sigma})|+1/2 \bm{\varepsilon}_t^\prime (\bm{\alpha})\Sigma^{-1}(\bm{\sigma})\bm{\varepsilon}_t(\bm{\alpha})$. 
	It holds that
	\begin{align*}
		E\left(\sup_{\bm \theta\in \Theta}|l_t(\bm \theta)|\right) 
		\le \frac{1}{2} \sup_{\bm \theta\in \Theta} \left|\ln |\Sigma(\bm \sigma)| \right| + \frac{1}{2} E\left(\sup_{\bm \theta\in \Theta} \bm \varepsilon_t^\prime (\bm \alpha) \Sigma^{-1}(\bm \sigma) \bm \varepsilon_t(\bm \alpha)\right).
	\end{align*}
	Thus it suffices to show that 
	(a) $\sup_{\bm \theta\in \Theta} \left|\ln |\Sigma(\bm \sigma)| \right| < \infty$, and (b) $E(\sup_{\bm \theta\in \Theta}
	\bm \varepsilon_t^\prime (\bm \alpha) \Sigma^{-1}(\bm \sigma) \bm \varepsilon_t(\bm \alpha)) < \infty$. 
	For (a), 
	under Assumption \ref{assum_space}($ii$), it can be shown that 
	\begin{align*}
		\sup_{\bm \theta\in \Theta} \left|\ln |\Sigma(\bm \sigma)| \right| 
		\le & -\sup_{\bm \theta\in \Theta} \ln |\Sigma(\bm \sigma)|I(|\Sigma(\bm \sigma)|<1) +\sup_{\bm \theta\in \Theta} \ln |\Sigma(\bm \sigma)|I(|\Sigma(\bm \sigma)|>1)\nonumber\\
		\le & - N\ln (\min\{\underline{\sigma},1\})+N \ln (\max\{\bar{\sigma},1\}) < \infty.
	\end{align*}
	For (b), by Assumption \ref{assum_space}($ii$) and Lemma \ref{lemma_boundness_partial}($i$), we have that 
	\begin{align*}
		E\left(\sup_{\bm \theta\in \Theta} \bm \varepsilon_t^\prime (\bm \alpha) \Sigma^{-1}(\bm \sigma) \bm \varepsilon_t(\bm \alpha)\right)
		=& E\left(\sup_{\bm \theta\in \Theta} \|\bm \varepsilon_t^\prime (\bm \alpha) \Sigma^{-1}(\bm \sigma) \bm \varepsilon_t(\bm \alpha)\|\right)\\
		\le& E\left(\sup_{\bm \theta\in \Theta} \|\Sigma^{-1}(\bm \sigma)\| \|\bm \varepsilon_t(\bm \alpha)\|^2\right)\\
		=& O(1)E\left(\sup_{\bm \theta\in \Theta} \|\bm \varepsilon_t (\bm \alpha) \|^2 \right)
		\le O(1)E(\xi_{{\rho},t-1} +\|\bm \varepsilon_{0t}\|)^2,
	\end{align*}
	where $\xi_{\rho t}= 1+\sum_{j=0}^{\infty}\rho^{j}\|\bm y_{t-j}\|_2$. 
	This together with $E(\xi_{{\rho}t}^2) < \infty$ and $E(\|\bm \varepsilon_{0t}\|^2) < \infty$ by $E(\|\bm{y}_t\|^2) < \infty$ and Assumptions \ref{assum_stationarity_y}--\ref{assum_stationarity_epsilon}, implies that (b) holds. 
	Thus, ($i$) holds by (a) and (b).
	
	Next we show ($ii$). 
	Similarly, by \eqref{de_l_theta}, Lemma \ref{lemma_boundness_partial} and Assumption \ref{assum_space}, it holds that
	\begin{align*}
		\sup_{\bm \theta \in \Theta} \left\| \dfrac{\partial l_{t}(\bm \theta)}{\partial \bm \theta } \right\| 
		\le& O(1) \sup_{\bm \theta \in \Theta} \left(\left\| \frac{\partial \bm \varepsilon_{t}^\prime(\bm \alpha)}{\partial \bm \alpha } \right\| \|\Sigma^{-1}(\bm \sigma)\| \|\bm \varepsilon(\bm \alpha)\|+\|\Sigma^{-1}(\bm \sigma)\|+ \|\Sigma^{-1}(\bm \sigma)\|^2\|\bm \varepsilon_t(\bm \alpha)\|^2\right)\\
		\le & O(1)\xi_{{\rho},t-1} (\xi_{{\rho},t-1} +\|\bm \varepsilon_{0t}\|) + O(1)[1+(\xi_{{\rho},t-1}+\|\bm \varepsilon_{0t}\|)^2]\\
		\le & O(1) (\xi_{{\rho},t-1}^2 + \|\bm \varepsilon_{0t}\|^2),
	\end{align*}
	and then $E\left( \sup_{\bm \theta \in \Theta} \left\| \partial l_{t}(\bm \theta)/\partial \bm \theta  \right\|\right)<\infty$ by $E(\|\bm \varepsilon_{0t}\|^2) < \infty$ and $E(\xi_{{\rho}t}^2) < \infty$.
	Thus ($ii$) holds.
	
	Lastly we prove ($iii$). 
	With similar argument in the proof of ($i$) and ($ii$), 
	\eqref{de_l_theta^2} and Lemma \ref{lemma_boundness_partial} imply that 
	\begin{align*}
		\sup_{\bm \theta \in \Theta} \left\| \frac{\partial^2 l_t(\bm \theta)}{\partial \bm{\alpha}\partial \bm{\alpha}^\prime} \right\|
		\le & O(1)(\xi_{{\rho},t-1}^2 + \|\bm \varepsilon_{0t}\|^2),\\
		\sup_{\bm \theta \in \Theta} \left\| \frac{\partial^2 l_t(\bm \theta)}{\partial \bm{\alpha}\partial \bm{\sigma}^\prime} \right\|
		\le& O(1) (\xi_{{\rho},t-1}^2 + \|\bm \varepsilon_{0t}\|^2),\\
		\sup_{\bm \theta \in \Theta} \left\| \frac{\partial^2 l_t(\bm \theta)}{\partial \bm{\sigma}\partial \bm{\sigma}^\prime} \right\|
		\le& O(1) (\xi_{{\rho},t-1}^2 + \|\bm \varepsilon_{0t}\|^2),
	\end{align*}
	and then it holds that 
	\begin{align*}
		E\left(\sup_{\bm \theta \in \Theta} \left\|\dfrac{\partial^{2} l_{t}(\bm \theta)}{\partial \bm \theta \partial \bm \theta^{\prime}}\right\|\right)
		\le& E\left[\sup_{\bm \theta \in \Theta} \left(
						\left\| \frac{\partial^2 l_t(\bm \theta)}{\partial \bm{\alpha}\partial \bm{\alpha}^\prime} \right\|+
						\left\| \frac{\partial^2 l_t(\bm \theta)}{\partial \bm{\alpha}\partial \bm{\sigma}^\prime} \right\|+
						\left\| \frac{\partial^2 l_t(\bm \theta)}{\partial \bm{\sigma}\partial \bm{\sigma}^\prime} \right\|
						\right)\right]\\
		\le& O(1)E\left(\xi_{{\rho},t-1}^2+\|\bm \varepsilon_{0t}\|^2\right) < \infty.
	\end{align*}
	Thus ($iii$) holds.
\end{proof}

\begin{lemma}\label{lemma_three consistent_MLE}
	Under Assumptions \ref{assum_stationarity_y}--\ref{assum_space} , the following results hold:
	\begin{align*}
		&\text { ($i$) } \sup_{\bm \theta \in \Theta}\left|\dfrac{1}{T} \sum_{t=1}^T \left\{l_t(\bm \theta)-E \left[l_{t}(\bm \theta)\right]\right\}\right|=o_{p}(1); \\ 
		&\text{ ($ii$) } \sup_{\bm \theta \in \Theta}\left\|\dfrac{1}{T} \sum_{t=1}^T  \left\{\dfrac{\partial l_t(\bm \theta)}{\partial \bm \theta}-E \left[\dfrac{\partial l_{t}(\bm \theta)}{\partial \bm \theta}\right]\right\}\right\|=o_{p}(1); \\ 
		&\text { ($iii$) } \sup_{\bm \theta \in \Theta}\left\|\dfrac{1}{T} \sum_{t=1}^{T}\left\{\frac{\partial^{2} l_{t}(\bm \theta)}{\partial \bm \theta \partial \bm \theta^{\prime}}-E\left[\frac{\partial^{2} l_{t}(\bm \theta)}{\partial \bm \theta \partial \bm \theta^{\prime}}\right]\right\}\right\|=o_{p}(1).
	\end{align*}
\end{lemma}
\begin{proof}
	This lemma directly follows from Lemma \ref{lemma_boundness_MLE} and Theorem 3.1 of \cite{ling2003asymptotic}, so its proof is omitted here. 
\end{proof}

\begin{lemma}\label{lemma_initial_effect_MLE}
	Under Assumption \ref{assum_space}, the following results hold:
	\begin{align*}
		(i)& \quad \sup_{\bm \theta \in \Theta} |l_t(\bm \theta)- \widetilde{l}_t(\bm \theta)|\le O(\rho^t)\xi_{{\rho}0}(\xi_{{\rho},t-1} +\|\bm \varepsilon_{0t}\|);\\
		(ii)& \quad \sup_{\bm \theta \in \Theta} \left\|\frac{\partial l_t(\bm \theta)}{\partial \bm \theta}- \frac{\partial \widetilde{l}_t(\bm \theta)}{\partial \bm \theta}\right\|\le O(\rho^t)\xi_{{\rho}0}(\xi_{{\rho},t-1} +\|\bm \varepsilon_{0t}\|);\\
		(iii) & \quad \sup_{\bm \theta \in \Theta} \left\|\frac{\partial^2 l_t(\bm \theta)}{\partial \bm \theta \partial \bm \theta^\prime }- \frac{\partial^2 \widetilde{l}_t(\bm \theta)}{\partial \bm \theta  \partial \bm \theta^\prime}\right\|\le O(\rho^t)\xi_{{\rho}0}(\xi_{{\rho},t-1} +\|\bm \varepsilon_{0t}\|),
	\end{align*}
	where $\xi_{\rho t}= 1+\sum_{j=0}^{\infty}\rho^{j}\|\bm y_{t-j}\|_2$. 
\end{lemma}
\begin{proof}
	We first show ($i$). 
	Recall that $l_t(\bm{\theta}) =1/2\ln |\Sigma(\bm{\sigma})|+1/2 \bm{\varepsilon}_t^\prime (\bm{\alpha})\Sigma^{-1}(\bm{\sigma})\bm{\varepsilon}_t(\bm{\alpha})$. 
	Under Assumption \ref{assum_space}($ii$), it holds that 
	\begin{align*}
		\sup_{\bm \theta}|l_t(\bm \theta)- \widetilde{l}_t(\bm \theta)| 
		=&\sup_{\bm \theta\in \Theta}|\bm \varepsilon_{t}(\bm \alpha )^\prime \Sigma^{-1}(\bm \sigma ) \bm \varepsilon_{t}(\bm \alpha ) - \widetilde{\bm \varepsilon}_{t}(\bm \alpha )^\prime \Sigma^{-1}(\bm \sigma ) \widetilde{\bm \varepsilon}_{t}(\bm \alpha )|  \\
		\le&O(1)\sup_{\bm \theta\in \Theta}|\widetilde{\bm \varepsilon}_{t}(\bm \alpha )^\prime \Sigma^{-1}(\bm \sigma ) (\bm \varepsilon_t(\bm \alpha ) -\widetilde{\bm \varepsilon}_{t}(\bm \alpha ) ) |\\
		\le &O(1) \sup_{\bm \theta\in \Theta}\| \widetilde{\bm \varepsilon}_{t}(\bm \alpha )\| \sup_{\bm \theta\in \Theta}\|\bm \varepsilon_t(\bm \alpha ) -\widetilde{\bm \varepsilon}_{t}(\bm \alpha)\|.
	\end{align*}
	This together with Lemma \ref{lemma_boundness_partial}($i$) and \eqref{sup epsilon}, implies that 
	\begin{align*}
		\sup_{\bm \theta}|l_t(\bm \theta)- \widetilde{l}_t(\bm \theta)| 
		\le O(\rho^t)\xi_{{\rho}0}(\xi_{{\rho},t-1} +\|\bm \varepsilon_{0t}\|).
	\end{align*}
	Hence ($i$) holds. 

	Similarly, by \eqref{ddepsilon_initial}, \eqref{de_l_theta}--\eqref{de_l_theta^2} and Assumption \ref{assum_space}($ii$), we can obtain that ($ii$) and ($iii$) hold. 
\end{proof}

\begin{lemma}\label{lemma_identifiability}
	Under Assumptions \ref{assum_identifiability} and \ref{assum_stationarity_epsilon}, if $\bm \varepsilon_t(\bm \alpha)= \bm \varepsilon_{0t}$ a.s., then $\bm \alpha = \bm \alpha_0$. 
\end{lemma}

\begin{proof}
	Recall that $\bm{\varepsilon}_t(\bm{\alpha}) = \bm y_t - \sum_{h=1}^{\infty}A_{h}(\bm{\alpha})\bm y_{t-h}$ and $\bm \varepsilon_{0t} = \bm y_t - \sum_{h=1}^{\infty}A_{0h}\bm y_{t-h}$. 
	It holds that $\bm \varepsilon_{t}(\bm \alpha)- \bm \varepsilon_{0t} =\sum_{h=1}^{\infty} [A_{0h} - A_{h}(\bm{\alpha})] \bm y_{t-h}$. 
	Thus $\bm \varepsilon_t(\bm \alpha)= \bm \varepsilon_{0t}$ a.s. implies that  
	\begin{align}\label{epsilon_t = epsilon_0t}
		\sum_{h=1}^{\infty} (A_h(\bm \alpha)-A_{0h}) \bm y_{t-h} =0 \;\; \text{a.s.}. 
	\end{align}
	
	We first show that $A_h(\bm{\alpha}) = A_{0h}$ for all $h \ge 1$ by contradiction. 
	Suppose that there exist $k$'s such that $A_k(\bm{\alpha}) \neq A_{0k}$. 
	Note that it is obviously impossible that there is only one $k$ such that $A_{k}(\bm{\alpha}) \neq A_{0k}$ when $\bm y_{t-k} \neq 0_{N \times 1}$. 
	Then \eqref{epsilon_t = epsilon_0t} implies that $\bm y_t$ can be represented as a linear combination of $\bm y_{t-k}$'s a.s. 
	And it follows that $\bm \varepsilon_{0t} = \bm y_t -E(\bm y_t \mid \mathcal{F}_{t-1})=0_{N \times 1} $ a.s., 
	which is in contradiction with that $\var (\bm \varepsilon_{0t})=\Sigma_0$ is positive definite under Assumption \ref{assum_stationarity_epsilon}. 
	Thus we have that $A_h(\bm{\alpha}) = A_{0h}$ for all $h \ge 1$. 
	
	Recall that $A_h(\bm{\alpha}) = G_h$ for $1\le h\le p$, and $A_h(\bm{\alpha}) = \sum_{k=p+1}^{d} \ell_{hk}(\bm \omega) G_k$ for $h \ge p+1$ with $d = p + r + 2s$. 
	Since $A_h(\bm{\alpha}) = A_{0h}$ for all $h \ge 1$, we have that $G_h = G_{0h}$ for $1 \le h\le p$ and 
	\begin{align}
		&\sum_{i=1}^{r} \lambda_i^m G_{p+i} + 
		\sum_{k=1}^{s} \gamma_k^m \left[\cos(m \varphi_k) G_{p+r+2k-1} + \sin(m \varphi_k) G_{p+r+2k}\right] \nonumber\\
		=&\sum_{i=1}^{r} \lambda_{0i}^m G_{0,p+i} +
		\sum_{k=1}^{s} \gamma_{0k}^m \left[\cos(m \varphi_{0k}) G_{0,p+r+2k-1}+ \sin(m \varphi_{0k}) G_{0,p+r+2k} \right]  \label{identifi_w_g} 
	\end{align}
	for $m = h-p \ge 1$. 
	Then to prove $\bm \alpha = \bm \alpha_0$ under Assumption \ref{assum_identifiability}, it suffices to show that 
	\begin{enumerate}[($i$)]
		\item Equation \eqref{linearindependence} below holds for all $m\ge 1$ if and only if $\bm c= (c_1,\ldots,c_{r} ,c_{11},\\ c_{12}, \ldots c_{s1} ,c_{s2})^\prime= 0_{(r+2s) \times 1}$: 
		\begin{align}\label{linearindependence}
			\sum_{i=1}^{r} c_i\lambda_i^m  + 
			\sum_{k=1}^{s}[ c_{k1}\gamma_k^m \cos(m \varphi_k) + c_{k2}\gamma_k^m\sin(m \varphi_k)] =0,
		\end{align}
		where $0<|\lambda_i|<1$ and $0<\gamma_k<1$, and $\{\lambda_i\}$ and $\{\gamma_k\}$ are distinct, respectively; 
		\item Equation \eqref{identifi_phi} below holds for all $m\ge 1$ if and only if $\varphi_k = \varphi_{0k}$, $G_{p+r+2k-1} =G_{0,p+r+2k-1}$ and $G_{p+r+2k} =G_{0,p+r+2k}$: 
		\begin{align}\label{identifi_phi}
			\cos(m \varphi_k) G_{p+r+2k-1} + \sin(m \varphi_k) G_{p+r+2k}  = \cos(m \varphi_{0k}) G_{0,p+r+2k-1} + \sin(m \varphi_{0k}) G_{0,p+r+2k}
		\end{align} 
		where $\varphi_k, \varphi_{0k} \in (0, \pi)$. 
	\end{enumerate}
	Since that \eqref{identifi_w_g} and ($i$) imply that $ \lambda_i = \lambda_{0i}$ and $G_{p+i} =  G_{0,p+i}$ for $1\le i \le r$ as well as $\gamma_k = \gamma_{0k}$ for $1 \le k \le s$, then \eqref{identifi_phi} holds which together with ($ii$) implies that $\varphi_k = \varphi_{0k}$, $G_{p+r+2k-1} =G_{0,p+r+2k-1}$ and $G_{p+r+2k} =G_{0,p+r+2k}$ for $1 \le k \le s$.  

	For ($i$), obviously we only need to show that \eqref{linearindependence} holds for all $m\ge 1$ implies that $\bm c = 0$. 
	Without loss of generality, we assume that $|\lambda_1| > \cdots > |\lambda_{r}|$ and $\gamma_1 > \cdots > \gamma_{s}$. 
	Let $\tau = \max\{|\lambda_1|, \gamma_1\}$. Dividing \eqref{linearindependence} by $\tau^m$ and considering the dominant terms, we obtain
	\begin{align}\label{lambda1 + gamma1}
		c_1 \left(\frac{\lambda_1}{\tau}\right)^m 
		+ c_{11} \left(\frac{\gamma_1}{\tau}\right)^m \cos(m \varphi_1) +c_{12} \left(\frac{\gamma_1}{\tau}\right)^m  \sin(m \varphi_1)= 0.
	\end{align}
	Then we discuss \eqref{lambda1 + gamma1} in situations below: 
	\begin{enumerate}[(a)]
		\item If $\tau = |\lambda_1| = \gamma_1$, then $|c_1|= |c_{11} \cos(m \varphi_1) +c_{12}   \sin(m \varphi_1)| = \Big|\sqrt{c_{11}^2+c_{12}^2} \\\cos \left(m \varphi_1 -\arccos(\iota)\right)\Big|$ with $\iota =c_{11}/(\sqrt{c_{11}^2+c_{12}^2})$,
		and it follows that $c_1 =c_{11} =c_{12} =0$; 
		\item if $\tau = |\lambda_1| > \gamma_1$ , then $c_1 = 0$, and it follows that $c_{11} =c_{12} =0$ by $c_{11} \cos(m \varphi_1) +c_{12} \sin(m \varphi_1)= 0$; 
		\item if $\tau = \gamma_1 > |\lambda_1|$ , then $c_{11} =c_{12} =0$, and it follows that $c_1 = 0$. 
	\end{enumerate}
	Thus \eqref{lambda1 + gamma1} implies that $c_1 =c_{11} =c_{12} =0$. 
	Similarly it can be shown that the other elements of $\bm c$ are zeros. 
	It follows that ($i$) holds. 
	
	For ($ii$), using $a \cos(x) + b \sin(x) = \sqrt{a^2 + b^2} \cos(x - y)$ with $y = \arccos(a / \sqrt{a^2 + b^2})$, it can be easily shown that \eqref{identifi_phi} holds for all $m\ge 1$ if and only if $\varphi_k = \varphi_{0k}$, $G_{p+r+2k-1} =G_{0,p+r+2k-1}$ and $G_{p+r+2k} =G_{0,p+r+2k}$. 
	As a result, $\bm \alpha = \bm \alpha_0$ holds by ($i$) and ($ii$). 
\end{proof}

\begin{lemma}\label{lemma_uni_max_MLE}
	Under Assumptions \ref{assum_identifiability}, \ref{assum_stationarity_epsilon} and \ref{assum_space}, $E[l_t(\bm \theta)]$ has a unique minimum at $\bm \theta_0$.
\end{lemma}

\begin{proof}
	Recall that $l_t(\bm{\theta}) =1/2\ln |\Sigma(\bm{\sigma})|+1/2 \bm{\varepsilon}_t^\prime (\bm{\alpha})\Sigma^{-1}(\bm{\sigma})\bm{\varepsilon}_t(\bm{\alpha})$ with $\bm{\varepsilon}_t(\bm{\alpha}) = \bm y_t - \sum_{h=1}^{\infty}A_{h}(\bm{\alpha})\bm y_{t-h}$, and $\bm \varepsilon_{0t} = \bm y_t - \sum_{h=1}^{\infty}A_{0h}\bm y_{t-h}$. 
	It holds that 
	\begin{align*}
		2E[l_{t}(\bm{\theta})]
		=&\ln|\Sigma(\bm \sigma)|+E [\bm \varepsilon_t(\bm \alpha)^\prime \Sigma^{-1}(\bm \sigma) \bm \varepsilon_t(\bm \alpha)]\\
		=&\ln|\Sigma(\bm \sigma)|+E[(\bm \varepsilon_t(\bm \alpha)-\bm \varepsilon_{0t})^\prime \Sigma^{-1}(\bm \sigma) (\bm \varepsilon_t(\bm \alpha)-\bm \varepsilon_{0t})]\\
		&+E[\bm \varepsilon_{0t}^\prime \Sigma^{-1}(\bm \sigma) \bm\varepsilon_{0t}]
		+2E[(\bm \varepsilon_t(\bm \alpha)-\bm \varepsilon_{0t})^\prime \Sigma^{-1}(\bm \sigma) \bm\varepsilon_{0t}]. 
	\end{align*}
	Note that $\bm \varepsilon_t(\bm \alpha)- \bm \varepsilon_{0t}= - \sum_{h=1}^{\infty} (A_h(\bm{\alpha})-A_{0h}) \bm y_{t-h} \in \mathcal{F}_{t-1}$ with $\mathcal{F}_{t}$ being the $\sigma$-field generated by $\{\bm \varepsilon_{0s}, s \le t\}$, and $E(\bm \varepsilon_{0t}\mid \mathcal{F}_{t-1}) = 0$ under Assumption \ref{assum_stationarity_epsilon}. 
	Then using the law of iterated expectation, we have that 
	\begin{align*}
		2E[l_{t}(\bm{\theta})]
		=&\left\{E[(\bm \varepsilon_t(\bm \alpha)-\bm \varepsilon_{0t})^\prime\Sigma^{-1}(\bm \sigma)(\bm \varepsilon_t(\bm \alpha)-\bm \varepsilon_{0t})]\right\}
		+\left\{\ln|\Sigma(\bm \sigma)|+E[\bm \varepsilon_{0t}^\prime \Sigma^{-1}(\bm \sigma) \bm \varepsilon_{0t}]\right\} \\
		:=& Q_1(\bm \theta)+Q_2(\bm \sigma).
	\end{align*}
	By Lemma \ref{lemma_identifiability}, it can be easily shown that $Q_1(\bm \theta)$ has a unique minimum at $\bm \theta_0$. 
	Thus we are left to show that $Q_2(\bm \sigma)$ has a unique minimum at $\bm \sigma_0$. 
	Rewrite $Q_2(\bm \sigma)$ as follows, 
	\begin{align*}
		Q_2(\bm \sigma) = -\ln|\Sigma_0 \Sigma^{-1}(\bm \sigma)|+\text{tr}(\Sigma_0 \Sigma^{-1}(\bm \sigma) ) + \ln|\Sigma_0|.
	\end{align*}
	By Lemma A.6 of \cite{johansen1995likelihood}, $-\ln|M|+\text{tr}M\ge N$ holds for any $N$-dimensional positive definite matrix $M$, and the equality holds if and only if $M$ is an identity matrix. 
	Note that under Assumptions \ref{assum_stationarity_epsilon} and \ref{assum_space}, $\Sigma_0 \Sigma^{-1}(\bm \sigma)$ is positive definite. 
	Thus $Q_2(\bm \sigma) \ge N+\ln|\Sigma_0|$ and the equality holds if and only if $\Sigma_0 \Sigma^{-1}(\bm \sigma) = I_N$, 
	which implies that $Q_2(\bm \sigma)$ has a unique minimum at $\bm \sigma_0$.
	As a result, $E[l_t(\bm \theta)]$ is uniquely minimized at $\bm \theta_0$. 
\end{proof}

\begin{lemma} \label{lemma_postive_J_MLE}
	Under Assumptions \ref{assum_identifiability} and \ref{assum_stationarity_epsilon}, $\mathcal{J}$ is positive definite.
\end{lemma}

\begin{proof}
	To show $\mathcal{J}$ is positive definite, it suffices to verify that $\bm c' \mathcal{J} \bm c>0$ holds for any nonzero constant vector $\bm c \in \mathbb{R}^{n}$.
	Recall that $\mathcal{J} = E(P_{0t}^\prime S_2 P_{0t})$ in \eqref{I_2 and J_2} with 
	\begin{align*}
		S_2 = \left(
		\begin{matrix}
			\Sigma_0^{-1} & 0_{N \times N^2}\\
			* & \frac{1}{2} \mathcal{P}
		\end{matrix}
		\right),
		\;\;
		P_{0t} = \left(
		\begin{matrix}
			\frac{\partial \bm \varepsilon_t (\bm \alpha_0)}{\partial \bm{\alpha}^\prime } & 0_{N \times N(N+1)/2  }\\
			0_{N^2 \times n_\alpha} & D
		\end{matrix}
		\right), 
	\end{align*}
	$\mathcal{P} = \Sigma_0^{-1} \otimes \Sigma_0^{-1}$ and $D$ being a duplication matrix. 
	Note that $\Sigma_0^{-1}$ is positive definite under Assumptions \ref{assum_stationarity_epsilon}. 
	Hence $S_2$ is positive definite, which implies that $\bm c' S_2 \bm c>0$ holds for any nonzero constant vector $\bm c \in \mathbb{R}^{n}$. 
	It follows that $\bm c' \mathcal{J} \bm c>0$ holds for any $\bm c$ that satisfies $P_{0t} \bm c\neq 0_{(N+N^2) \times 1}$ a.s. 
	Then we only need to show that $P_{0t} \bm c = 0_{(N+N^2) \times 1}$ holds a.s. if and only if $\bm c = 0_{n \times 1}$. 

	Suppose that $P_{0t}\bm c =0_{(N+N^2) \times 1}$ a.s. with $\bm c= (\bm c_1^\prime, \bm c_2^\prime)^\prime$, $\bm c_1 \in \mathbb{R}^{n_\alpha}$ and $\bm c_2 \in \mathbb{R}^{N(N+1)/2}$.
	It follows that 
	\begin{align*}
		\frac{\partial \bm \varepsilon_{0t}}{\partial \bm \alpha^\prime} \bm c_1 =0_{N \times 1} \;\; \text{a.s.} 
		\quad \text{and} \quad 
		D \bm c_2 =0_{N^2 \times 1}.
	\end{align*}
	Recall that $D$ is a duplication matrix. 
	Since $D$ has full column rank, $D \bm c_2 =0_{N^2 \times 1}$ implies that $\bm c_2 =0_{N(N+1)/2 \times 1}$.
	Moreover, by the proof of Lemma \ref{lemma_postive_J_LS}, $(\partial \bm \varepsilon_{0t}/ \partial \bm \alpha') \bm c_1 = 0_{N \times 1}$ a.s. implies that $\bm c_1 = 0_{n_\alpha \times 1}$. 
	Thus $P_{0t} \bm c = 0_{(N+N^2) \times 1}$ holds a.s. if and only if $\bm c = 0_{n \times 1}$. 
	As a result, $\mathcal{J}$ is positive definite.
\end{proof}

\begin{lemma} \label{lemma_score_asy_Norm_MLE}
	Suppose that $\{\bm y_t\}$ is generated by model \eqref{model_VARinf} and Assumptions \ref{assum_identifiability}--\ref{assum_space} hold. 
	If $E(\|\bm \varepsilon_{0t}\|^4)<\infty$, $E(\|\bm y_t\|^2)<\infty$ and  $\mathcal{I}=E[\partial l_{t}(\bm \theta_0)/ \partial \bm \theta\partial l_{t}(\bm \theta_0)/ \partial \bm \theta']$ is positive definite, then it holds that  
	\begin{align*}
		\frac{1}{\sqrt{T}} \sum_{t=1}^{T} \frac{\partial \widetilde{l}_{t}(\bm \theta_0)}{\partial \bm \theta} \rightarrow_{d} N(0_{n\times 1},\mathcal{I}) \quad \text{as} \quad T\to \infty.
	\end{align*}
\end{lemma}

\begin{proof}
	By \eqref{de_l_theta}, Assumption \ref{assum_space} and Lemma \ref{lemma_boundness_partial}($ii$), it can be shown that 
	\begin{align*}
		E\left[\left\|\frac{\partial l_t (\bm{\theta}_0)}{\partial \bm{\theta} } \frac{\partial l_t (\bm{\theta}_0)}{\partial \bm{\theta}^\prime}\right\|\right]
		\le& E\left[\left\| \frac{\partial l_{t}(\bm \theta_0)}{\partial \bm \alpha} \right\|^2 
		+\left\| \frac{\partial l_{t}(\bm \theta_0)}{\partial \bm \sigma} \right\|^2 \right]\\
		\le& O(1)E\left[\left\| \frac{\partial \bm \varepsilon_{t}(\bm \alpha_0)}{\partial \bm \alpha} \right\|^2\|\bm \varepsilon_{0t}\|^2 
						+ \left(1+\|\bm \varepsilon_{0t}\|^2 \right)^2\right]\\
		\le& O(1)E\left[\xi_{{\rho},t-1}^2\|\bm \varepsilon_{0t}\|^2 
						+ \left(1+\|\bm \varepsilon_{0t}\|^2 \right)^2\right],
	\end{align*}
	where $\xi_{\rho t}= 1+\sum_{h=0}^{\infty}\rho^{h}\|\bm y_{t-h}\|_2$. 
	This together with $E(\|\bm \varepsilon_{0t}\|^4)<\infty$ and $E\left(\xi_{{\rho}t}^2\right) < \infty$ by $E(\|\bm y_t\|^2)<\infty$, Assumptions \ref{assum_stationarity_y} and \ref{assum_space}($ii$), implies that $\mathcal{I}$ is finite.
    
	Let $S_{  2} = \sum_{t=1}^T S_{2t}$ with $S_{2t} =\bm c'\partial l_{t}(\bm \theta_0)/\partial \bm \theta$, where $\bm c$ is a constant vector with the same dimension as $\bm \theta$. 
	Note that by Assumptions \ref{assum_identifiability}--\ref{assum_stationarity_epsilon}, $E(\bm \varepsilon_{0t} \mid \mathcal{F}_{t-1})=\bm 0$, $E(\bm \varepsilon_{0t}\bm \varepsilon_{0t}^\prime \mid \mathcal{F}_{t-1})=\Sigma_0$ \eqref{de_epsilon_alpha} and \eqref{de_l_theta}, 
	$\{S_{2t}\}$ is ergodic and stationary, 
	$S_{2t}$ is measurable with respect to $\mathcal{F}_{t}$ that is the $\sigma$-field generated by $\{\bm \varepsilon_{0s}, s \le t\}$, and 
	$E(S_{2t}\mid \mathcal{F}_{t-1})=0$. 
	Moreover, since $\mathcal{I}$ is finite, we can obtain that 
	\begin{align*}
		\var(S_{2t})= \bm c' E[\partial l_{t}(\bm \theta_0)/\partial \bm \theta \partial l_{t}(\bm \theta_0)/\partial \bm \theta']\bm c = \bm c'\mathcal{I} \bm c < \infty.
	\end{align*}
	Thus $\{S_{2t}, \mathcal{F}_{t}\}_t$ is an ergodic, stationary and square integrable martingale difference. 
	By the martingale central limit theorem for a stationary and ergodic sequence of square integrable martingale increments\citep{Billingsley1961lindeberg}, we have that 
	\begin{align*}
		\frac{1}{\sqrt{T}} S_{  2}\rightarrow_{d} N(0,\bm c' \mathcal{I} \bm c) \;\; \text{as} \;\; T \to \infty. 
	\end{align*}
	This together with the Cram\'{e}r-Wold device, implies that 
	\begin{align}\label{sum dlt convergence}
		\frac{1}{\sqrt{T}} \sum_{t=1}^{T} \frac{\partial l_{t}(\bm \theta_0)}{\partial \bm \theta} \rightarrow_d N(0_{n\times 1}, \mathcal{I}) \;\; \text{as} \;\; T \to \infty. 
	\end{align}
	Furthermore, by Lemma \ref{lemma_initial_effect_MLE}($ii$), it holds that 
	\begin{align*}
		\left\|\frac{1}{\sqrt{T}}\sum_{t=1}^{T} \left[\frac{\partial l_t(\bm \theta_0)}{\partial \bm \theta}- \frac{\partial \widetilde{l}_t(\bm \theta_0)}{\partial \bm \theta}\right]\right\| 
		\leq \frac{1}{\sqrt{T}}\sum_{t=1}^{T} O(\rho^t)\xi_{{\rho}0}(\xi_{{\rho},t-1} +\|\bm \varepsilon_{0t}\|), 
	\end{align*}
	where $\xi_{\rho t}= 1+\sum_{j=0}^{\infty}\rho^{j}\|\bm y_{t-j}\|_2$. 
	Then by $0 < {\rho} <1$ under Assumption \ref{assum_space}($ii$), $E\left(\xi_{{\rho}t}^2\right) < \infty$ and Markov's inequality, it can be shown that 
	\begin{align}\label{o_1_de_l}
		\left\|\frac{1}{\sqrt{T}}\sum_{t=1}^{T} \left[\frac{\partial l_t(\bm \theta_0)}{\partial \bm \theta}- \frac{\partial \widetilde{l}_t(\bm \theta_0)}{\partial \bm \theta}\right]\right\|= O_p\left(\frac{1}{\sqrt{T}}\right) .
	\end{align}
	Thus by \eqref{sum dlt convergence}--\eqref{o_1_de_l}, we can obtain that 
	\begin{align*}
		\frac{1}{\sqrt{T}} \sum_{t=1}^{T} \frac{\partial \widetilde{l}_{t}(\bm \theta_0)}{\partial \bm \theta} \rightarrow_d N(0_{n\times 1}, \mathcal{I}) \;\; \text{as} \;\; T \to \infty.
	\end{align*}
\end{proof}

\begin{lemma} \label{lemma_boundness_LS}
	Under Assumptions \ref{assum_stationarity_y}--\ref{assum_space}, it holds that 
	\begin{align*}
		(i)& E\left( \sup _{\bm \alpha \in \mathcal{A}} |\bm \varepsilon_{t}^\prime (\bm\alpha)\bm \varepsilon_{t}(\bm\alpha)|\right)<\infty; 
		\quad
		(ii) E\left( \sup_{\bm \alpha \in \mathcal{A}} \left\| \dfrac{\partial \bm \varepsilon_{t}^\prime(\bm \alpha)}{\partial \bm \alpha } \bm \varepsilon_t(\bm \alpha) \right\|\right)<\infty;\\
		(iii)& E\left( \sup _{\bm \alpha \in \mathcal{A}} \left\| (\bm \varepsilon_t^\prime (\bm \alpha)\otimes I_{n_\alpha})\dfrac{\partial \rm{vec}(\partial \bm \varepsilon_{t}^\prime (\bm \alpha)/ \partial \bm \alpha)}{\partial \bm \alpha^\prime}
		+ \frac{\partial \bm \varepsilon_t(\bm \alpha)}{\partial \bm \alpha} \frac{\partial \bm \varepsilon_t(\bm \alpha)}{\partial \bm \alpha^\prime}  \right\|\right)<\infty.
	\end{align*}
\end{lemma}

\begin{proof}
	We first show ($i$). 
	Under Assumption \ref{assum_space}, together with Lemma \ref{lemma_boundness_partial}($i$), it holds that 
	\begin{align*}
		\sup _{\bm \alpha \in \mathcal{A}} |\bm \varepsilon_{t}^\prime (\bm\alpha)\bm \varepsilon_{t}(\bm\alpha)| 
		\leq \sup_{\bm \alpha\in \mathcal{A}} \|\bm \varepsilon_t^\prime (\bm \alpha) \|^2 
		\leq O(1)\left(\xi_{{\rho},t-1} +\|\bm \varepsilon_{0t}\|\right)^2. 
	\end{align*}
	Note that by Assumption \ref{assum_stationarity_y} that $E(\|\bm y_t\|_2^2) < \infty$, we have $E(\xi_{{\rho}t}^2) < \infty$. 
	These together with $E(\|\bm \varepsilon_{0t}\|^2) < \infty$ by Assumption \ref{assum_stationarity_epsilon} imply that ($i$) holds. 
	
	For ($ii$), by Lemmas \ref{lemma_boundness_partial}($i$)--($ii$), we have that
	\begin{align*}
		&\sup_{\bm \alpha \in \mathcal{A}} \left(\left\| \frac{\partial \bm \varepsilon_{t}^\prime(\bm \alpha)}{\partial \bm \alpha } \bm \varepsilon_t(\bm \alpha)\right\|\right)
		<\sup_{\bm \alpha \in \mathcal{A}} \left(\left\| \frac{\partial \bm \varepsilon_{t}^\prime(\bm \alpha)}{\partial \bm \alpha } \right\|  \|\bm \varepsilon_t(\bm \alpha)\|\right)\\
		\le & O(1)\xi_{{\rho},t-1} (\xi_{{\rho},t-1} +\|\bm \varepsilon_{0t}\|)
		\le O(1) (\xi_{{\rho},t-1}^2 + \|\bm \varepsilon_{0t}\|^2).
	\end{align*}
	Similarly, ($ii$) holds by $E(\xi_{{\rho}t}^2) < \infty$ and $E(\|\bm \varepsilon_{0t}\|^2) < \infty$. 
	
	For ($iii$), under Assumption \ref{assum_space}, by Lemma \ref{lemma_boundness_partial}, it can be shown that 
	\begin{align*}
		&\sup _{\bm \alpha \in \mathcal{A}} \left\| (\bm \varepsilon_t^\prime (\bm \alpha)\otimes I_{n_\alpha}) \dfrac{\partial \text{vec}(\partial \bm \varepsilon_{t}^\prime (\bm \alpha)/ \partial \bm \alpha)}{\partial \bm \alpha^\prime}
		+ \frac{\partial \bm \varepsilon_t(\bm \alpha)}{\partial \bm \alpha} \frac{\partial \bm \varepsilon_t(\bm \alpha)}{\partial \bm \alpha^\prime}  \right\| \\
		\le & O(1) \sup_{\bm \alpha \in \mathcal{A}} \left(\|\bm \varepsilon_{t}(\bm \alpha)\| \left\| \frac{\partial(\text{vec}(\partial \bm{\varepsilon}_t^\prime (\bm{\alpha}) /\partial \bm{\alpha}))}{\partial\bm{\alpha}^\prime} \right\|
		+\left\|\frac{\partial \bm \varepsilon_t^\prime (\bm \alpha)}{\partial \bm \alpha}\right\| \left\|\frac{\partial \bm \varepsilon_t(\bm \alpha)}{\partial \bm \alpha^\prime} \right\| \right)\\
		\le & O(1)\left[ (\xi_{{\rho},t-1} +\|\bm \varepsilon_{0t}\|)\xi_{{\rho},t-p-1}  + \xi_{{\rho},t-1}^2\right]\\
		\le & O(1)(\xi_{{\rho},t-1}^2 + \|\bm \varepsilon_{0t}\|^2).
	\end{align*}
	Then ($iii$) holds by $E(\xi_{{\rho}t}^2) < \infty$ and $E(\|\bm \varepsilon_{0t}\|^2) < \infty$. 
\end{proof}

\begin{lemma}\label{lemma_three consistent_LS}
	Under Assumptions \ref{assum_stationarity_y}--\ref{assum_space}, the following results hold:
	\begin{align*}
		(i)& \sup_{\bm \alpha \in \mathcal{A}}\left|\dfrac{1}{T} \sum_{t=1}^T \left\{\bm \varepsilon_{t}^\prime (\bm\alpha)\bm \varepsilon_{t}(\bm\alpha)-E \left[\bm \varepsilon_{t}^\prime (\bm\alpha)\bm \varepsilon_{t}(\bm\alpha)\right]\right\}\right|=o_{p}(1); \\ 
		(ii)& \sup_{\bm \alpha \in \mathcal{A}}\left\|\dfrac{1}{T} \sum_{t=1}^T \left\{\frac{\partial \bm \varepsilon_{t}^\prime(\bm \alpha)}{\partial \bm \alpha }\bm \varepsilon_{t}(\bm\alpha) -E \left[\frac{\partial \bm \varepsilon_{t}^\prime(\bm \alpha)}{\partial \bm \alpha }\bm \varepsilon_{t}(\bm\alpha)\right]\right\}\right\|=o_{p}(1); \\ 
		(iii)& \sup_{\bm \alpha \in \mathcal{A}}\left\|\frac{1}{T}\sum_{t=1}^T  \left\{\left[
		(\bm \varepsilon_t^\prime (\bm \alpha)\otimes I_{n_\alpha})\frac{\partial {\rm vec}(\partial \bm \varepsilon^\prime_{t} (\bm \alpha)/ \partial \bm \alpha)}{\partial \bm \alpha^\prime}
		+ \frac{\partial \bm \varepsilon_t(\bm \alpha)}{\partial \bm \alpha} \frac{\partial \bm \varepsilon_t(\bm \alpha)}{\partial \bm \alpha^\prime}  
		 \right] \right.\right.\\
		&\quad \quad \quad\quad \quad \left.\left.-E\left[(\bm \varepsilon_t^\prime (\bm \alpha)\otimes I_{n_\alpha})\dfrac{\partial {\rm vec}(\partial \bm \varepsilon_{t}^\prime (\bm \alpha)/ \partial \bm \alpha)}{\partial \bm \alpha^\prime}
		+ \frac{\partial \bm \varepsilon_t(\bm \alpha)}{\partial \bm \alpha} \frac{\partial \bm \varepsilon_t(\bm \alpha)}{\partial \bm \alpha^\prime}  \right]\right\}\right\|=o_{p}(1).
	\end{align*}
\end{lemma}
\begin{proof}
	This lemma directly follows from Lemma \ref{lemma_boundness_LS} and Theorem 3.1 of \cite{ling2003asymptotic}, so its proof is omitted here. 
\end{proof}

\begin{lemma}\label{lemma_initial_effect_LS}
	Under Assumption \ref{assum_space}, the following results hold:
	\begin{align*}
		(i)& \quad \sup_{\bm \alpha \in \mathcal{A}} |\bm \varepsilon_{t}^\prime (\bm\alpha)\bm \varepsilon_{t}(\bm\alpha)- \widetilde{\bm \varepsilon}_{t}^\prime (\bm\alpha )\widetilde{\bm \varepsilon}_{t}(\bm\alpha)|\le O(\rho^t)\xi_{{\rho}0}(\xi_{{\rho},t-1} +\|\bm \varepsilon_{0t}\|);\\
		(ii)& \quad \sup_{\bm \alpha \in \mathcal{A}} \left\|\frac{\partial \bm \varepsilon_{t}^\prime(\bm \alpha)}{\partial \bm \alpha }\bm \varepsilon_{t}(\bm\alpha)-\frac{\partial \widetilde{\bm \varepsilon}_{t}^\prime(\bm \alpha)}{\partial \bm \alpha }\widetilde{\bm \varepsilon}_{t}(\bm\alpha)\right\|\le O(\rho^t)\xi_{{\rho}0}(\xi_{{\rho},t-1} +\|\bm \varepsilon_{0t}\|);\\
		(iii) & \quad \sup_{\bm \alpha \in \mathcal{A}} \left\|\left[ (\bm \varepsilon_t^\prime (\bm \alpha)\otimes I_{n_\alpha})
		\dfrac{\partial  {\rm  vec}(\partial \bm \varepsilon_{t}^\prime (\bm \alpha)/ \partial \bm \alpha)}{\partial \bm \alpha^\prime}\bm \varepsilon_t(\bm \alpha)
		+ \frac{\partial \bm \varepsilon_t(\bm \alpha)}{\partial \bm \alpha} \frac{\partial \bm \varepsilon_t(\bm \alpha)}{\partial \bm \alpha^\prime}  
		\right]
		\right. \\
		&\left. \quad \quad \quad-\left[
		(\widetilde{\bm \varepsilon}_t^\prime (\bm \alpha)\otimes I_{n_\alpha}) \frac{\partial {\rm vec} (\partial \widetilde{\bm \varepsilon}_{t}^\prime (\bm \alpha)/ \partial \bm \alpha)}{\partial \bm \alpha^\prime} 
		+ \frac{\partial \widetilde{\bm \varepsilon}_t(\bm \alpha)}{\partial \bm \alpha} \frac{\partial \widetilde{\bm \varepsilon}_t(\bm \alpha)}{\partial \bm \alpha^\prime}  
		\right] \right\|\le O(\rho^t)\xi_{{\rho}0}(\xi_{{\rho},t-1} +\|\bm \varepsilon_{0t}\|),
	\end{align*}
	where $\xi_{\rho t}= 1+\sum_{j=0}^{\infty}\rho^{j}\|\bm y_{t-j}\|_2$. 
\end{lemma}
\begin{proof}
	We first show ($i$). 
	Recall that $\bm{\varepsilon}_t(\bm{\alpha}) = \bm y_t - \sum_{h=1}^{\infty}A_{h}(\bm{\alpha})\bm y_{t-h}$ and $\widetilde{\bm{\varepsilon}}_t(\bm{\alpha}) =\bm{y}_t-\sum_{h=1}^{t-1}A_{h}(\bm{\alpha}) \bm{y}_{t-h}$. 
	By \eqref{||A_h||} under Assumptions \ref{assum_space}($ii$)--($iii$), it can be shown that 
	\begin{equation}\label{sup epsilon}
		\begin{aligned}
			&\sup_{\bm \alpha\in \mathcal{A}} \|\bm \varepsilon_{t}(\bm \alpha) - \widetilde{\bm \varepsilon}_{t}(\bm \alpha) \| \le 
			\sup_{\bm \alpha\in \mathcal{A}} \sum_{h=t}^{\infty}\|A_h(\bm{\alpha})\|\|\bm y_{t-h}\|_2 \le
			O(\rho^t) \xi_{{\rho}0}. 
		\end{aligned}
	\end{equation}
	This together with Lemma \ref{lemma_boundness_partial}($i$) implies that 
	\begin{align*}
		&\sup_{\bm \alpha \in \mathcal{A}} |\bm \varepsilon_{t}^\prime (\bm\alpha)\bm \varepsilon_{t}(\bm\alpha)- \widetilde{\bm \varepsilon}_{t}^\prime (\bm\alpha )\widetilde{\bm \varepsilon}_{t}(\bm\alpha)| \\
		\le& \sup_{\bm \alpha \in \mathcal{A}} |[\bm \varepsilon_{t}(\bm\alpha) -\widetilde{\bm \varepsilon}_{t}(\bm\alpha)]^\prime \bm \varepsilon_{t}(\bm\alpha)|
		+\sup_{\bm \alpha \in \mathcal{A}} | \widetilde{\bm \varepsilon}_{t}^\prime (\bm\alpha )[\bm \varepsilon_{t}(\bm\alpha)-\widetilde{\bm \varepsilon}_{t}(\bm\alpha)]|
		\le O(\rho^t)\xi_{{\rho}0}(\xi_{{\rho},t-1} +\|\bm \varepsilon_{0t}\|). 
	\end{align*}
	Hence ($i$) holds. 
	
	Similarly, by \eqref{de_epsilon_alpha}--\eqref{de_epsilon_alpha^2} and Assumptions \ref{assum_space}($ii$)--($iii$), it can be shown that 
	\begin{align}
		&\sup_{\bm \alpha \in \mathcal{A}} \left\| \frac{\partial \bm \varepsilon_{t}(\bm \alpha)}{\partial \bm{\alpha}^\prime} 
		-\frac{\partial \widetilde{\bm \varepsilon}_{t}(\bm \alpha )}{\partial \bm{\alpha}^\prime}\right\| 
		\le O(\rho^{t-p-1})\xi_{{\rho}0} \quad \text{and} \quad\nonumber\\
		&\sup_{\bm \alpha \in \mathcal{A}} \left\|\frac{\partial(\text{vec}(\partial \bm{\varepsilon}_t^\prime (\bm{\alpha}) /\partial \bm{\alpha}))}{\partial\bm{\alpha}^\prime}
		-\frac{\partial(\text{vec}(\partial \widetilde{\bm \varepsilon}_t^\prime (\bm{\alpha}) /\partial \bm{\alpha}))}{\partial\bm{\alpha}^\prime}
		\right\| 
		\le O(\rho^{t-p-2}) \xi_{{\rho}0}.\label{ddepsilon_initial}
	\end{align}
	These together with Lemma \ref{lemma_boundness_partial}, imply that ($ii$)-($iii$) hold. 
\end{proof}

\begin{lemma}\label{lemma_uni_max}
	Under Assumptions \ref{assum_identifiability} and \ref{assum_stationarity_epsilon}, $E[\bm \varepsilon_t^\prime (\bm \alpha) \bm \varepsilon_{t}(\bm \alpha)]$ has a unique minimum at $\bm \alpha_0$.
\end{lemma}

\begin{proof}
	Recall that $\bm{\varepsilon}_t(\bm{\alpha}) = \bm y_t - \sum_{h=1}^{\infty}A_{h}(\bm{\alpha})\bm y_{t-h}$ and $\bm \varepsilon_{0t} = \bm y_t - \sum_{h=1}^{\infty}A_{0h}\bm y_{t-h}$. 
	It holds that 
	\begin{align*}
		E[\bm \varepsilon_t^\prime (\bm \alpha) \bm \varepsilon_{t}(\bm \alpha)]
		=&E[(\bm \varepsilon_t(\bm \alpha)-\bm \varepsilon_{0t})^\prime  (\bm \varepsilon_t(\bm \alpha)-\bm \varepsilon_{0t})]+E[\bm \varepsilon_{0t}^\prime \bm\varepsilon_{0t}]
		+2E[(\bm \varepsilon_t(\bm \alpha)-\bm \varepsilon_{0t})^\prime \bm\varepsilon_{0t}].
	\end{align*}
	Note that $\bm \varepsilon_t(\bm \alpha)- \bm \varepsilon_{0t}= - \sum_{h=1}^{\infty} (A_h(\bm{\alpha})-A_{0h}) \bm y_{t-h} \in \mathcal{F}_{t-1}$ with $\mathcal{F}_{t}$ being the $\sigma$-field generated by $\{\bm \varepsilon_{0s}, s \le t\}$, and $E(\bm \varepsilon_{0t}\mid \mathcal{F}_{t-1}) = 0$ under Assumption \ref{assum_stationarity_epsilon}. 
	Then using the law of iterated expectation, we have that 
	\begin{align*}
		E[\bm \varepsilon_t^\prime (\bm \alpha) \bm \varepsilon_{t}(\bm \alpha)]
		=&E[(\bm \varepsilon_t(\bm \alpha)-\bm \varepsilon_{0t})^\prime(\bm \varepsilon_t(\bm \alpha)-\bm \varepsilon_{0t})]
		+E[\bm \varepsilon_{0t}^\prime  \bm \varepsilon_{0t}].
	\end{align*}
	By Lemma \ref{lemma_identifiability}, it can be easily shown that $E[(\bm \varepsilon_t(\bm \alpha)-\bm \varepsilon_{0t})^\prime(\bm \varepsilon_t(\bm \alpha)-\bm \varepsilon_{0t})]$ has a unique minimum at $\bm \alpha_0$.
	It then follows that $E[\bm \varepsilon_t^\prime (\bm \alpha) \bm \varepsilon_{t}(\bm \alpha)]$ has a unique minimum at $\bm \alpha_0$. 
\end{proof}

\begin{lemma} \label{lemma_postive_J_LS}
	Under Assumptions \ref{assum_identifiability} and \ref{assum_stationarity_epsilon}, $\mathcal{I}_{LS}$ and $\mathcal{J}_{LS}$ are positive definite.
\end{lemma}

\begin{proof}
	We first show $\mathcal{J}_{LS}$ is positive definite, and it suffices to verify that $\bm c' \mathcal{J}_{LS} \bm c>0$ holds for any nonzero constant vector $\bm c \in \mathbb{R}^{n_\alpha}$.
	Recall that $\mathcal{J}_{LS} = E\left[\partial \bm \varepsilon_{t}^\prime(\bm \alpha_0) / \partial \bm \alpha \partial \bm \varepsilon_{t}(\bm \alpha_0)/ \partial \bm \alpha'\right]$, 
	hence $\bm c' \mathcal{J}_{LS} \bm c>0$ holds for any $\bm c$ that satisfies $(\partial \bm \varepsilon_{0t}/ \partial \bm \alpha') \bm c\neq 0_{N \times 1}$ a.s. 
	Then we only need to show that $(\partial \bm \varepsilon_{0t}/ \partial \bm \alpha') \bm c = 0_{N \times 1}$ holds a.s. if and only if $\bm c = 0_{n_\alpha \times 1}$. 	
	Suppose that $(\partial \bm \varepsilon_{0t}/ \partial \bm \alpha') \bm c = 0_{N \times 1}$. 
	Denote $\bm y_t = (y_{1t}, \ldots, y_{Nt})^\prime$. 
	Note that
	(a) $\{y_{kt}, 1\le k \le N\}$ are linearly independent, since if $y_{kt}$ can be represented as a linear combination of $\{y_{k^\prime t}, k^\prime \neq k\}$ then $\var(\bm y_t \mid \mathcal{F}_{t-1}) =\Sigma_0$ is not full rank, which contradicts with the positive definiteness of $\Sigma_0$ under Assumption \ref{assum_stationarity_epsilon}; 
	(b) $\{\bm y_{t-j}, j\ge 0\}$ are linearly independent by the proof of Lemma \ref{lemma_identifiability}. 
	By (a)--(b) and ($i$) in the proof of Lemma \ref{lemma_identifiability}, we can easily show that the following terms are linearly independent: 
	\begin{align*}
		\bigg \{ &\sum_{m=1}^{\infty} m \lambda^{m-1} G_{p+i} y_{k,t-p-m}, \quad
		\sum_{m=1}^{\infty} m \gamma_j^{m-1}[ \cos(m \varphi_j)G_{p+r+2j-1} + \sin(m \varphi_j) G_{p+r+2j}]y_{k,t-p-m},\\
		&\sum_{m=1}^{\infty} m \gamma_j^{m}[ \sin(m \varphi_j)G_{p+r+2j-1} - \cos(m \varphi_j) G_{p+r+2j}]y_{k,t-p-m}, \quad 
		y_{k, t-h},\\
		&\sum_{m=1}^{\infty}  \lambda^{m} G_{p+i} y_{k,t-p-m}, \quad \text{and}\quad 
		\sum_{m=1}^{\infty} \gamma_j^{m}[ \cos(m \varphi_j)G_{p+r+2j-1} + \sin(m \varphi_j) G_{p+r+2j}]y_{k,t-p-m},\\
		&1\le h\le p, 1\le i \le r, 1\le k \le N  \bigg\}.
	\end{align*}
	This together with \eqref{de_epsilon_alpha} implies that $\bm c=0_{n_\alpha\times 1}$.
	Thus $(\partial \bm \varepsilon_{0t}/ \partial \bm \alpha') \bm c = 0_{N \times 1}$ holds a.s. if and only if $\bm c = 0_{n_\alpha \times 1}$. 
	As a result, $\mathcal{J}_{LS}$ is positive definite. 

	Recall that $\mathcal{I}_{LS} = E\left[(\partial \bm \varepsilon_{t}^\prime(\bm \alpha_0) / \partial \bm \alpha) \Sigma_{0} (\partial \bm \varepsilon_{t}(\bm \alpha_0)/ \partial \bm \alpha')\right]$. 
	Since $\Sigma_{0}$ is positive definite under Assumption \ref{assum_stationarity_epsilon}, it can be similarly shown that $\mathcal{I}_{LS}$ is also positive definite. 
\end{proof}

\begin{lemma} \label{lemma_score_asy_Norm_LS}
	Suppose that $\{\bm y_t\}$ is generated by model \eqref{model_VARinf} and Assumptions \ref{assum_identifiability}--\ref{assum_space} hold, 
	then 
	\begin{align*}
		\frac{1}{\sqrt{T}} \sum_{t=1}^{T} \frac{\partial \widetilde{\bm \varepsilon}_{t}^\prime(\bm \alpha_0) }{\partial \bm \alpha} \widetilde{\bm \varepsilon}_{0t} \rightarrow_{d} N(0_{n_\alpha\times 1},\mathcal{I}_{LS}) \quad \text{as} \quad T\to \infty.
	\end{align*}
\end{lemma}

\begin{proof}
	By Lemma \ref{lemma_boundness_partial}($ii$), it can be shown that 
	\begin{align*}
		E\left(\left\|\frac{\partial\bm \varepsilon_{t}^\prime(\bm \alpha_0)}{\partial \bm{\alpha} } \Sigma_0 \frac{\partial \bm \varepsilon_{t}(\bm \alpha_0) }{\partial \bm{\alpha}^\prime}\right\|\right)
		\le& E\left(\left\| \frac{\partial \bm \varepsilon_{t}^\prime(\bm \alpha_0)}{\partial \bm \alpha} \right\|^2\|\Sigma_0\| \right)
		\le O(1) E\left(\xi_{{\rho},t-1}^2\right)\|\Sigma_0\|, 
	\end{align*}
	where $\xi_{\rho t}= 1+\sum_{h=0}^{\infty}\rho^{h}\|\bm y_{t-h}\|_2$. 
	This together with $\|\Sigma_0\| < \infty$ under Assumption \ref{assum_stationarity_epsilon}, and $E\left(\xi_{{\rho}t}^2\right) < \infty$ by Assumptions \ref{assum_stationarity_y} and \ref{assum_space}($ii$), implies that $\mathcal{I}_{LS}$ is finite.
	
	Let $S_{  1} = \sum_{t=1}^T S_{1t}$ with $S_{1t} =\bm c'\partial \bm \varepsilon_{t}^\prime(\bm \alpha_0) /\partial \bm \alpha \bm \varepsilon_{0t}$, where $\bm c$ is a constant vector with the same dimension as $\bm \alpha$. 
	Note that by Assumptions \ref{assum_identifiability}--\ref{assum_stationarity_epsilon} and \eqref{de_epsilon_alpha}, 
	$\{S_{1t}\}$ is ergodic and stationary, 
	$S_{1t}$ is measurable with respect to $\mathcal{F}_{t}$ that is the $\sigma$-field generated by $\{\bm \varepsilon_{0s}, s \le t\}$, and 
	$E(S_{1t}\mid \mathcal{F}_{t-1})=0$. 
	Moreover, since $\mathcal{I}_{LS}$ is finite, 
	we can obtain that 
	\begin{align*}
		\var(S_{1t})= \bm c' E\left[\frac{\partial \bm \varepsilon_{t}^\prime(\bm \alpha_0)}{\bm \alpha}\Sigma_0\frac{\partial \bm \varepsilon_{t}(\bm \alpha_0)}{\partial \bm \alpha^\prime } \right]\bm c = \bm c'\mathcal{I}_{LS} \bm c < \infty.
	\end{align*}
	Thus $\{S_{1t}, \mathcal{F}_{t}\}_t$ is an ergodic, stationary and square integrable martingale difference. 
	By the martingale central limit theorem (CLT) for a stationary and ergodic sequence of square integrable martingale increments \citep{Billingsley1961lindeberg}, we have that 
	\begin{align*}
		\frac{1}{\sqrt{T}} S_{  1}\rightarrow_{d} N(0,\bm c' \mathcal{I}_{LS}\bm c) \;\; \text{as} \;\; T \to \infty. 
	\end{align*}
	This together with the Cram\'{e}r-Wold device, implies that 
	\begin{align}\label{sum dLS convergence}
		\frac{1}{\sqrt{T}} \sum_{t=1}^{T} \frac{\partial \bm \varepsilon_{t}^\prime(\bm \alpha_0)}{\partial \bm \alpha} \bm \varepsilon_{0t} \rightarrow_d N(0_{n_\alpha\times 1}, \mathcal{I}_{LS}) \;\; \text{as} \;\; T \to \infty. 
	\end{align}
	Furthermore, by Lemma \ref{lemma_initial_effect_LS}($ii$), it holds that 
	\begin{align}
		\left\| \frac{1}{\sqrt{T}}\sum_{t=1}^{T} \left[\frac{\partial \bm \varepsilon_{t}^\prime(\bm \alpha_0)}{\partial \bm \alpha} \bm \varepsilon_{0t}
		- \frac{\partial \widetilde{\bm \varepsilon}_{t}^\prime(\bm \alpha_0)}{\partial \bm \alpha} \widetilde{\bm \varepsilon}_{0t} \right]\right\|
		&\leq \frac{1}{\sqrt{T}} \sum_{t=1}^{T} \left\| \frac{\partial \bm \varepsilon_{t}^\prime(\bm \alpha_0)}{\partial \bm \alpha} \bm \varepsilon_{0t}
		- \frac{\partial \widetilde{\bm \varepsilon}_{t}^\prime(\bm \alpha_0) }{\partial \bm \alpha} \widetilde{\bm \varepsilon}_{0t}\right\|\nonumber\\
		&\leq \frac{1}{\sqrt{T}} \sum_{t=1}^{T}O(\rho^t)\xi_{{\rho}0}(\xi_{{\rho},t-1} +\|\bm \varepsilon_{0t}\|). \label{o_1_dLS}
	\end{align}
	Then by $0 < {\rho} <1$ under Assumption \ref{assum_space}($ii$), $E\left(\xi_{{\rho}t}^2\right) < \infty$ and Markov's inequality, 
	it can be shown that
	\begin{align}\label{initial_depsilon}
		\left\| \frac{1}{\sqrt{T}}\sum_{t=1}^{T} \left[\frac{\partial \bm \varepsilon_{t}^\prime(\bm \alpha_0)}{\partial \bm \alpha} \bm \varepsilon_{0t}
		- \frac{\partial \widetilde{\bm \varepsilon}_{t}^\prime(\bm \alpha_0)}{\partial \bm \alpha} \widetilde{\bm \varepsilon}_{0t} \right]\right\|
		=O_p(T^{-1/2}) .
	\end{align}
	Thus by \eqref{sum dLS convergence},\eqref{initial_depsilon} and Slutsky's theorem, we can obtain that 
	\begin{align*}
		\frac{1}{\sqrt{T}} \sum_{t=1}^{T} \frac{\partial \widetilde{\bm \varepsilon}_{t}^\prime(\bm \alpha_0)}{\partial \bm \alpha} \widetilde{\bm \varepsilon}_{0t}\rightarrow_d N(0_{n_\alpha\times 1}, \mathcal{I}_{LS}) \;\; \text{as} \;\; T \to \infty.
	\end{align*}
\end{proof}

\section{Technical Proofs} \label{appendix_proof_theorems}

\subsection{Proof of Theorem \ref{thmQMLE}}\label{proof_QMLE}

\begin{proof}
	We first show the consistency of the QMLE $\widehat{\bm \theta}$. 
	By Theorem 4.1.1 in \cite{Amemiya1985}, it suffices to show the following results hold: 
	\begin{enumerate}
		\item[(A1)] the space $\Theta$ is compact and $\bm \theta_{0}$ is an interior point in $\Theta$;
		\item[(A2)] $\widetilde{\mathcal{L}}(\bm \theta)$ is continuous in $\bm \theta \in \Theta$ and is a measurable function of $\{\bm y_t,~t\in Z\}$ for all $\bm \theta \in \Theta$;
		\item[(A3)] $\sup_{\bm \theta \in \Theta}|\widetilde{\mathcal{L}}(\bm \theta)-E[l_t(\bm \theta)]|\rightarrow_p 0$ as $T \to \infty$;
		\item[(A4)] $E[l_t(\bm \theta)]$ has a unique minimum at $\bm \theta_0$. 
	\end{enumerate}
	(A1) and (A2) hold under Assumptions \ref{assum_stationarity_y} and \ref{assum_space}($i$), 
	and (A4) holds by Lemma \ref{lemma_uni_max_MLE}. 
	Hence we only need to verify (A3). 
	For (A3), 
	recall that $\widetilde{\mathcal{L}}(\bm{\theta})=T^{-1} \sum_{t=1}^{T} \widetilde{l}_t (\bm{\theta})$. 
	By Lemma \ref{lemma_initial_effect_MLE}($i$), we obtain that 
	\begin{align} \label{lt - lttilde}
		\sup_{\bm \theta \in \Theta} |l_t(\bm \theta)- \widetilde{l}_t(\bm \theta)|\le O(\rho^t)\xi_{{\rho}0}\xi_{{\rho}t}, 
	\end{align}
	where $\xi_{\rho t}= 1+\sum_{j=0}^{\infty}\rho^{j}\|\bm y_{t-j}\|_2$. 
	Note that $0 < {\rho} <1$ under Assumption \ref{assum_space}($ii$), and $E\left(\xi_{{\rho}t}^2\right) < \infty$ by $0 < {\rho} <1$, $E(\|\bm y_t\|_2^2)<\infty$ and Assumption \ref{assum_stationarity_y}. 
	Then using Markov's inequality, \eqref{lt - lttilde} implies that 
	\begin{align}\label{o_l_initial}
		\sup_{\bm \theta \in \Theta} \left|\frac{1}{T} \sum_{t=1}^{T} \left[ l_t(\bm \theta)- \widetilde{l}_t(\bm \theta)\right] \right| 
		\leq \frac{1}{T} \sum_{t=1}^{T} \sup_{\bm \theta \in \Theta} \left| l_t(\bm \theta)- \widetilde{l}_t(\bm \theta) \right| 
		= O_p\left(\frac{1}{T}\right). 
	\end{align}
	Moreover, 
	by Lemma \ref{lemma_three consistent_MLE}($i$), it holds that 
	\begin{align}\label{o_l}
		\sup_{\bm \theta \in \Theta}\left|\dfrac{1}{T} \sum_{t=1}^T \left\{l_t(\bm \theta)-E \left[l_{t}(\bm \theta)\right]\right\}\right|=o_{p}(1)
		\;\; \text{as} \;\; T \to \infty.  
	\end{align}
	Thus by \eqref{o_l_initial} and \eqref{o_l}, we have that 
	\begin{align*}
		\sup_{\bm \theta \in \Theta}\left|\widetilde{\mathcal{L}}(\bm \theta)-E[l_t(\bm \theta)]\right|\rightarrow_p 0 \;\; \text{as} \;\; T \to \infty, 
	\end{align*}
	and then (A3) holds. 
	As a result, Theorem \ref{thmQMLE}($i$) holds by (A1)--(A4). 
	
	Next we prove the asymptotic normality of $\widehat{\bm \theta}$. 
	By Theorem \ref{thmQMLE}($i$), it holds that $\widehat{\bm{\theta}} \rightarrow_{p} \bm{\theta}_0$ as $T \to \infty$. 
	To utilize Theorem 4.1.3 of \cite{Amemiya1985}, we then show that the following results hold:
	\begin{enumerate}
		\item[(B1)] $\partial^2 \widetilde{\mathcal{L}}(\bm{\theta})/(\partial \bm \theta \partial \bm \theta') $ exists and is continuous on $\Theta$;
		\item[(B2)] $ \partial^{2} \widetilde{\mathcal{L}}(\widehat{\bm \theta}) / (\partial \bm \theta \partial \bm \theta^{\prime})\rightarrow_{p} \mathcal{J}$ as $T\to \infty$, and $\mathcal{J}$ is positive definite;
		\item[(B3)] $T^{1/2} \partial \widetilde{\mathcal{L}}(\bm \theta_0) / \partial \bm \theta \rightarrow_d N\left(0_{n\times 1}, \mathcal{I}\right)$ as $T \to \infty$. 
  	\end{enumerate}
  (B1) holds under Assumption \ref{assum_stationarity_y}, 
	and (B3) holds by Lemma \ref{lemma_score_asy_Norm_MLE}. 
	Hence we only need to show that (B2) holds. 
  Similar to the proof of \eqref{o_l_initial}, using Lemma \ref{lemma_initial_effect_MLE}($iii$) and Markov's inequality, it can be shown that 
  \begin{align*}%\label{o_dl^2_initial}
  	\sup_{\bm \theta \in \Theta} \left\|\frac{1}{T} \sum_{t=1}^{T}  \left(\frac{\partial^2 l_t(\bm \theta)}{\partial \bm \theta \partial \bm \theta^\prime }- \frac{\partial^2 \widetilde{l}_t(\bm \theta)}{\partial \bm \theta  \partial \bm \theta^\prime}\right)\right\| =O_p\left(\frac{1}{T}\right). 
  \end{align*}
  This together with Lemma \ref{lemma_three consistent_MLE}($iii$) implies that 
  \begin{align}\label{ddL_ML convergence}
  	\frac{\partial^{2} \widetilde{\mathcal{L}}(\widehat{\bm \theta})} {\partial \bm \theta \partial \bm \theta^{\prime}}\rightarrow_{p} E\left[\frac{\partial^{2} l_t(\widehat{\bm \theta})}{\partial \bm \theta \partial \bm \theta^{\prime}}\right] \;\; \text{as} \;\; T \to \infty. 
  \end{align}
  Moreover, since $\widehat{\bm{\theta}} \rightarrow_{p} \bm{\theta}_0$ as $T \to \infty$ and $\mathcal{J} = E(\partial ^2 l_t(\bm{\theta}_0)/\partial \bm{\theta}\partial \bm{\theta}^\prime)$ with $E(\partial ^2 l_t(\bm{\theta})/\partial \bm{\theta}\partial \bm{\theta}^\prime)$ being continuous on $\Theta$, we can obtain that 
  \begin{align}\label{ddl_t convergence}
  	E\left[\frac{\partial^{2} l_t(\widehat{\bm \theta})} {\partial \bm \theta \partial \bm \theta^{\prime}}\right]\rightarrow_{p} \mathcal{J} \;\; \text{as} \;\; T \to \infty. 
  \end{align}
  Then by \eqref{ddL_ML convergence} and \eqref{ddl_t convergence}, it holds that 
  \begin{align*}
  	 \frac{\partial^{2} \widetilde{\mathcal{L}}(\widehat{\bm \theta})} {\partial \bm \theta \partial \bm \theta^{\prime}} \rightarrow_{p} \mathcal{J} \;\; as \;\; T \to \infty. 
  \end{align*}
	Furthermore, $\mathcal{J}$ is positive definite by Lemma \ref{lemma_postive_J_MLE}. 
	Thus (B2) holds.
	As a result, by the consistency of $\widehat{\bm \theta}$, (B1)--(B3) and Theorem 4.1.3 of \cite{Amemiya1985}, it can be shown that 
	\begin{align*} % \label{norm_theta ML}
		\sqrt{T}(\widehat{\bm \theta}-\bm \theta_0) \rightarrow_{d} N(0_{n\times 1}, \mathcal{J}^{-1} \mathcal{I} \mathcal{J}^{-1}) \;\;
		\text{as} \;\; T \to \infty. 
	\end{align*}
	This completes the proof of Theorem \ref{thmQMLE}.
\end{proof}

\subsection{Proof of Proposition \ref{thmLS}}
\begin{proof}
	We first show the consistency of the LSE $\widehat{\bm \alpha}_{LS}$. 
	By Theorem 4.1.1 in \cite{Amemiya1985}, it suffices to show the following results hold: 
	\begin{enumerate}
		\item[(a1)] the space $\mathcal{A}$ is compact and $\bm \alpha_{0}$ is an interior point in $\mathcal{A}$;
		\item[(a2)] $\widetilde{\mathcal{L}}_{LS}(\bm \alpha)$ is continuous in $\bm \alpha \in \mathcal{A}$ and is a measurable function of $\{\bm y_t,~t\in Z\}$ for all $\bm \alpha \in \mathcal{A}$;
		\item[(a3)] $\sup_{\bm \alpha \in \mathcal{A}}|\widetilde{\mathcal{L}}_{LS}(\bm \alpha)-E[\bm \varepsilon_{t}^\prime (\bm\alpha)\bm \varepsilon_{t}(\bm\alpha)]|\rightarrow_p 0$ as $T \to \infty$;
		\item[(a4)] $E[\bm \varepsilon_{t}^\prime(\bm \alpha)  \bm \varepsilon_{t}(\bm \alpha)]$ has a unique minimum at $\bm \alpha_0$. 
	\end{enumerate}
	(a1) and (a2) hold under Assumptions \ref{assum_stationarity_y}, \ref{assum_space}($i$) and Lemma \ref{lemma_boundness_LS}($i$), 
	and (a4) holds by Lemma \ref{lemma_uni_max}. 
	Hence we only need to verify (a3). 
	For (a3), recall that $\widetilde{\mathcal{L}}_{LS}(\bm{\alpha})=T^{-1} \sum_{t=1}^{T} \widetilde{\bm \varepsilon}_{t}^\prime (\bm\alpha) \widetilde{\bm \varepsilon}_{t}(\bm\alpha)$. 
	By Lemma \ref{lemma_initial_effect_LS}($i$), we obtain that 
	\begin{align} \label{epsilont2 - epsilont2tilde}
		\sup_{\bm \alpha \in \mathcal{A}} |\bm \varepsilon_{t}^\prime (\bm\alpha)\bm \varepsilon_{t}(\bm\alpha)- \widetilde{\bm \varepsilon}_{t}^\prime (\bm\alpha )\widetilde{\bm \varepsilon}_{t}(\bm\alpha)|\le O(\rho^t)\xi_{{\rho}0}(\xi_{{\rho},t-1} +\|\bm \varepsilon_{0t}\|), 
	\end{align}
	where $\xi_{\rho t}= 1+\sum_{j=0}^{\infty}\rho^{j}\|\bm y_{t-j}\|_2$. 
	Note that $0 < {\rho} <1$ under Assumption \ref{assum_space}($ii$), and $E\left(\xi_{{\rho}t}^2\right) < \infty$ by $0 < {\rho} <1$, $E(\|\bm y_t\|_2^2)<\infty$ and Assumption \ref{assum_stationarity_y}. 
	Then using Markov's inequality, \eqref{epsilont2 - epsilont2tilde} implies that
	\begin{align}\label{o_epsilon_initial}
		\sup_{\bm \alpha \in \mathcal{A}}  \left|\frac{1}{T} \sum_{t=1}^{T}[\bm \varepsilon_{t}^\prime (\bm\alpha)\bm \varepsilon_{t}(\bm\alpha)- \widetilde{\bm \varepsilon}_{t}^\prime (\bm\alpha) \widetilde{\bm \varepsilon}_{t}(\bm\alpha)]\right| 
		\leq \frac{1}{T} \sum_{t=1}^{T} \sup_{\bm \alpha \in \mathcal{A}}  \left|\bm \varepsilon_{t}^\prime (\bm\alpha)\bm \varepsilon_{t}(\bm\alpha)- \widetilde{\bm \varepsilon}_{t}^\prime (\bm\alpha) \widetilde{\bm \varepsilon}_{t}(\bm\alpha)\right| 
		= O_p\left(\frac{1}{T}\right).
	\end{align}
	Moreover, by Lemma \ref{lemma_three consistent_LS}($i$), it holds that 
	\begin{align}\label{o_epsilon}
		\sup_{\bm \alpha \in \mathcal{A}}\left|\dfrac{1}{T} \sum_{t=1}^T \left\{\bm \varepsilon_{t}^\prime (\bm\alpha)\bm \varepsilon_{t}(\bm\alpha)-E \left[\bm \varepsilon_{t}^\prime (\bm\alpha)\bm \varepsilon_{t}(\bm\alpha)\right]\right\}\right|=o_{p}(1)
		\;\; \text{as} \;\; T \to \infty.  
	\end{align}
	Thus by \eqref{o_epsilon_initial} and \eqref{o_epsilon}, we have that 
	\begin{align*}
		\sup_{\bm \alpha \in \mathcal{A}}\left|\widetilde{\mathcal{L}}_{LS}(\bm \alpha)-E[\bm \varepsilon_{t}^\prime (\bm\alpha)\bm \varepsilon_{t}(\bm\alpha)]\right|\rightarrow_p 0 \;\; \text{as} \;\; T \to \infty, 
	\end{align*}
	and then (a3) holds. 
	As a result, the consistency of $\widehat{\bm \alpha}_{LS}$ is established by (a1)--(a4). 
	
	Next we show the consistency of $\widehat{\Sigma}_{LS}$. 
	Recall that $\widehat{\Sigma}_{LS} = T^{-1} \sum_{t=1}^{T} \widetilde{\bm \varepsilon}_t(\widehat{\bm \alpha}_{LS}) \widetilde{\bm \varepsilon}_t^\prime(\widehat{\bm \alpha}_{LS})$. 
	With similar arguments in the proof of Lemma \ref{lemma_initial_effect_LS} and \eqref{o_epsilon_initial}, we can show that 
	\begin{align*}
		\sup_{\bm \alpha \in \mathcal{A}} \left\| \frac{1}{T} \sum_{t=1}^{T}[\bm \varepsilon_{t} (\bm\alpha)\bm \varepsilon_{t}^\prime(\bm\alpha)- \widetilde{\bm \varepsilon}_{t} (\bm\alpha) \widetilde{\bm \varepsilon}^\prime_{t}(\bm\alpha)]\right\| = O_p\left(\frac{1}{T}\right). 
	\end{align*}
	This together with $\widehat{\bm{\alpha}}_{LS} \rightarrow_{p} \bm{\alpha}_0$ as $T \to \infty$ and $\bm \alpha_0$ is an interior point in $\mathcal{A}$ under Assumption \ref{assum_space}($i$), implies that 
	\begin{align*}
		\left\|\widehat{\Sigma}_{LS} - \frac{1}{T} \sum_{t=1}^{T}\bm \varepsilon_{t} (\widehat{\bm \alpha}_{LS})\bm \varepsilon_{t}^\prime(\widehat{\bm \alpha}_{LS})\right\| = O_p\left(\frac{1}{T}\right).  
	\end{align*}
	And it follows that 
	\begin{align} \label{SigmaLS eq1}
		\widehat{\Sigma}_{LS} - \frac{1}{T} \sum_{t=1}^{T}\bm \varepsilon_{t} (\widehat{\bm \alpha}_{LS})\bm \varepsilon_{t}^\prime(\widehat{\bm \alpha}_{LS}) = O_p\left(\frac{1}{T}\right).  
	\end{align}
	Moreover, by Taylor's expansion and $\bm \varepsilon_{0t} = \bm \varepsilon_{t}(\bm \alpha_0)$, we have that 
	\begin{align}
		&\frac{1}{T} \sum_{t=1}^{T} {\rm vec} [\bm \varepsilon_{t}(\widehat{\bm \alpha}_{LS}) \bm \varepsilon_{t}^{\prime}(\widehat{\bm \alpha}_{LS})] \nonumber \\
		=&\frac{1}{T} \sum_{t=1}^{T}\Bigg\{  {\rm vec} (\bm \varepsilon_{0t} \bm \varepsilon_{0t}^{\prime}) 
		+ \left[(\bm \varepsilon_{t}(\bm \alpha^\star) \otimes I_N) \frac{\partial \bm \varepsilon_t(\bm \alpha^\star)}{\partial \bm \alpha^\prime} +(I_N \otimes \bm \varepsilon_{t}(\bm \alpha^\star)) \frac{\partial \bm \varepsilon_t(\bm \alpha^\star)}{\partial \bm \alpha^\prime}  \right] (\widehat{\bm \alpha}_{LS} - \bm \alpha_0)\Bigg\}, \label{taylor sigma}
	\end{align}
	where $\bm \alpha^\star$ lies between $\widehat{\bm \alpha}_{LS}$ and $\bm \alpha_{0}$. 
	With analogous arguments in the proof of Lemma \ref{lemma_three consistent_LS}, it can be shown that 
	\begin{align} \label{d - Ed}
		\sup_{\bm \alpha \in \mathcal{A}} \bigg\|
		&\frac{1}{T} \sum_{t=1}^{T} \left\{\left[(\bm \varepsilon_{t}(\bm \alpha) \otimes I_N) \frac{\partial \bm \varepsilon_t(\bm \alpha)}{\partial \bm \alpha^\prime} +(I_N \otimes \bm \varepsilon_{t}(\bm \alpha)) \frac{\partial \bm \varepsilon_t(\bm \alpha)}{\partial \bm \alpha^\prime}  \right]\right. \nonumber \\
		&-E\left.\left[(\bm \varepsilon_{t}(\bm \alpha) \otimes I_N) \frac{\partial \bm \varepsilon_t(\bm \alpha)}{\partial \bm \alpha^\prime} +(I_N \otimes \bm \varepsilon_{t}(\bm \alpha)) \frac{\partial \bm \varepsilon_t(\bm \alpha)}{\partial \bm \alpha^\prime}  \right]\right\} \bigg\| = o_p(1), 
	\end{align}
	as $T \to \infty$. 
	Note that by \eqref{de_epsilon_alpha}, $\bm{\varepsilon}_t(\bm{\alpha}) = \bm y_t - \sum_{h=1}^{\infty}A_{h}(\bm{\alpha})\bm y_{t-h}$ and $E[\bm \varepsilon_{t}(\bm \alpha_0)]=0$ under Assumption \ref{assum_stationarity_epsilon} and using the law of iterated expectation, we have that 
	\begin{align} \label{Ed at alpha0}
		E\left[(\bm \varepsilon_{t}(\bm \alpha_0) \otimes I_N) \frac{\partial \bm \varepsilon_t(\bm \alpha_0)}{\partial \bm \alpha^\prime} +(I_N \otimes \bm \varepsilon_{t}(\bm \alpha_0)) \frac{\partial \bm \varepsilon_t(\bm \alpha_0)}{\partial \bm \alpha^\prime}\right] = 0. 
	\end{align}
	Then by \eqref{d - Ed}--\eqref{Ed at alpha0}, $\widehat{\bm{\alpha}}_{LS} \rightarrow_{p} \bm{\alpha}_0$ as $T \to \infty$ and $\bm \alpha_0$ is an interior point in $\mathcal{A}$, it can be obtained that 
	\begin{align}\label{o_1d_sigma}
		\frac{1}{T} \sum_{t=1}^{T} \left[(\bm \varepsilon_{t}(\bm \alpha^\star) \otimes I_N) \frac{\partial \bm \varepsilon_t(\bm \alpha^\star)}{\partial \bm \alpha^\prime} +(I_N \otimes \bm \varepsilon_{t}(\bm \alpha^\star)) \frac{\partial \bm \varepsilon_t(\bm \alpha^\star)}{\partial \bm \alpha^\prime}  \right] = o_p(1) 
		\;\; \text{as} \;\; T \to \infty,  
	\end{align} 
	and thus \eqref{taylor sigma} implies that 
	\begin{align} \label{SigmaLS eq2}
		\frac{1}{T} \sum_{t=1}^{T} {\rm vec} [\bm \varepsilon_{t}(\widehat{\bm \alpha}_{LS}) \bm \varepsilon_{t}^{\prime}(\widehat{\bm \alpha}_{LS})]
		=\frac{1}{T} \sum_{t=1}^{T}  {\rm vec} (\bm \varepsilon_{0t} \bm \varepsilon_{0t}^{\prime}) + o_p(1) 
		\;\; \text{as} \;\; T \to \infty.  
	\end{align}
	Furthermore, using the law of large number under Assumption \ref{assum_stationarity_epsilon}, it holds that 
	\begin{align} \label{SigmaLS eq3}
		\frac{1}{T} \sum_{t=1}^{T}  \bm \varepsilon_{0t} \bm \varepsilon_{0t}^{\prime} - \Sigma_0 
		= \frac{1}{T} \sum_{t=1}^{T}  \bm \varepsilon_{0t} \bm \varepsilon_{0t}^{\prime} - \frac{1}{T} \sum_{t=1}^{T} E\left(\bm \varepsilon_{0t} \bm \varepsilon_{0t}^{\prime}\right) 
		= o_p(1) 
		\;\; \text{as} \;\; T \to \infty.  
	\end{align}
	As a result, \eqref{SigmaLS eq1}, \eqref{SigmaLS eq2} and \eqref{SigmaLS eq3} imply that $\widehat{\Sigma}_{LS} \to_p \Sigma_0$ as $T\to \infty$. Then Proposition \ref{thmLS}($i$) is established. 
	
	We then prove the asymptotic normality of $\widehat{\bm \alpha}_{LS}$. 
	By Proposition \ref{thmLS}($i$), it holds that $\widehat{\bm{\alpha}}_{LS} \rightarrow_{p} \bm{\alpha}_0$ as $T \to \infty$. 
	To utilize Theorem 4.1.3 of \cite{Amemiya1985}, we then show that the following results hold:
	\begin{enumerate}
		\item[(b1)] $\partial^2 \widetilde{\mathcal{L}}_{LS}(\bm{\alpha})/(\partial \bm \alpha \partial \bm \alpha') $ exists and is continuous on $\Theta$;
		\item[(b2)] $ \partial^{2} \widetilde{\mathcal{L}}_{LS}(\widehat{\bm \alpha}_{LS}) / (\partial \bm \alpha \partial \bm \alpha^{\prime})\rightarrow_{p} \mathcal{J}_{LS}$ as $T\to \infty$, and $\mathcal{J}_{LS}$ is positive definite;
		\item[(b3)] $T^{1/2} \partial \widetilde{\mathcal{L}}_{LS}(\bm \alpha_0) / \partial \bm \alpha \rightarrow_d N\left(0_{n_\alpha\times 1}, \mathcal{I}_{LS}\right)$ as $T \to \infty$.
	\end{enumerate}	
	(b1) holds under Assumption \ref{assum_stationarity_y}, 
	and (b3) holds by Lemma \ref{lemma_score_asy_Norm_LS}. 
	Hence we only need to show that (b2) holds. 
  Similar to the proof of \eqref{o_epsilon_initial}, using Lemma \ref{lemma_initial_effect_LS}($iii$) and Markov's inequality, it can be shown that 
	\begin{align*}%\label{o_depsilon^2_initial}
		\sup_{\bm \alpha \in \mathcal{A}}  &\left\|  \frac{1}{T} \sum_{t=1}^{T} \bigg\{\left[ (\bm \varepsilon_t^\prime (\bm \alpha)\otimes I_{n_\alpha})
		\dfrac{\partial \rm{vec}(\partial \bm \varepsilon_{t}^\prime (\bm \alpha)/ \partial \bm \alpha)}{\partial \bm \alpha^\prime}\bm \varepsilon_t(\bm \alpha)
		+ \frac{\partial \bm \varepsilon_t^\prime(\bm \alpha)}{\partial \bm \alpha} \frac{\partial \bm \varepsilon_t(\bm \alpha)}{\partial \bm \alpha^\prime}  
		\right]
		\right. \nonumber\\
		&\left.\quad\quad\quad \quad 
		-\left[
		(\widetilde{\bm \varepsilon}_t^\prime (\bm \alpha)\otimes I_{n_\alpha}) \dfrac{\partial \rm{vec}(\partial \widetilde{\bm \varepsilon}_{t}^\prime (\bm \alpha)/ \partial \bm \alpha)}{\partial \bm \alpha^\prime} 
		+ \frac{\partial \widetilde{\bm \varepsilon}_t^\prime(\bm \alpha)}{\partial \bm \alpha} \frac{\partial \widetilde{\bm \varepsilon}_t(\bm \alpha)}{\partial \bm \alpha^\prime}  
		\right]\bigg\} \right\| = O_p\left(\frac{1}{T}\right).
	\end{align*}
	This together with Lemma \ref{lemma_three consistent_LS}($iii$) implies that 
	\begin{align}\label{ddL_LS convergence}
		\frac{\partial^{2} \widetilde{\mathcal{L}}_{LS}(\widehat{\bm \alpha}_{LS})} {\partial \bm \alpha \partial \bm \alpha^{\prime}}\rightarrow_{p} E\left[\frac{\partial^{2} \mathcal{L}_{LS}(\widehat{\bm \alpha}_{LS})}{\partial \bm \alpha \partial \bm \alpha^{\prime}}\right] \;\; \text{as} \;\; T \to \infty. 
	\end{align}
	Moreover, note that 
	\begin{align*}
		E\left[\frac{\partial^{2} \mathcal{L}_{LS}(\bm{\alpha}_0)}{\partial \bm \alpha \partial \bm \alpha^{\prime}}\right] 
		= E\left[ \frac{\partial \bm \varepsilon_{t}^\prime(\bm \alpha_0)}{\partial \bm \alpha} 
		\frac{\partial \bm \varepsilon_{t}(\bm \alpha_0)}{\partial \bm \alpha^\prime}\right] 
		= \mathcal{J}_{LS}. 
	\end{align*}
	Since $\widehat{\bm{\alpha}}_{LS} \rightarrow_{p} \bm{\alpha}_0$ as $T \to \infty$ and $E[\partial ^2 \mathcal{L}_{LS}(\bm{\alpha})/\partial \bm{\alpha}\partial \bm{\alpha}^\prime]$ is continuous on $\mathcal{A}$, we have that 
	\begin{align}\label{ddepsilon convergence}
		E\left[\frac{\partial^{2}\widetilde{\mathcal{L}}_{LS}(\widehat{\bm \alpha}_{LS})} {\partial \bm \alpha \partial \bm \alpha^{\prime}}\right]\rightarrow_{p} \mathcal{J}_{LS} \;\; \text{as} \;\; T \to \infty. 
	\end{align}
	Then by \eqref{ddL_LS convergence} and \eqref{ddepsilon convergence}, it holds that 
	\begin{align*}
		\frac{\partial^{2} \widetilde{\mathcal{L}}_{LS}(\widehat{\bm \alpha}_{LS})} {\partial \bm \alpha \partial \bm \alpha^{\prime}} \rightarrow_{p} \mathcal{J}_{LS} \;\; as \;\; T \to \infty. 
	\end{align*}
	Furthermore, $\mathcal{J}_{LS}$ is positive definite by Lemma \ref{lemma_postive_J_LS}. 
	Thus (b2) holds. 
	As a result, by the consistency of $\widehat{\bm \alpha}_{LS}$, (b1)--(b3) and Theorem 4.1.3 of \cite{Amemiya1985}, it can be shown that 
	\begin{align*} % \label{norm_theta LS}
		\sqrt{T}(\widehat{\bm \alpha}_{LS}-\bm \alpha_0) \rightarrow_{d} N(0_{n_\alpha\times 1}, \mathcal{J}_{LS}^{-1} \mathcal{I}_{LS} \mathcal{J}_{LS}^{-1}) 
		\;\; \text{as} \;\; T \to \infty,  
	\end{align*}
	and then Proposition \ref{thmLS}($ii$) holds. 
	
	Lastly we prove Proposition \ref{thmLS}($iii$), that is the asymptotic normality of ${\rm vec}(\widehat{\Sigma}_{LS})$. 
	By \eqref{SigmaLS eq1}--\eqref{taylor sigma}, as $T \to \infty$, it holds that  
	\begin{align} \label{SigmaLS converge in d eq1}
		&\sqrt{T}{\rm vec}(\widehat{\Sigma}_{LS}) \nonumber \\
		=&\frac{1}{\sqrt{T}} \sum_{t=1}^{T}\Bigg\{  {\rm vec} (\bm \varepsilon_{0t} \bm \varepsilon_{0t}^{\prime})
		+ \left[(\bm \varepsilon_{t}(\bm \alpha^\star) \otimes I_N) \frac{\partial \bm \varepsilon_t(\bm \alpha^\star)}{\partial \bm \alpha^\prime} +(I_N \otimes \bm \varepsilon_{t}(\bm \alpha^\star)) \frac{\partial \bm \varepsilon_t(\bm \alpha^\star)}{\partial \bm \alpha^\prime}  \right] (\widehat{\bm \alpha}_{LS} - \bm \alpha_0)\Bigg\} + o_p(1),
	\end{align}
	where $\bm \alpha^\star$ lies between $\widehat{\bm \alpha}_{LS}$ and $\bm \alpha_{0}$. Using Lindeberg–L\'{e}vy central limit theorem under Assumption \ref{assum_stationarity_epsilon}, $E(\bm \varepsilon_{0t} \bm \varepsilon_{0t}^\prime \mid \mathcal{F}_{t-1}) = \Sigma_0$ and $E(\|\bm \varepsilon_{0t}\|^4) < \infty$, then by the martingale CLT for a stationary and ergodic sequence of square integrable martingale increments\citep{Billingsley1961lindeberg} we obtain that 
	\begin{align} \label{SigmaLS converge in d eq2}
		\frac{1}{\sqrt{T}} \sum_{t=1}^{T}{\rm vec}(\bm \varepsilon_{0t} \bm \varepsilon_{0t}^{\prime}-\Sigma_0) \rightarrow_{d} N(0_{N^2 \times 1}, \mathcal{K})
		\;\; \text{as} \;\; T \to \infty,  
	\end{align}
	where $\mathcal{K} = \var[{\rm vec}(\bm \varepsilon_{0t} \bm \varepsilon_{0t}^{\prime})]$.
	Moreover, by Proposition \ref{thmLS}($ii$) and \eqref{o_1d_sigma}, it holds that
	\begin{align} \label{SigmaLS converge in d eq3}
		\frac{1}{T} \sum_{t=1}^{T}\left[(\bm \varepsilon_{t}(\bm \alpha^\star) \otimes I_N) \frac{\partial \bm \varepsilon_t(\bm \alpha^\star)}{\partial \bm \alpha^\prime} +(I_N \otimes \bm \varepsilon_{t}(\bm \alpha^\star)) \frac{\partial \bm \varepsilon_t(\bm \alpha^\star)}{\partial \bm \alpha^\prime}  \right]
		\sqrt{T}(\widehat{\bm \alpha}_{LS} - \bm \alpha_0)
		=o_p(1),
	\end{align}
	as $T \to \infty$. 
	As a result, using Slutsky's theorem, \eqref{SigmaLS converge in d eq1}--\eqref{SigmaLS converge in d eq3} imply that 
	\begin{align*}
		\sqrt{T} \left({\rm vec}(\widehat{\Sigma}_{LS})- {\rm vec}(\Sigma_{0})\right)\rightarrow_{d}N(0_{N^2\times 1},\mathcal{K}) 
		\;\; \text{as} \;\; T \to \infty.  
	\end{align*} 
	This completes the proof of Proposition \ref{thmLS}.
\end{proof}

\subsection{Proof of Proposition \ref{thmcompare_efficient}}
\begin{proof}
	We first establish the following results: 
	\begin{enumerate}[(a)]
		\item $\sqrt{T}(\widehat{\bm \alpha}- \bm \alpha_0)\rightarrow_d N(0_{n_\alpha\times 1},\Xi)$ as $T \to \infty$; 
		\item $\sqrt{T} ({\rm vec}(\widehat{\Sigma})- {\rm vec}(\Sigma_{0}))\rightarrow_{d}N(0_{N^2\times 1},\mathcal{K})$ as $T \to \infty$, 
	\end{enumerate}
	where 
	$\bm{\theta}_0=(\bm{\alpha}_0^\prime, \bm{\sigma}_0^\prime)^\prime$ with $\bm{\sigma}_0= \text{vech}(\Sigma_0)$, 
	$\widehat{\bm \theta} = (\widehat{\bm \alpha}^\prime , \widehat{\bm \sigma}^\prime)^\prime$ with $\widehat{\bm \sigma}= \text{vech}(\widehat{\Sigma})$, 
	$\Xi = \{E [(\partial \bm \varepsilon_t^\prime (\bm \alpha_0)/\partial \bm \alpha) \\
	\Sigma_0^{-1} (\partial \bm \varepsilon_t (\bm \alpha_0)/\partial \bm \alpha^\prime) ]\}^{-1} $ and $\mathcal{K} = \var[{\rm vec}(\bm \varepsilon_{0t} \bm \varepsilon_{0t}^{\prime})]$. 
	For (a), 
	by Theorem \ref{thmQMLE}, it holds that 
	\begin{align} \label{asymptotic distribution of QMLE}
		\sqrt{T}(\widehat{\bm{\theta}} -\bm{\theta}_0 ) \rightarrow_d N(0_{n\times 1},\mathcal{J}^{-1} \mathcal{I} \mathcal{J}^{-1}), 
	\end{align}
	where $\mathcal{I} = E(P_{0t}^\prime S_1 P_{0t})$ and $\mathcal{J} = E(P_{0t}^\prime S_2 P_{0t})$ with $S_1$, $S_2$ and $P_{0t}$ block structured by \eqref{I_2 and J_2}. 
	By matrix multiplication, it can be directly obtained that the asymptotic covariance matrix of $\sqrt{T}(\widehat{\bm \alpha}- \bm \alpha_0)$ is $\Xi$. 
	This together with \eqref{asymptotic distribution of QMLE} implies that (a) holds. 
	For (b), note that $\widehat{\bm \theta} = (\widehat{\bm \alpha}^\prime , \widehat{\bm \sigma}^\prime)^\prime$ satisfies that 
	\begin{align*}
		0 = \sum_{t=1}^{T} \frac{\partial \widetilde{l}_t(\widehat{\bm \theta})}{\partial \bm\sigma}
		= \frac{1}{2}D^\prime \sum_{t=1}^{T} \text{vec}\left(\Sigma^{-1}(\widehat{\bm \sigma})-\Sigma^{-1}(\widehat{\bm \sigma})\widetilde{\bm{\varepsilon}}_t(\widehat{\bm \alpha}) \widetilde{\bm{\varepsilon}}_t^\prime(\widehat{\bm \alpha}) \Sigma^{-1}(\widehat{\bm \sigma})\right), 
	\end{align*}
	where $D$ is a duplication matrix and $\widetilde{\bm{\varepsilon}}_t(\bm{\alpha}) =\bm{y}_t-\sum_{h=1}^{t-1}A_{h}(\bm{\alpha}) \bm{y}_{t-h}$. 
	This together with $\widehat{\bm \sigma}= \text{vech}(\widehat{\Sigma})$ and $\widehat{\Sigma}=\Sigma(\widehat{\bm \sigma})$ implies that 
	\begin{align*}
		\widehat{\Sigma} = \frac{1}{T} \sum_{t=1}^{T}\widetilde{\bm \varepsilon}_{t}(\widehat{\bm \alpha}) \widetilde{\bm \varepsilon}_{t}^\prime (\widehat{\bm \alpha}). 
	\end{align*}
	Moreover, by Proposition \ref{thmLS}, it holds that 
	\begin{align} \label{asymptotic distribution of Sigmahat1}
		\widehat{\Sigma}_{LS} = \frac{1}{T} \sum_{t=1}^{T}\widetilde{\bm \varepsilon}_{t}(\widehat{\bm \alpha}_{LS}) \widetilde{\bm \varepsilon}_{t}^\prime (\widehat{\bm \alpha}_{LS}) 
		\;\; \text{and} \;\; 
		\sqrt{T} ({\rm vec}(\widehat{\Sigma}_{LS})- {\rm vec}(\Sigma_{0}))\rightarrow_{d}N(0_{N^2\times 1},\mathcal{K}), 
	\end{align}
	where the asymptotic distribution depends on the consistency of $\widehat{\bm \alpha}_{LS}$. 
	Thus $\widehat{\Sigma}$ has the same asymptotic distribution as $\widehat{\Sigma}_{LS}$, and then (b) holds.  

	Next we show ($i$). 
	Recall that $\Xi_{1} = \mathcal{J}_{LS}^{-1}\mathcal{I}_{LS}\mathcal{J}_{LS}^{-1}$ with $\mathcal{I}_{LS} = E[(\partial \bm \varepsilon_{t}^\prime(\bm \alpha_0) / \partial \bm \alpha) \Sigma_{0} (\partial \bm \varepsilon_{t}(\bm \alpha_0)/ \partial \bm \alpha')]$ and $\mathcal{J}_{LS} = E[(\partial \bm \varepsilon_{t}^\prime(\bm \alpha_0) / \partial \bm \alpha) (\partial \bm \varepsilon_{t}(\bm \alpha_0)/ \partial \bm \alpha')]$, and $\Xi =\Big\{E [(\partial \bm \varepsilon_t^\prime (\bm \alpha_0)/\partial \bm \alpha) \Sigma_0^{-1} (\partial \bm \varepsilon_t (\bm \alpha_0)/\partial \bm \alpha^\prime) ]\Big\}^{-1} $. 
	Denote the random matrix 
	\begin{align*}
		\mathcal{R}_t=
		\left[
		\begin{matrix}
			A_t
			& -B_t\\
			-B_t
			& C_t\\
		\end{matrix}
		\right],
	\end{align*}
	where $A_t= (\partial \bm \varepsilon_{t}^\prime(\bm \alpha_0)/\partial \bm \alpha) \Sigma_0^{-1} (\partial \bm \varepsilon_{t}(\bm \alpha_0)/\partial \bm \alpha^\prime) $, $B_t=(\partial \bm \varepsilon_{t}^\prime(\bm \alpha_0)/\partial \bm \alpha) (\partial \bm \varepsilon_{t}(\bm \alpha_0)/\partial \bm \alpha^\prime)$, and $C_t=(\partial \bm \varepsilon_t^\prime (\bm \alpha_0)/\partial \bm \alpha) \Sigma_0 (\partial \bm \varepsilon_t (\bm \alpha_0)/\partial \bm \alpha^\prime)$. 
	Then it holds that $\Xi_{1}- \Xi = (EB_t)^{-1} (EC_t) (EB_t)^{-1} - (EA_t)^{-1}$, 
	and $(EB_t) (\Xi_{1}- \Xi) (EB_t) = EC_t - (EB_t) (EA_t)^{-1} (EB_t)$ is the Schur complement of the block $EA_t$ of the block matrix $ER_t$. 
	Note that by Lemma \ref{lemma_postive_J_LS} and its proof, matrices $EA_t$, $EB_t$ and $EC_t$ are positive definite. 
	Thus using the property of the Schur complement of a block matrix, to show that $\Xi_{1}- \Xi$ is positive semi-definite, it suffices to show that $ER_t$ is positive semi-definite. 
	Let $\bm c = ( \bm c_1^\prime , \bm c_2^\prime )^\prime$ be an arbitrary constant vector with $\bm c_1, \bm c_2 \in \mathbb{R}^{n_\alpha}$. 
	We have that $\bm c^\prime \mathcal{R}_t \bm c = \bm c_1 ^\prime A_t \bm c_1 -2\bm c_1 ^\prime B_t \bm c_2+ \bm c_2 ^\prime C_t \bm c_2$. 
	Denote $\widetilde{\bm c}_1 = \Sigma_{0}^{-1/2} (\partial \bm \varepsilon_{t}(\bm \alpha_0)/ \partial \bm \alpha') \bm c_1$ and $\widetilde{\bm c}_2 = \Sigma_{0}^{1/2} (\partial \bm \varepsilon_{t}(\bm \alpha_0)/ \partial \bm \alpha') \bm c_2$, 
	it then follows that $\bm c^\prime \mathcal{R}_t \bm c = \widetilde{\bm c}_1^\prime  \widetilde{\bm c}_1 - 2 \widetilde{\bm c}_1^\prime  \widetilde{\bm c}_2 +\widetilde{\bm c}_2^\prime  \widetilde{\bm c}_2 = (\widetilde{\bm c}_1 - \widetilde{\bm c}_2)^\prime (\widetilde{\bm c}_1 - \widetilde{\bm c}_2) \geq 0$ a.s. 
	Thus $\bm c^\prime E\mathcal{R}_t \bm c \geq 0$, which implies that $E\mathcal{R}_t$ is positive semi-definite. 
	As a result, $\Xi_{1}- \Xi$ is positive semi-definite. 
	Moreover, it is obvious that $\Xi_{1} = \Xi$ holds if and only if $\Sigma_{0}=\sigma^2 I_{N}$ for some $\sigma^2>0$. 

	Note that ($ii$) can be verified directly by (b) and \eqref{asymptotic distribution of Sigmahat1}. 
	We accomplish the proof of this proposition. 
\end{proof}

\subsection{Proofs of Remark \ref{order-condition} and Theorem \ref{thmBIC}}
\begin{proof}[\bf Proof of Remark \ref{order-condition}]
	By ($i$) and ($ii$) in the proof of Lemma \ref{lemma_identifiability}, Assumption \ref{assum_identifiability} ensures the orders $r$ and $s$ non-degenerate. 

	Next we show that the condition $G_{0p} \neq \sum_{i=1}^{r} G_{0,p+i} + \sum_{j=1}^{s} G_{0,p+r+2j-1}$ guarantees the order $p$ non-decreasing. 
	Note that the order $p$ will reduce to $p-1$ if and only if there exist $\{\lambda_{\star i}\}_{i=1}^{r}$, $\{\gamma_{\star j}\}_{j=1}^{s}$, $\{\varphi_{\star j}\}_{j=1}^{s}$ and $\{G_{\star k}\}_{k=p}^{p+r+2s}$ 
	such that the following equations holds: 
	\begin{align}
		&G_{0p} = \sum_{i=1}^{r} \lambda_{\star i} G_{\star, p+i-1} 
			   + \sum_{j=1}^{s} \gamma_{\star j} \left[\cos(\varphi_{\star j}) G_{\star, p+r+2j-2} + \sin(\varphi_{\star j}) G_{\star, p+r+2j-1}\right], \label{eqs h equal p}\\
		&\sum_{i=1}^{r} \lambda_{0i}^{h-p} G_{0,p+i} 
	 + \sum_{j=1}^{s} \gamma_{0j}^{h-p} \left[\cos((h-p) \varphi_{0j}) G_{0,p+r+2j-1} + \sin((h-p) \varphi_{0j}) G_{0,p+r+2j}\right] \notag\\
		&\mathrel{\phantom{G_{0p}}}= \sum_{i=1}^{r} \lambda_{\star i}^{h-p+1} G_{\star, p+i-1} 
		                        + \sum_{j=1}^{s} \gamma_{\star j}^{h-p+1} \left[\right.\cos((h-p+1) \varphi_{\star j}) G_{\star, p+r+2j-2}\\
								& \quad \quad\quad\quad \quad\quad\quad+ \sin((h-p+1) \varphi_{\star j}) G_{\star, p+r+2j-1}\left.\right] \notag\\  
		&\mathrel{\phantom{G_p}}\text{for all} \;\; h \geq p+1. \label{eqs h more than p}
	\end{align}
	By ($i$) and ($ii$) in the proof of Lemma \ref{lemma_identifiability}, \eqref{eqs h more than p} implies that 
	$\lambda_{\star i} = \lambda_{0i}$, $\gamma_{\star j} = \gamma_{0j}$, $\varphi_{\star j} = \varphi_{0j}$, $\lambda_{\star i} G_{\star, p+i-1} = G_{0,p+i}$, and 
	\begin{align} \label{eq cos(m+1) and cos(m)}
		&\gamma_{\star j} \left[\cos((h-p+1) \varphi_{\star j}) G_{\star, p+r+2j-2} + \sin((h-p+1) \varphi_{\star j}) G_{\star, p+r+2j-1}\right] \notag\\
		= &\cos((h-p) \varphi_{0j}) G_{0,p+r+2j-1} + \sin((h-p) \varphi_{0j}) G_{0,p+r+2j}, 
	\end{align}
	hold for all $1 \leq i \leq r$, $1 \leq j \leq s$ and $h \geq p+1$. 
	Furthermore, using the fact $a \cos((m+1)\varphi) + b \sin((m+1)\varphi) = (a \cos(\varphi) + b \sin(\varphi)) \cos(m \varphi) + (b \cos(\varphi) - a \sin(\varphi)) \sin(m \varphi)$, 
	\eqref{eq cos(m+1) and cos(m)} implies that 
	\begin{align*}
		&\gamma_{\star j} (\cos(\varphi_{\star j}) G_{\star, p+r+2j-2} + \sin(\varphi_{\star j}) G_{\star, p+r+2j-1}) \cos((h-p) \varphi_{\star j}) \\
		+ &\gamma_{\star j} (\cos(\varphi_{\star j}) G_{\star, p+r+2j-1} - \sin(\varphi_{\star j}) G_{\star, p+r+2j-2}) \sin((h-p) \varphi_{\star j}) \\
		= &G_{0,p+r+2j-1} \cos((h-p) \varphi_{0j}) + G_{0,p+r+2j} \sin((h-p) \varphi_{0j}), 
	\end{align*}
	for all $1 \leq j \leq s$ and $h \geq p+1$, and then by $\gamma_{\star j} = \gamma_{0j}$ and $\varphi_{\star j} = \varphi_{0j}$, we have that 
	\begin{align}
		&\cos(\varphi_{\star j}) G_{\star, p+r+2j-2} + \sin(\varphi_{\star j}) G_{\star, p+r+2j-1} = \gamma_{\star j}^{-1} G_{0,p+r+2j-1} 
		\;\; \text{and} \label{eq G cos + G sin}\\ 
		& \cos(\varphi_{\star j}) G_{\star, p+r+2j-1} - \sin(\varphi_{\star j}) G_{\star, p+r+2j-2} = \gamma_{\star j}^{-1} G_{0,p+r+2j}, \notag
	\end{align}
	for all $1 \leq j \leq s$ and $h \geq p+1$. 
	Above all, by $\lambda_{\star i} G_{\star, p+i-1} = G_{0,p+i}$, \eqref{eq G cos + G sin} and \eqref{eqs h equal p}, it can be obtained that 
	\begin{align*}
		G_{0p} = \sum_{i=1}^{r} G_{0,p+i} 
			     + \sum_{j=1}^{s} G_{0,p+r+2j-1}. 
	\end{align*}
	As a result, if $G_{0p} \neq \sum_{i=1}^{r} G_{0,p+i} + \sum_{j=1}^{s} G_{0,p+r+2j-1}$, then \eqref{eqs h equal p}--\eqref{eqs h more than p} do not hold, which implies that the order $p$ is irreducible. 
\end{proof}

\begin{proof}[\bf Proof of Theorem~\ref{thmBIC}]
	Recall that $\bm{\theta} =(\bm{\alpha}^{\prime},\bm{\sigma}^{\prime})^\prime$. 
	For $\iota=(p,r,s)$, $\iota^*=(p^*,r^*,s^*) \in \mathscr{M} = \{\iota \mid 0\le p \le p_{\rm max}, 0\le r\le r_{\rm max} ,0 \le s\le s_{\rm max}\}$, 
	denote $\Theta_{\iota}$ (or $\Theta_{\iota^*}$) as the parameter space of $\bm{\theta}$ with the order set to $\iota$ (or $\iota^*$), and 
	$\mathcal{A}_{\iota}$ (or $\mathcal{A}_{\iota^*}$) as the parameter space of $\bm{\alpha}$ with the order set to $\iota$ (or $\iota^*$). 
	Let $\widehat{\bm \alpha}_{\iota}$ (or $\widehat{\bm \alpha}_{\iota^*}$) be the LSE $\widehat{\bm \alpha}_{1}$ or QMLE $\widehat{\bm \alpha}_{2}$ with the order set to $\iota$ (or $\iota^*$). 
	To prove the selection consistency of the proposed ${\rm BIC}$ in \eqref{BIC}, it suffices to show that the following result holds for any $\iota \neq \iota^*$: 
	\begin{align}\label{BICconsistency}
		\lim_{T \to \infty}P({\rm BIC}(\iota)-{\rm BIC}(\iota^*)>0)=1.
	\end{align}
	By \eqref{BIC}, it holds that 
	\begin{align} \label{BIC_iota - BIC_iota*}
		{\rm BIC}(\iota)-{\rm BIC}(\iota^*)
		=& T\left(\ln|\widehat{\Sigma}(\widehat{\bm \alpha}_{\iota})|-\ln | \widehat{\Sigma}(\widehat{\bm \alpha}_{\iota^*})|\right) + (n(\iota)-n(\iota^*)) \ln T, 
	\end{align}
	where $\widehat{\Sigma}(\bm \alpha) = T^{-1} \sum_{t=1}^{T} \widetilde{\bm{\varepsilon}}_t(\bm \alpha) \widetilde{\bm{\varepsilon}}_t^\prime(\bm \alpha)$ with $\widetilde{\bm{\varepsilon}}_t(\bm{\alpha}) =\bm{y}_t-\sum_{h=1}^{t-1}A_{h}(\bm{\alpha}) \bm{y}_{t-h}$. 
	By Lemma \ref{lemma_uni_max_MLE}, we have that 
	\begin{align*}
		\bm{\theta}_0=(\bm{\alpha}_0^\prime, \bm{\sigma}_0^\prime)^\prime = \argmin_{\bm{\theta}\in \Theta_{\iota^*}} \left\{\ln|\Sigma(\bm \sigma)|+E [\bm \varepsilon_t(\bm \alpha)^\prime \Sigma^{-1}(\bm \sigma) \bm \varepsilon_t(\bm \alpha)]\right\}, 
	\end{align*}
	with $\Sigma_{0} = \Sigma(\bm \sigma_0)$. 
	In addition, let $\mathring{\Sigma} = \Sigma(\mathring{\bm\sigma})$ and 
	\begin{align*}
		\mathring{\bm{\theta}}=(\mathring {\bm \alpha}^\prime,\mathring{\bm\sigma}^\prime)^\prime =\argmin_{\bm{\theta}\in \Theta_{\iota}} \left\{\ln|\Sigma(\bm \sigma)|+E [\bm \varepsilon_t(\bm \alpha)^\prime \Sigma^{-1}(\bm \sigma) \bm \varepsilon_t(\bm \alpha)]\right\}.  
	\end{align*}
	
	To verify \eqref{BICconsistency}, we next consider two cases. 
	
	Case I (overfitting): $p\ge p^*, r\ge r^*, s\ge  s^*$, and at least one inequality holds. 
	% For this case, we set $G_{0,p+i}=0, G_{0,p+r+2j-1}=0, G_{0,p+r+2j}=0$ for $r^*<i\le r$ and $s^*<j\le s$. 
	Note that $n(\iota)-n(\iota^*) > 0$ in Case I, which implies that $(n(\iota)-n(\iota^*)) \ln T \to \infty$ as $T \to \infty$. 
	Thus by \eqref{BIC_iota - BIC_iota*}, to establish that \eqref{BICconsistency} holds for Case I, it suffices to show that  
	\begin{align*}
		T(\ln|\widehat{\Sigma}(\widehat{\bm \alpha}_{\iota})|-\ln | \widehat{\Sigma}(\widehat{\bm \alpha}_{\iota^*})|) = O_{p}(1) 
		\;\; \text{as} \;\; T \to \infty. 
	\end{align*}
	Rewrite $\ln|\widehat{\Sigma}(\widehat{\bm \alpha}_{\iota})|-\ln | \widehat{\Sigma}(\widehat{\bm \alpha}_{\iota^*})|$ as follows, 
	\begin{align} \label{ln|Sigmahat(alphabreve)| - ln|Sigmahat(alphahat)|}
		\ln|\widehat{\Sigma}(\widehat{\bm \alpha}_{\iota})|- \ln|\widehat{\Sigma}(\widehat{\bm \alpha}_{\iota^*})| 
		= &\left(\ln|\widehat{\Sigma}(\widehat{\bm \alpha}_{\iota})| - \ln |\Sigma(\mathring{\bm \alpha})|\right) 
		+ \left(\ln |\Sigma(\mathring{\bm \alpha})|- \ln |\Sigma(\bm \alpha_0)|\right) \notag\\
		&+ \left(\ln |\Sigma(\bm \alpha_0)| - \ln|\widehat{\Sigma}(\widehat{\bm \alpha}_{\iota^*})|\right),  
	\end{align}
	where $\Sigma(\bm \alpha) = T^{-1} \sum_{t=1}^{T} \bm{\varepsilon}_t(\bm \alpha) \bm{\varepsilon}_t^\prime(\bm \alpha)$ with $\bm{\varepsilon}_t(\bm{\alpha}) = \bm y_t - \sum_{h=1}^{\infty}A_{h}(\bm{\alpha})\bm y_{t-h}$. 
	Notice that the model with order $\iota=(p,r,s)$ in Case I corresponds to a bigger model. 
	It holds that $\bm \varepsilon_t(\bm \alpha_0) = \bm \varepsilon_t(\mathring{\bm \alpha})$ for all $t$, which implies that $\Sigma(\bm \alpha_0) = \Sigma(\mathring{\bm \alpha})$ and then the second parenthesis in \eqref{ln|Sigmahat(alphabreve)| - ln|Sigmahat(alphahat)|} is zero. 
	% Moreover, the first parenthesis is $O_{p}(T^{-1})$ implies that the last parenthesis is $O_{p}(T^{-1})$ in \eqref{ln|Sigmahat(alphabreve)| - ln|Sigmahat(alphahat)|} because $\widehat{\bm \alpha}_{\iota} = \widehat{\bm \alpha}_{\iota^*}$ and $\mathring{\bm \alpha} = \bm \alpha_0$ when $\iota = \iota^*$. 
	% Hence by \eqref{ln|Sigmahat(alphabreve)| - ln|Sigmahat(alphahat)|}, we are left to show that  
	Below we will show that the following result holds when $p\ge p^*, r\ge r^*, s\ge  s^*$: 
	\begin{align} \label{ln|Sigmahat| - ln|Sigma|}
		% \ln|\widehat{\Sigma}(\widehat{\bm \alpha}_{\iota^*})| - \ln |\Sigma(\bm \alpha_0)| = O_{p}(T^{-1}) 
		% \;\; \text{and} \;\; 
		\ln|\widehat{\Sigma}(\widehat{\bm \alpha}_{\iota})| - \ln |\Sigma(\mathring{\bm \alpha})| = O_{p}(T^{-1}),
	\end{align}
	which implies that the first and last parenthesis are $O_{p}(T^{-1})$. 
	Since it holds that 
	\begin{equation} \label{ln|Sigmahat| - ln|Sigma0|}
		\ln|\widehat{\Sigma}(\widehat{\bm \alpha}_{\iota})| - \ln |\Sigma(\mathring{\bm \alpha})| 
		= \ln \left(1 + \frac{|\widehat{\Sigma}(\widehat{\bm \alpha}_{\iota})| - |\Sigma(\mathring{\bm \alpha})|}{|\Sigma(\mathring{\bm \alpha})|}\right), 
		% = \ln (1+O_p(T^{-1})) = O_{p}(T^{-1}), 
	\end{equation}
	we only need to verify that (a) $|\widehat{\Sigma}(\widehat{\bm \alpha}_{\iota})| - |\Sigma(\mathring{\bm \alpha})| = O_{p}(T^{-1})$; and (b) $|\Sigma(\mathring{\bm \alpha})| = O_p(1)$. 
	In addition, as Case I corresponds to a bigger model, we set $G_{0,p+i}=0, G_{0,p+r+2j-1}=0, G_{0,p+r+2j}=0$ for $r^*<i\le r$ and $s^*<j\le s$ in this case. 

	We first show (a). 
	It suffices to show that $\widehat{\Sigma}(\widehat{\bm \alpha}_{\iota}) - \Sigma(\mathring{\bm \alpha}) = O_{p}(T^{-1})$ holds. 
	By Lemma \ref{lemma_initial_effect_LS}, we have 
	\begin{equation} \label{Sigmahathat - Sigmahat}
		\widehat{\Sigma}(\widehat{\bm \alpha}_{\iota}) - \Sigma(\widehat{\bm \alpha}_{\iota}) = O_{p}(T^{-1}).
	\end{equation} 
	Then we are left to establish that 
	\begin{equation} \label{Sigmahat - Sigma0}
		\Sigma(\widehat{\bm \alpha}_{\iota}) -\Sigma(\mathring{\bm \alpha}) = O_{p}(T^{-1}). 
	\end{equation}
	Denote $\widehat{\Delta}_h = A_{h}(\widehat{\bm \alpha}_{\iota}) - A_{h}(\mathring{\bm \alpha})$. 
	It can be shown that 
	\begin{align*}
		T \left[\Sigma(\widehat{\bm \alpha}_{\iota}) -\Sigma(\mathring{\bm \alpha})\right] 
		= & \sum_{t=1}^{T} \left[\bm \varepsilon_{t}(\widehat{\bm \alpha}_{\iota})\bm \varepsilon_{t}^\prime(\widehat{\bm \alpha}_{\iota}) - \bm \varepsilon_{t}(\mathring{\bm \alpha}) \bm \varepsilon_{t}^\prime(\mathring{\bm \alpha})\right]\\
		= & \sum_{t=1}^{T} \left[\left(\bm \varepsilon_{t}(\mathring{\bm \alpha}) - \sum_{h=1}^{\infty} \widehat{\Delta}_h \bm y_{t-h}\right)  \left(\bm \varepsilon_{t}(\mathring{\bm \alpha}) - \sum_{h=1}^{\infty} \widehat{\Delta}_h \bm y_{t-h}\right)^\prime
		- \bm \varepsilon_{t}(\mathring{\bm \alpha}) \bm \varepsilon_{t}^\prime(\mathring{\bm \alpha})\right]\\
		= & \sum_{t=1}^{T} \left[ -\sum_{h=1}^{\infty} \bm \varepsilon_{t}(\mathring{\bm \alpha}) \bm y_{t-h}^\prime \widehat{\Delta}_h^\prime  - \sum_{h=1}^{\infty}  \widehat{\Delta}_h \bm y_{t-h} \bm \varepsilon_{t}^\prime(\mathring{\bm \alpha}) + \sum_{h_1=1}^{\infty}\sum_{h_2=1}^{\infty}   \widehat{\Delta}_{h_1} \bm y_{t-h_1} \bm y_{t-h_2}^\prime \widehat{\Delta}_{h_2}^\prime \right]\\
		= & - \left(\sum_{h=1}^{\infty} \sqrt{T}\widehat{\Delta}_h \frac{1}{\sqrt{T}} \sum_{t=1}^{T}  \bm y_{t-h} \bm \varepsilon_{t}^\prime(\mathring{\bm \alpha})\right)^\prime
		-\sum_{h=1}^{\infty} \sqrt{T}\widehat{\Delta}_h \frac{1}{\sqrt{T}} \sum_{t=1}^{T}  \bm y_{t-h} \bm \varepsilon_{t}^\prime(\mathring{\bm \alpha})\\
		&+ \sum_{h_1=1}^{\infty} \sum_{h_2=1}^{\infty}\sqrt{T}\widehat{\Delta}_{h_1}\frac{1}{T} \sum_{t=1}^{T} \bm y_{t-h_2}  \bm y_{t-h_2}^\prime \sqrt{T} \widehat{\Delta}_{h_2}^\prime.
	\end{align*}
	Since that $T^{-1/2} \sum_{t=1}^{T}  \bm y_{t-h} \bm \varepsilon_{t}^\prime(\mathring{\bm \alpha}) = O_p(1)$ holds by the martingale CLT theorem under Assumptions \ref{assum_stationarity_y}--\ref{assum_stationarity_epsilon} and $T^{-1}\sum_{t=1}^{T} \bm y_{t-h_2}  \bm y_{t-h_2}^\prime =O_p(1)$ holds under Assumption \ref{assum_stationarity_y}, 
	we only need to verify that $\sqrt{T}\widehat{\Delta}_h = O_p(1)$ for $1 \leq h \leq p$ and $\sqrt{T}\widehat{\Delta}_h = O_p(h\rho^{h})$ for $h > p$. 
	Below we will focus on the cases $p\ge p^*, r=r^*, s=s^*=0 $ and $p=p^*, r>r^*, s=s^*=0 $, and the results for other cases can be verified similarly. 

	When $p\ge p^*, r=r^*$ and $s=s^*=0$, we have that the model is identifiable from Remark \ref{order-condition} and its proof, and then we can verify $\sqrt{T}(\widehat{\bm \alpha}_{\iota}- \mathring{\bm \alpha}) = O_p(1)$ with similar arguments in the proof of Theorems \ref{thmLS}--\ref{thmQMLE}. 
	Under Assumption \ref{assum_space}, it follows that 
	\begin{align*}
		\sqrt{T}\widehat{\Delta}_h 
		=&  \sqrt{T}(A_{h}(\widehat{\bm \alpha}_{\iota}) - A_{h}(\mathring{\bm \alpha})) 
		= \sqrt{T} (\widehat{G}_{\iota,h} - \mathring{G}_{h})
		= O_p(1) \;\; \text{for} \;\; 1 \leq h \leq p, \;\; \text{and} \\
		\sqrt{T}\widehat{\Delta}_h 
		=&  \sqrt{T}(A_{h}(\widehat{\bm \alpha}_{\iota}) - A_{h}(\mathring{\bm \alpha}))
		= \sqrt{T}\sum_{i=1}^{r} (\widehat{\lambda}_{\iota,i}^{h-p} \widehat{G}_{\iota,p+i}- \mathring{\lambda}_{i}^{h-p} \mathring{G}_{p+i} )\\
		=& \sum_{i=1}^{r}\left[\widehat{\lambda}_{\iota,i}^{h-p} \sqrt{T} (\widehat{G}_{\iota,p+i} - \mathring{G}_{p+i}) + \sqrt{T}(\widehat{\lambda}_{\iota,i}^{h-p} - \mathring{\lambda}_{i}^{h-p}) \mathring{G}_{p+i}\right]\\
		=& \sum_{i=1}^{r}\left[\widehat{\lambda}_{\iota,i}^{h-p} \sqrt{T} (\widehat{G}_{\iota,p+i} - \mathring{G}_{p+i}) + \sqrt{T}(\widehat{\lambda}_{\iota,i} - \mathring{\lambda}_{i})\sum_{k=0}^{h-p-1}(\widehat{\lambda}_{\iota,i}^{h-p-1-k}\mathring{\lambda}_{i}^k) \mathring{G}_{p+i}\right]\\
		=& O_p(h\rho^{h}) \;\; \text{for} \;\; h > p, 
	\end{align*}
	where $\mathring{\lambda}_{i}$ and $\mathring{G}_{i}$ are defined corresponding to $\mathring{\bm \alpha}$. 

	When $p=p^*, r>r^*$ and $s=s^*=0 $, the model is non-identifiable in estimation due to the non-identification of parameters $\{\lambda_{i}: r^*<i\le r\}$. % the terms $\lambda_{0i}^{j} G_{0,p+i} = 0$ for $r^*<i\le r$ and $j \geq 1$. 
	However, $\{\lambda_{i}, G_{j}: 1\le i\le r^*, 1\le j\le p+r\}$ are identifiable in estimation under Assumption \ref{assum_space}, which can be obtained by Remark \ref{order-condition} and its proof. 
	Thus similar to the proof of Theorems \ref{thmQMLE}, we can show that $\sqrt{T}(\widehat{\lambda}_{\iota,i} - \mathring{\lambda}_{i})=O_p(1)$ for $1\le i\le r^*$ and $\sqrt{T} (\widehat{G}_{\iota,i} - \mathring{G}_{i}) = O_p(1)$ for $1\le i\le p+r$, where $\widehat{\lambda}_{\iota,i}$ (or $\widehat{G}_{\iota,i}$) is the LSE or QMLE of $\lambda_i$ (or $G_{i}$) with the order set to $\iota$. 
	Recall that $G_{0i}=0$ for $p+r^*<i\le p+r$. 
	Note that $\mathring{\lambda}_{i} = \lambda_{0i}$ for $1 \leq i \leq r^*$ and $\mathring{G}_{i} = G_{0i}$ for $1 \leq i \leq p+r$ in this case. 
	Under Assumption \ref{assum_space}, it then follows that 
	\begin{align*}
		\sqrt{T}\widehat{\Delta}_h 
		=&  \sqrt{T}(A_{h}(\widehat{\bm \alpha}_{\iota}) - A_{h}(\mathring{\bm \alpha})) 
		= \sqrt{T} (\widehat{G}_{\iota,h} - G_{0h})
		= O_p(1) \;\; \text{for} \;\; 1 \leq h \leq p, \;\; \text{and} \\
		\sqrt{T}\widehat{\Delta}_h 
		=& \sqrt{T}(A_{h}(\widehat{\bm \alpha}_{\iota}) - A_{h}(\mathring{\bm \alpha})) \\
		=& \sqrt{T}\sum_{i=1}^{r^*} (\widehat{\lambda}_{\iota,i}^{h-p} \widehat{G}_{\iota,p+i}- \lambda_{0i}^{h-p} G_{0,p+i} ) 
		+ \sqrt{T}\sum_{i=r^*+1}^{r} \widehat{\lambda}_{\iota,i}^{h-p} \widehat{G}_{\iota,p+i}\\
		=& \sum_{i=1}^{r^*}\left[\widehat{\lambda}_{\iota,i}^{h-p} \sqrt{T} (\widehat{G}_{\iota,p+i} - G_{0,p+i}) + \sqrt{T}(\widehat{\lambda}_{\iota,i}^{h-p} - \lambda_{0i}^{h-p}) G_{0,p+i}\right] \\
		&+ \sqrt{T}\sum_{i=r^*+1}^{r} \widehat{\lambda}_{\iota,i}^{h-p} (\widehat{G}_{\iota,p+i} - G_{0,p+i})\\
		=& \sum_{i=1}^{r^*}\left[\widehat{\lambda}_{\iota,i}^{h-p} \sqrt{T} (\widehat{G}_{\iota,p+i} - G_{0,p+i}) + \sqrt{T}(\widehat{\lambda}_{\iota,i} - \lambda_{0i})\sum_{k=0}^{h-p-1}(\widehat{\lambda}_{\iota,i}^{h-p-1-k}\lambda_{0i}^k) G_{0,p+i}\right]\\
		&+ \sum_{i=r^*+1}^{r} \widehat{\lambda}_{\iota,i}^{h-p} \sqrt{T} (\widehat{G}_{\iota,p+i} - G_{0,p+i})\\
		=& O_p(h\rho^{h}) \;\; \text{for} \;\; h > p. 
	\end{align*}
	Thus we have established \eqref{Sigmahat - Sigma0}. This together with \eqref{Sigmahathat - Sigmahat} implies that (a) holds. 

	Next we show (b). 
	By Lemma \ref{lemma_three consistent_LS}($i$) and Assumption \ref{assum_stationarity_y}, we can obtain that 
	\begin{align} \label{Sigma(alpha0) and Sigma(alphacircle)}
		&\Sigma(\bm \alpha_0) 
		= E[\bm{\varepsilon}_t(\bm \alpha_0) \bm{\varepsilon}_t^\prime(\bm \alpha_0)] + o_{p}(1) 
		= \Sigma_0 + o_{p}(1). % \;\; \text{and} \notag\\
		% &\Sigma(\mathring{\bm \alpha}) 
		% = E[\bm{\varepsilon}_t(\mathring{\bm \alpha}) \bm{\varepsilon}_t^\prime(\mathring{\bm \alpha})] + o_{p}(1) 
		% = E[\bm{\varepsilon}_t(\bm \alpha_0) \bm{\varepsilon}_t^\prime(\bm \alpha_0)] + o_{p}(1) 
		% = \Sigma_0 + o_{p}(1). 
	\end{align}
	This together with $\Sigma(\bm \alpha_0) = \Sigma(\mathring{\bm \alpha})$ that we mention above as well as Assumption \ref{assum_stationarity_epsilon}, implies that (b) holds. 
	Thus by \eqref{ln|Sigmahat| - ln|Sigma0|} and (a)--(b), we have 
	$$
		\ln|\widehat{\Sigma}(\widehat{\bm \alpha}_{\iota})| - \ln |\Sigma(\mathring{\bm \alpha})| 
		% = \ln \left(1 + \frac{|\widehat{\Sigma}(\widehat{\bm \alpha}_{\iota})| - |\Sigma(\mathring{\bm \alpha})|}{|\Sigma(\mathring{\bm \alpha})|}\right) 
		= \ln (1+O_p(T^{-1})) = O_{p}(T^{-1}), 
	$$
	that is \eqref{ln|Sigmahat| - ln|Sigma|} holds. 
	As a result, \eqref{BICconsistency} holds for Case I. 
	
	Case II (misspecification): $p< p^*$ or $r< r^*$ or $s< s^*$. 
	By \eqref{BIC_iota - BIC_iota*}, it holds that 
	\begin{align} \label{BIC_iota - BIC_iota*, mis}
		{\rm BIC}(\iota)-{\rm BIC}(\iota^*)
		=& T\left(\ln|\widehat{\Sigma}(\widehat{\bm \alpha}_{\iota})|-\ln | \widehat{\Sigma}(\widehat{\bm \alpha}_{\iota^*})|\right) + (n(\iota)-n(\iota^*)) \ln T \notag\\
		=& T\Big[
		\left(\ln|\widehat{\Sigma}(\widehat{\bm \alpha}_{\iota})| - \ln |E(\bm \varepsilon_{t}(\mathring{\bm \alpha}) \bm \varepsilon_{t}^\prime(\mathring{\bm \alpha}))|\right)  \notag\\
		&+ \left(\ln |E(\bm \varepsilon_{t}(\mathring{\bm \alpha}) \bm \varepsilon_{t}^\prime(\mathring{\bm \alpha}))|- \ln |E(\bm \varepsilon_{t}(\bm \alpha_0) \bm \varepsilon_{t}^\prime(\bm \alpha_0))|\right) \notag\\
		&+ \left(\ln |E(\bm \varepsilon_{t}(\bm \alpha_0) \bm \varepsilon_{t}^\prime(\bm \alpha_0))| - \ln|\widehat{\Sigma}(\widehat{\bm \alpha}_{\iota^*})|\right) 
		\Big] \notag\\
		&+ (n(\iota)-n(\iota^*)) \ln T. 
	\end{align}
	Hence to establish \eqref{BICconsistency}, it suffices to show 
	(c) $\ln |E(\bm \varepsilon_{t}(\mathring{\bm \alpha}) \bm \varepsilon_{t}^\prime(\mathring{\bm \alpha}))|- \ln |E(\bm \varepsilon_{t}(\bm \alpha_0) \bm \varepsilon_{t}^\prime(\bm \alpha_0))| > \delta$ for some constant $\delta > 0$; 
	(d) $\ln|\widehat{\Sigma}(\widehat{\bm \alpha}_{\iota^*}) - E(\bm \varepsilon_{t}(\bm \alpha_0) \bm \varepsilon_{t}^\prime(\bm \alpha_0))| = o_{p}(1)$ as $T \to \infty$; and 
	(e) $\ln|\widehat{\Sigma}(\widehat{\bm \alpha}_{\iota})| - \ln |E(\bm \varepsilon_{t}(\mathring{\bm \alpha}) \bm \varepsilon_{t}^\prime(\mathring{\bm \alpha}))| = o_{p}(1)$ as $T \to \infty$. 
	
	For (c), denote $p_m = \max\{p,p^*\}$, $r_m = \max\{r,r^*\}$, $s_m = \max\{s,s^*\}$ and $\iota_m = (p_m, r_m, s_m)$. 
	And let $\bm \alpha_0^+$ (or $\mathring{\bm \alpha}^+$) be a parameter vector with the order set to $\iota_m$, where $\bm \alpha_0$ (or $\mathring{\bm \alpha}$) being the subvector of $\bm \alpha_0^+$ (or $\mathring{\bm \alpha}^+$) and the extra parameters being zero. 
	Moreover, denote $\mathcal{A}_{\iota_m}$ as the parameter space of $\bm \alpha_{\iota_m}$, including the points $\bm\alpha_0^+$ and $\mathring{\bm \alpha}^+$.
%	Note that we can obtain that the model is identifiable in this case by Remark \ref{order-condition} as well as its proof. 
	% Note that we can easily generalize Lemma \ref{lemma_identifiability} and Remark \ref{order-condition} to following argument, 
	% if $p< p^*$ or $r< r^*$ or $s< s^*$, then $A_{h}(\mathring{\bm \alpha})\neq A_{h}(\bm \alpha_0)$ for all $h\ge 1$, and $\mathring{\bm \alpha}$ is still identifiable under Assumption \ref{assum_identifiability}.
	Since $E[\bm \varepsilon_t^\prime (\bm \alpha_{\iota^*}) \bm \varepsilon_{t}(\bm \alpha_{\iota^*})]$ has a unique minimum at $\bm \alpha_0$ on $\mathcal{A}_{\iota^*}$ by Lemma \ref{lemma_uni_max}, $E[\bm \varepsilon_t^\prime (\bm \alpha_{\iota_m}) \bm \varepsilon_{t}(\bm \alpha_{\iota_m})]$ has a unique minimum at $\bm \alpha_0^+$ on $\mathcal{A}_{\iota_m}$. 
	This together with $\bm\alpha_0^+, \mathring{\bm \alpha}^+ \in \mathcal{A}_{\iota_m}$, $E[\bm \varepsilon_t(\mathring{\bm \alpha}) \bm \varepsilon_t^\prime(\mathring{\bm \alpha})]=E[\bm \varepsilon_t(\mathring{\bm \alpha}^+) \bm \varepsilon_t^\prime(\mathring{\bm \alpha}^+)]$ and $E[\bm \varepsilon_t(\bm \alpha_0) \bm \varepsilon_t^\prime(\bm \alpha_0)] = E[\bm \varepsilon_t(\bm \alpha_0^+) \bm \varepsilon_t^\prime(\bm \alpha_0^+)]$, implies that the following result holds for some constant $\delta>0$: 
	$$ 
%		\label{sec1, ln|Sigmahat| - ln|Sigmahat|}
		\ln\left|E[\bm \varepsilon_t(\mathring{\bm \alpha}) \bm \varepsilon_t^\prime(\mathring{\bm \alpha})]\right|- \ln\left|E[\bm \varepsilon_t(\bm \alpha_0) \bm \varepsilon_t^\prime(\bm \alpha_0)]\right| 
		= \ln\left|E[\bm \varepsilon_t(\mathring{\bm \alpha}^+) \bm \varepsilon_t^\prime(\mathring{\bm \alpha}^+)]\right|- \ln\left|E[\bm \varepsilon_t(\bm \alpha_0^+) \bm \varepsilon_t^\prime(\bm \alpha_0^+)]\right| 
		> \delta. 
	$$
	Thus, (c) holds.
	Furthermore, (d) holds by Theorem \ref{thmLS}($i$).
	
	For (e), 
	we only need to show 
	(e1) $\widehat{\Sigma}(\widehat{\bm \alpha}_{\iota}) - \Sigma(\widehat{\bm \alpha}_{\iota}) = o_{p}(1)$; 
	(e2) $\Sigma(\widehat{\bm \alpha}_{\iota}) - \Sigma(\mathring{\bm \alpha})= o_{p}(1)$; and 
	(e3) $\Sigma(\mathring{\bm \alpha}) - E(\bm \varepsilon_{t}(\mathring{\bm \alpha}) \bm \varepsilon_{t}^\prime(\mathring{\bm \alpha})) = o_{p}(1)$ as $T \to \infty$. 
	It can be obtained that 
	(e1) holds by Lemma \ref{lemma_initial_effect_LS}, and 
	(e3) holds by the ergodic theorem under Assumption \ref{assum_stationarity_y}. 
	Then we are left to verify (e2). 
	Assume that $E(\bm \varepsilon_t^\prime(\bm \alpha_{\iota}) \bm \varepsilon_t (\bm \alpha_{\iota}))$ has a unique minimum at $\mathring{\bm \alpha}$ on $\mathcal{A}_{\iota}$. Similar to the proof of Theorem  \ref{thmLS}($i$), we can show that $\widehat{\bm \alpha}_{\iota} \to_p \mathring{\bm \alpha}$ as $T \to \infty$. By the mean value theorem, we have
	\begin{align*}
		\|\Sigma(\widehat{\bm \alpha}_{\iota}) - \Sigma(\mathring{\bm \alpha})\|
		\leq \left\| \frac{1}{T} \sum_{t=1}^{T} \frac{\partial \bm \varepsilon_{t}(\bar{\bm \alpha}_\iota) \bm \varepsilon_t^\prime(\bar{\bm \alpha}_\iota)}{\partial \bm \alpha} \right\|
		\left\|\widehat{\bm \alpha}_{\iota} -\mathring{\bm \alpha} \right\|
	\end{align*}
	where $\bar{\bm \alpha}_\iota$ lies between $\widehat{\bm \alpha}_{\iota}$ and $\mathring{\bm \alpha}$.
	This together with $\left\| T^{-1} \sum_{t=1}^{T} \partial \bm \varepsilon_{t}(\bar{\bm \alpha}_\iota) \bm \varepsilon_t^\prime(\bar{\bm \alpha}_\iota)/\partial \bm \alpha \right\|=O_p(1)$ by Lemma \ref{lemma_three consistent_LS} and $\left\|\widehat{\bm \alpha}_{\iota} -\mathring{\bm \alpha} \right\|=o_p(1)$ implies (e2) holds.
	Hence (e) holds. 
	Then by \eqref{BIC_iota - BIC_iota*, mis} and (c)--(e), we have 
	\begin{align*}
		{\rm BIC}(\iota)-{\rm BIC}(\iota^*)
		= T (o_p(1) + \delta) + O(\ln(T)) \to \infty  
		\;\; \text{as} \;\; T \to \infty. 		
	\end{align*}
	As a result, \eqref{BICconsistency} holds for Case II. 
	This completes the proof of Theorem \ref{thmBIC}. 
\end{proof}

\section{Derivatives and Notations}

Recall that $\bm{\varepsilon}_t(\bm{\alpha}) = \bm y_t - \sum_{h=1}^{\infty}A_{h}(\bm{\alpha})\bm y_{t-h}$ with $\bm{\alpha}=(\lambda_1,\ldots,\lambda_r, \bm{\eta}_1^\prime, \ldots ,\bm{\eta}_s^\prime, \bm{g}^\prime)^\prime$ and $\bm{\eta}_j = (\gamma_j, \varphi_j)$. 
By \eqref{model_VARinf}, $A_h(\bm{\alpha}) = G_h$ for $1\le h\le p$, and $A_h(\bm{\alpha}) = \sum_{k=p+1}^{d} \ell_{hk}(\bm \omega) G_k$ for $h \ge p+1$ with  
$\ell_{hk}(\bm{\omega})$ being the $(h,k)$-th entry of 
\begin{align*}
	L(\bm{\omega})
	=\left( 
	\begin{matrix}
		I_p & 0_{p\times(d-p)}\\
		0_{\infty\times p} & \bm{\ell}^{I}(\lambda_1)\cdots \bm{\ell}^I(\lambda_r) ~~ \bm{\ell}^{II}(\bm{\eta}_1)\cdots \bm{\ell}^{II}(\bm{\eta}_s) 
	\end{matrix}
	\right), 
\end{align*}
$\bm{\ell}^I(\lambda_i)=(\lambda_i,\lambda_i^2,\ldots)^\prime$ and 
\begin{align*}
	\bm{\ell}^{II}(\bm{\eta}_j)
	=\left(
	\begin{matrix}
		\gamma_j \cos(\varphi_j) &\gamma_j^2 \cos(2\varphi_j) & \cdots \\
		\gamma_j \sin(\varphi_j) &\gamma_j^2 \sin(2\varphi_j) & \cdots   
	\end{matrix}
	\right)^\prime.
\end{align*} 
Then the first derivatives of $\bm{\varepsilon}_t(\bm \alpha)$ are 
\begin{equation} \label{de_epsilon_alpha}
	\begin{aligned}
		&\frac{\partial \bm{\varepsilon}_t(\bm \alpha)}{\partial \lambda_i}
		=-G_{p+i} \sum_{h=p+1}^{\infty} (h-p) \lambda_i^{h-p-1} \bm{y}_{t-h},\\
		&\frac{\partial \bm{\varepsilon}_t(\bm \alpha)}{\partial \gamma_j}
		= \sum_{h=p+1}^{\infty}(h-p) \gamma_j^{h-p-1}\left\{-\cos[(h-p)\varphi_j]G_{p+r+2j-1}  -\sin[(h-p)\varphi_j]G_{p+r+2j}\right\}\bm{y}_{t-h},\\
		&\frac{\partial \bm{\varepsilon}_t(\bm \alpha)}{\partial \varphi_j}
		=  \sum_{h=p+1}^{\infty}(h-p) \gamma_j^{h-p}\left\{\sin[(h-p)\varphi_j]G_{p+r+2j-1} -\cos[(h-p)\varphi_j]G_{p+r+2j}\right\}\bm{y}_{t-h},\\
		&\frac{\partial \bm{\varepsilon}_t(\bm \alpha )}{\partial \bm g^\prime }= -\bm{z}_t^\prime\otimes I_N, 
	\end{aligned}
\end{equation}
where $1 \le i \le r$, $1 \le j \le s$, and $\bm z_{t} = \text{vec}(\bm z_{1t},\ldots ,\bm z_{dt})$ with $\bm{z}_{kt}=\sum_{h=1}^{\infty}\ell_{hk}\bm y_{t-h}$. 
In addition, $\partial \bm{\varepsilon}_t(\bm \alpha )/\partial \bm g^\prime$ is obtained by rewriting $\bm{\varepsilon}_t(\bm \alpha )$ as $\bm{\varepsilon}_t(\bm \alpha )=\bm{y}_t -  (\bm{z}_t^\prime\otimes I_N )\bm g$. 
Furthermore, the second derivatives of $\bm{\varepsilon}_t(\bm \alpha)$ are 
\begin{align} \label{de_epsilon_alpha^2}
		& \frac{\partial^2 \bm{\varepsilon}_t(\bm \alpha)}{\partial \lambda_i^2}
		=-G_{p+i} \sum_{h=p+2}^{\infty} (h-p)(h-p-1) \lambda_i^{h-p-2} \bm{y}_{t-h},\notag\\
		& \frac{\partial^2 \bm{\varepsilon}_t(\bm \alpha)}{\partial \gamma_j^2}
		=\sum_{h=p+2}^{\infty}(h-p)(h-p-1)\gamma_j^{h-p-2}\Big\{-\cos[(h-p)\varphi_j]G_{p+r+2j-1} \\ 
		& \qquad \qquad \qquad  \qquad \qquad \qquad \qquad  \qquad \qquad -\sin[(h-p)\varphi_j]G_{p+r+2j}\Big\}\bm{y}_{t-h},\notag\\
		&\frac{\partial^2 \bm{\varepsilon}_t(\bm \alpha)}{\partial \varphi_j^2}
		=\sum_{h=p+1}^{\infty}(h-p)^2 \gamma_j^{h-p}\left\{\cos[(h-p)\varphi_j]G_{p+r+2j-1} +\sin[(h-p)\varphi_j]G_{p+r+2j}\right\}\bm{y}_{t-h},\notag\\
		&\frac{\partial^2 \bm{\varepsilon}_t(\bm \alpha)}{\partial \gamma_j \partial \varphi_j}
		=\sum_{h=p+1}^{\infty}(h-p)^2\gamma_j^{h-p-1}\left\{\sin[(h-p)\varphi_j]G_{p+r+2j-1} -\cos[(h-p)\varphi_j]G_{p+r+2j}\right\}\bm{y}_{t-h},\notag\\
		&\frac{\partial^2 \bm{\varepsilon}_t(\bm \alpha )}{\partial \lambda_i \partial \bm{g}_{p+i}^\prime}
		=-\left[ \sum_{h=p+1}^{\infty} (h-p) \lambda_i^{h-p-1} \bm{y}_{t-h}\right]^\prime \otimes I_N,\notag\\
		&\frac{\partial^2 \bm{\varepsilon}_t(\bm \alpha )}{\partial \gamma_j \partial \bm{g}_{p+r+2j-1}^{\prime}}
		=-\left\{ \sum_{h=p+1}^{\infty} (h-p) \gamma_j^{h-p-1} \cos[(h-p)\varphi_j] \bm{y}_{t-h}\right\}^\prime \otimes I_N,\notag\\
		&\frac{\partial^2 \bm{\varepsilon}_t(\bm \alpha )}{\partial \gamma_j \partial \bm{g}_{p+r+2j}^\prime}
		=-\left\{ \sum_{h=p+1}^{\infty} (h-p) \gamma_j^{h-p-1} \sin[(h-p)\varphi_j] \bm{y}_{t-h}\right\}^\prime \otimes I_N,\notag\\
		&\frac{\partial^2 \bm{\varepsilon}_t(\bm \alpha )}{\partial \varphi_j \partial \bm{g}_{p+r+2j-1}^\prime}
		=\left\{ \sum_{h=p+1}^{\infty} (h-p) \gamma_j^{h-p} \sin[(h-p)\varphi_j] \bm{y}_{t-h}\right\}^\prime \otimes I_N,\notag\\
		&\frac{\partial^2 \bm{\varepsilon}_t(\bm \alpha )}{\partial \varphi_j \partial \bm{g}_{p+r+2j}^\prime}
		=-\left\{ \sum_{h=p+1}^{\infty} (h-p) \gamma_j^{h-p} \cos[(h-p)\varphi_j] \bm{y}_{t-h}\right\}^\prime \otimes I_N,
\end{align}
where $1 \le i \le r$, $1 \le j \le s$, and $\bm{g}_k=\text{vec}(G_k)$ for $1\le k\le d$. 
The other second derivatives of $\bm{\varepsilon}_t(\bm \alpha)$ are zeros.

Recall that $l_t(\bm{\theta}) =1/2\ln |\Sigma(\bm{\sigma})|+1/2 \bm{\varepsilon}_t^\prime (\bm{\alpha})\Sigma^{-1}(\bm{\sigma})\bm{\varepsilon}_t(\bm{\alpha})$ and 
$\bm{\theta}=(\bm{\alpha}^\prime,\bm{\sigma}^\prime)^\prime$. 
Note that $\text{vec}(\Sigma)=D \text{vech}(\Sigma)$, where $D$ is an $N^2\times N(N+1)/2$ duplication matrix. 
Then the first derivatives of $l_t(\bm{\theta})$ are 
\begin{equation} \label{de_l_theta}
\begin{aligned}
	\frac{\partial l_t(\bm \theta)}{\partial \bm{\alpha} }
	=& \frac{\partial \bm{\varepsilon}_t^\prime(\bm{\alpha})}{\partial \bm{\alpha}} \Sigma^{-1}(\bm \sigma)\bm{\varepsilon}_t(\bm{\alpha}),\\
	\frac{\partial l_t(\bm \theta)}{\partial \bm\sigma}
	=&\frac{1}{2}D^\prime \text{vec}\left(\Sigma^{-1}(\bm \sigma)-\Sigma^{-1}(\bm \sigma)\bm{\varepsilon}_t(\bm{\alpha}) \bm{\varepsilon}_t^\prime(\bm{\alpha}) \Sigma^{-1}(\bm \sigma)\right), 
\end{aligned}
\end{equation}
and the second derivatives of $l_t(\bm{\theta})$ are 
\begin{equation} \label{de_l_theta^2}
\begin{aligned}
	\frac{\partial^2 l_t(\bm \theta)}{\partial \bm{\alpha}\partial \bm{\alpha}^\prime}
	=& (\bm{\varepsilon}_t^\prime(\bm{\alpha})  \Sigma^{-1} (\bm \sigma)\otimes I_N)  \frac{\partial(\text{vec}(\partial \bm{\varepsilon}_t^\prime (\bm{\alpha}) /\partial \bm{\alpha}))}{\partial\bm{\alpha}^\prime} +\frac{\partial \bm{\varepsilon}_t^\prime(\bm{\alpha}) }{\partial \bm{\alpha}} \Sigma^{-1}(\bm \sigma)\frac{\partial\bm{\varepsilon}_t(\bm{\alpha})}{\partial \bm{\alpha}^\prime},\\
	\frac{\partial^2  l_t(\bm \theta)}{\partial \bm{\sigma} \partial \bm{\sigma}^\prime}
	=&-\frac{1}{2}D^\prime(\Sigma^{-1}(\bm \sigma)\otimes\Sigma^{-1}(\bm \sigma)-\Sigma^{-1}(\bm \sigma)\otimes \Sigma^{-1}(\bm \sigma)\bm{\varepsilon}_t(\bm{\alpha})\bm{\varepsilon}_t^\prime(\bm{\alpha})\Sigma^{-1}(\bm \sigma)\\
	&-\Sigma^{-1}(\bm \sigma)\bm{\varepsilon}_t(\bm{\alpha})\bm{\varepsilon}_t^\prime(\bm{\alpha})\Sigma^{-1}(\bm \sigma)\otimes\Sigma^{-1}(\bm \sigma))D,\\
	\frac{\partial^2  l_t(\bm \theta)}{\partial \bm{\alpha}\partial \bm{\sigma}^\prime}
	=&- \left(\bm{\varepsilon}_t^\prime (\bm{\alpha}) \Sigma^{-1}(\bm \sigma) \otimes \frac{\partial \bm{\varepsilon}_t^\prime(\bm{\alpha}) }{\partial \bm{\alpha}}\Sigma^{-1}(\bm \sigma)\right)D,
\end{aligned}
\end{equation}
where the first and second derivatives of $\bm{\varepsilon}_t(\bm \alpha)$ are defined in \eqref{de_epsilon_alpha} and \eqref{de_epsilon_alpha^2}. 
In addition, for the $n \times n$ matrices $\mathcal{I}=E[\partial l_{t}(\bm \theta_0)/ \partial \bm \theta\partial l_{t}(\bm \theta_0)/ \partial \bm \theta']$ and $\mathcal{J} = E(\partial ^2 l_t(\bm{\theta}_0)/\partial \bm{\theta}\partial \bm{\theta}^\prime)$. 
It can be shown that 
\begin{align} \label{I_2 and J_2}
	\mathcal{I} = E(P_{0t}^\prime S_1 P_{0t}) 
	\;\; \text{and} \;\;
	\mathcal{J} = E(P_{0t}^\prime S_2 P_{0t}), 
\end{align}
where the $ (N+N^2)\times n$ matrix $P_{0t} = P_t(\bm{\theta}_0)$, and the $(N+N^2)\times (N+N^2)$ matrices 
\begin{align*}
	S_1 = \left(
	\begin{matrix}
		\Sigma_0^{-1} & -\frac{1}{2} \Sigma_0^{-1} E(\bm{\varepsilon}_{0t} \bm{u}_{0t}^\prime)\mathcal{P}\\
		* & \frac{1}{4} \left[\mathcal{P} E(\bm{u}_{0t} \bm{u}_{0t}^\prime) \mathcal{P}- \text{vec}(\Sigma_0^{-1}) \text{vec}(\Sigma_0^{-1})^\prime \right]		
	\end{matrix}
	\right)
	\quad \text{and} \quad
	S_2 = \left(
	\begin{matrix}
		\Sigma_0^{-1} & 0_{N \times N^2}\\
		* & \frac{1}{2} \mathcal{P}
	\end{matrix}
	\right),
\end{align*}
with $*$ denoting the symmetric elements, $\mathcal{P} = \Sigma_0^{-1} \otimes \Sigma_0^{-1}$, 
$\bm{u}_0 = \text{vec}(\bm{\varepsilon}_{0t} \bm{\varepsilon}_{0t}^\prime)$ and 
\begin{align*}
	P_t(\bm{\theta}) = \left(
	\begin{matrix}
		\frac{\partial \bm \varepsilon_t (\bm \alpha)}{\partial \bm{\alpha}^\prime } & 0_{N \times N(N+1)/2  }\\
		0_{N^2 \times n_\alpha} & D
	\end{matrix}
	\right).
\end{align*}

\section{The Connection between the Scalable ARMA and VARMA Models}\label{Sec-relatedVARMA}
	We first consider the VARMA$(1,1)$ model 
	\begin{align} \label{VARMA(1,1)}
		\bm y_t= \Phi \bm y_{t-1}+ \bm \varepsilon_{t} - \Theta \bm \varepsilon_{t-1}, 
	\end{align}
	where $\Phi, \Theta \in \mathbb{R}^{N \times N}$. 
	Suppose that all eigenvalues of $\Theta$ are less than one in absolute value. 
	\eqref{VARMA(1,1)} can be written as the following VAR$(\infty)$ process: 
	\begin{align} \label{VAR(infinity)}
		\bm y_t = \sum_{h=1}^{\infty} A_h(\Phi, \Theta) \bm y_{t-h} +\bm \varepsilon_t
		\;\; \text{with} \;\; 
		A_h(\Phi, \Theta)=\Theta^{h-1}(\Phi-\Theta). 
	\end{align}
	Then we reparameterize $A_h$'s in \eqref{VAR(infinity)}. 
	Suppose that $\Theta$ has $r$ nonzero real eigenvalues, $\lambda_1, \ldots, \lambda_r$, and $s$ conjugate pairs of nonzero complex eigenvalues, $\left(\lambda_{r+2 j-1}, \lambda_{r+2 j}\right)=\left(\gamma_j e^{i \varphi_j}, \gamma_j e^{-i \varphi_j}\right)$ for $1 \leqslant j \leqslant s$ with $i$ being the imaginary unit. 
	$\Theta$ can be block-diagonalized as $\Theta =BJB^{-1}$ via the Jordan decomposition, where $B$ is an invertible matrix and $J = \diag\{\lambda_{1},\ldots,\lambda_{r},C_{1},\ldots,C_{s},0_{(N-r-2s) \times 1}^{\prime}\}$ with each $C_j$ being the $2\times 2 $ block determined by $\bm \eta_{j} = (\gamma_{j}, \varphi_{j})$ as follows
	\begin{align*}
		C_{j} = \gamma_{j}\left(
		\begin{matrix}
			\cos(\varphi_{j}) & \sin (\varphi_{j})\\
			-\sin (\varphi_{j}) & \cos(\varphi_{j})
		\end{matrix}
		\right).
	\end{align*}
	Substituting $\Theta =BJB^{-1}$ into \eqref{VAR(infinity)}, we obtain that 
	\begin{align*}
		\bm y_t = \sum_{h=1}^{\infty} BJ^{h-1}B^{-1} (\Phi-\Theta) \bm y_{t-h} +\bm \varepsilon_t. 
	\end{align*}
	Let $\underline{A} = B^{-1} (\Phi-\Theta)$ and $\underline{B} = B$. 
	It holds that 
	\begin{align*}
		\bm y_t = (\Phi-\Theta) \bm y_{t-1} + \sum_{h=2}^{\infty} \underline{B} J^{h-1} \underline{A} \bm y_{t-h} +\bm \varepsilon_t. 
	\end{align*}
	Furthermore, according to the block form of $J$, 
	let $\underline{A} = (\underline{A}_{1}, \ldots, \underline{A}_{r}, \underline{A}_{r+1}, \ldots, \\\underline{A}_{r+s}, \underline{A}_{r+s+1})^{\prime}$ and $\underline{B} = (\underline{B}_{1}, \ldots, \underline{B}_{r}, \underline{B}_{r+1}, \ldots, \underline{B}_{r+s}, \underline{B}_{r+s+1})$ with 
	$\underline{A}_{i} = \bm{a}_{i}$ and $\underline{B}_{i} = \bm{b}_{i}$ being $N \times 1$ matrices for $1 \leq i \leq r$, 
	$\underline{A}_{r+j} = (\bm{a}_{r+j}, \bar{\bm{a}}_{r+j})$ and $\underline{B}_{r+j} = (\bm{b}_{r+j}, \bar{\bm{b}}_{r+j})$ being $N \times 2$ matrices for $1 \leq j \leq s$, 
	and $\underline{A}_{r+s+1}$ and $\underline{B}_{r+s+1}$ being $N \times (N-r-2s)$ matrices. 
	It then follows that 
	\begin{align} \label{equivalent version}
			\bm y_t 
			&= (\Phi-\Theta) \bm y_{t-1}
			+ \sum_{h=2}^{\infty} \left[ 
			\sum_{i=1}^{r} \lambda_{i}^{h-1} \bm{b}_{i} \bm{a}_{i}^{\prime} 
			+ \sum_{j=1}^{s} \left(\bm{b}_{r+j}, \bar{\bm{b}}_{r+j}\right) C_{j}^{h-1} \left(\bm{a}_{r+j}, \bar{\bm{a}}_{r+j}\right)^{\prime}
			\right] \bm{y}_{t-h} + \bm \varepsilon_{t} \notag\\
			&= (\Phi-\Theta) \bm y_{t-1}
			+ \sum_{h=2}^{\infty} \Bigg\{ 
			\sum_{i=1}^{r} \lambda_{i}^{h-1} \bm{b}_{i} \bm{a}_{i}^{\prime} 
			+ \sum_{j=1}^{s} \gamma_{j}^{h-1} \Big[\cos((h-1) \varphi_{j}) \left(\bm{b}_{r+j} \bm{a}_{r+j}^{\prime} + \bar{\bm{b}}_{r+j} \bar{\bm{a}}_{r+j}^{\prime}\right) \notag\\
			&\mathrel{\phantom{= \underline{\bm{\omega}} + \sum_{h=0}^{\infty} \Bigg\{ \sum_{i=1}^{r} \lambda_{i}^{h-1} \bm{b}_{i} \bm{a}_{i}^{\prime} + \sum_{j=1}^{s} \gamma_{j}^{h-1} \Big[}}
			+ \sin((h-1) \varphi_{j}) \left(\bm{b}_{r+j} \bar{\bm{a}}_{r+j}^{\prime} - \bar{\bm{b}}_{r+j} \bm{a}_{r+j}^{\prime}\right)\Big]
			\Bigg\} \bm{y}_{t-h} + \bm \varepsilon_{t}. 
	\end{align}
	Let $G_{1} = \Phi-\Theta$, $G_{1+i} = \bm{b}_{i} \bm{a}_{i}^{\prime}$ for $1 \leq i \leq r$, and $G_{1+r+2j-1} = \bm{b}_{r+j} \bm{a}_{r+j}^{\prime} + \bar{\bm{b}}_{r+j} \bar{\bm{a}}_{r+j}^{\prime}$ and $G_{1+r+2j} = \bm{b}_{r+j} \bar{\bm{a}}_{r+j}^{\prime} - \bar{\bm{b}}_{r+j} \bm{a}_{r+j}^{\prime}$ for $1 \leq j \leq s$. 
	We can rewrite \eqref{equivalent version} in the form of scalable ARMA model in \eqref{model_VARinf}, that is, 
	\begin{align} \label{SARMA(1,1)}
		\bm y_t = \sum_{h=1}^{\infty} A_h(\bm \omega, \bm g) \bm{y}_{t-h} + \bm \varepsilon_{t}, 
	\end{align}
	where $A_1(\bm \omega, \bm g) = G_1$ and 
	$A_h(\bm \omega, \bm g) = \sum_{i=1}^{r} \lambda_{i}^{h-1} G_{1+i} 
			+ \sum_{j=1}^{s} \gamma_{j}^{h-1} [(\cos((h-1) \varphi_{j}) G_{1+r+2j-1} + \sin((h-1) \varphi_{j}) G_{1+r+2j})]$ for $h \geq 2$. 

	For the general VARMA$(p,q)$ model 
	\begin{align} \label{VARMA(p,q)}
		\bm y_t= \sum_{i=1}^{p} \Phi_{i} \bm y_{t-i}+ \bm \varepsilon_{t} - \sum_{j=1}^{q} \Theta_{j} \bm \varepsilon_{t-j}, 
	\end{align}
	where $\Phi_i$, $\Theta_j \in \mathbb{R}^{N \times N}$ for $1 \leq i \leq p$ and $1 \leq j \leq q$. 
	Suppose that \eqref{VARMA(p,q)} is invertible, then it can be written as the following VAR$(\infty)$ process: 
	$$
	\bm{y}_t=\sum_{h=1}^{\infty} \underbrace{\left(\sum_{i=0}^{p \wedge h} P \underline{\Theta}^{h-i} P^{\prime} \Phi_i\right)}_{A_h} \bm{y}_{t-h}+\bm{\varepsilon}_t 
	$$
	with
	$$
	\underline{\Theta}=
	\left(
	\begin{matrix}
	\Theta_1 & \Theta_2 & \cdots & \Theta_{q-1} & \Theta_q \\
	I_N & 0_{N \times N} & \cdots & 0_{N \times N} & 0_{N \times N} \\
	0_{N \times N} & I_N & \cdots & 0_{N \times N} & 0_{N \times N} \\
	\vdots & \vdots & & \vdots & \vdots \\
	0_{N \times N} & 0_{N \times N} & \cdots & I_N & 0_{N \times N}
	\end{matrix}\right),
	$$
	where $\Phi_0=-I_N$ and $P=\left(I_N, 0_{N \times N(q-1)}\right)$ are constant matrices. 
	Using the same techniques, we can also rewrite \eqref{VARMA(p,q)} in the form of scalable ARMA model in \eqref{model_VARinf}; see also the proofs of Proposition 2.2 of \cite{huang2025} and Proposition 1 of \cite{zheng2025interpretable}.

\newpage

\linespread{1.54}
\selectfont{}

\setlength{\bibsep}{1pt}
\bibliography{SARMA}

\end{document}